\documentclass[11pt]{article}
\usepackage{fullpage,amsfonts,amsmath,amssymb,bm, xspace}
\usepackage{multirow,graphicx, algorithm,algorithmic,rotating,caption}
\usepackage{url}
\usepackage{color}
\usepackage{cite}
\usepackage{float}
\usepackage[colorlinks = true]{hyperref}
\usepackage{amsthm}
\usepackage{caption}
\usepackage{subcaption}

\newtheorem{theorem}{Theorem}[section]
\newtheorem{lemma}[theorem]{Lemma}
\newtheorem{corollary}[theorem]{Corollary}
\newtheorem{proposition}[theorem]{Proposition}

\theoremstyle{remark}
\newtheorem{remark}[theorem]{Remark}

\numberwithin{equation}{section}

\newcommand{\A}{\mathcal{A}}
\newcommand{\C}{\mathbb{C}}
\newcommand{\R}{\mathbb{R}}
\newcommand{\E}{\mathbb{E}}
\renewcommand{\Pr}{\mathbb{P}}

\newcommand{\rank}{\operatorname*{rank}}

\newcommand{\trace}{\operatorname*{trace}}
\newcommand{\sign}{\operatorname*{sign}}
\newcommand{\toep}{\operatorname*{Toep}}

\newcommand{\minimize}{\operatorname*{minimize}}

\newcommand{\st}{\mbox{subject to}}
\newcommand{\nn}{\nonumber}
\newcommand{\eq}[1]{(\ref{eq:#1})}

\def \mathbbm {\mathbb}

\newcommand{\tmop}{\operatorname*}

\def \nn {\nonumber}
\title{Compressed Sensing Off the Grid}

\author{Gongguo Tang$^\dagger,$ Badri Narayan Bhaskar$^\dagger,$ Parikshit Shah$^\sharp$,
and Benjamin Recht$^\sharp$\\
$^\dagger$Department of Electrical and Computer Engineering\\
$^\sharp$Department of Computer Sciences\\
University of Wisconsin-Madison
}

\date{July 2012; Last Revised July 2013}

\begin{document}

\maketitle

\noindent
\begin{abstract}
This work investigates the problem of estimating the frequency components of a mixture of $s$ complex sinusoids from a random subset of $n$ regularly spaced samples. Unlike previous work in compressed sensing, the frequencies are not assumed to lie on a grid, but can assume any values in the normalized frequency domain $[0, 1]$. An atomic norm minimization approach is proposed to exactly recover the unobserved samples and identify the unknown frequencies, which is then reformulated as an exact semidefinite program. Even with this continuous dictionary, it is shown that $O(s\log s \log n)$ random samples are sufficient to guarantee exact frequency localization with high probability, provided the frequencies are well separated. Extensive numerical experiments are performed to illustrate the effectiveness of the proposed method. 
\end{abstract}


\section{Introduction}

Compressed sensing has demonstrated that data acquisition and compression can often be combined, dramatically reducing the time and space needed to acquire many signals of interest~\cite{Candes:2006eq,Candes:2008hb,Donoho:ci,Baraniuk:ix}.
Despite the tremendous impact of compressed sensing on signal processing theory and practice, its development thus far has focused on signals with sparse representations in finite discrete dictionaries. However, signals encountered in applications such as radar, array processing, communication, seismology, and remote sensing are usually specified by parameters in a \emph{continuous} domain~\cite{Stoica:2005wf,Ekanadham:2011tj,Parrish:2009ie}.  In order to apply the theory of compressed sensing to such applications, researchers typically adopt a discretization procedure to reduce the continuous parameter space to a finite set of grid points \cite{Malioutov:2005jw,Herman:ew, Fannjiang:2010dl, Baraniuk:2007gj,Bajwa:cm, Duarte:2011wy, Rauhut:2007ij, spice_likes, stoica_new}. While this simple strategy yields state-of-the-art performance for problems where the true parameters lie on the grid, discretization has several significant drawbacks --- \emph{(i)} In cases where the true parameters do not fall onto the finite grid,  the signal cannot often not be sparsely represented by the discrete dictionary \cite{Chi:2011ke, Herman:ek, Duarte:2011wy}. \emph{(ii)} It is difficult to characterize the performance of discretization using standard compressed sensing analyses since the dictionary becomes very coherent as we increase the number of grid points. \emph{(iii).} Although finer grids may improve the reconstruction error in theory, very fine grids often lead to numerical instability issues. 


We sidestep the issues arising from discretization by working directly on the continuous parameter space for estimating the continuous frequencies and amplitudes of a mixture of complex sinusoids from partially observed time samples. In particular, the frequencies are not assumed to lie on a grid, and can instead take arbitrary values across the bandwidth of the signal. With a time-frequency exchange, our model is exactly the same as the one in Cand\`es, Romberg, and Tao's foundational work on compressed sensing \cite{Candes:2006eq}, except that we do not assume the spikes to lie on an equispaced grid. This major difference presents a significant technical challenge as the resulting dictionary is no longer an orthonormal Fourier basis, but is an infinite dictionary with continuously many atoms and arbitrarily high correlation between candidate atoms.  We demonstrate that a sparse sum of complex sinusoids can be reconstructed exactly from a small sampling of its time samples provided the frequencies are sufficiently far apart from one another.

Our computational method and theoretical analysis is based upon the \emph{atomic norm} induced by samples of complex exponentials~\cite{Chandrasekaran:2010wp}.  Chandrasekaran \emph{et al} argue that the atomic norm is the best convex heuristic for underdetermined, structured linear inverse problems, and it generalizes the $\ell_1$ norm for sparse recovery and the nuclear norm for low-rank matrix completion. The norm is a convex function, and in the case of complex exponentials, can be computed via semidefinite programming.  We show how the atomic norm for moment sequences can be derived either from the perspective of sparse approximation or rank minimization~\cite{Recht:2010ht}, illuminating new ties between these related areas of study.  Much as was the case in other problems where the atomic norm has been studied, we prove that atomic norm minimization achieves nearly optimal recovery bounds for reconstructing sums of sinusoids from incomplete data.

To be precise, we consider signals whose spectra consist of spike trains with unknown locations in the normalized interval $[0,1]$, where we identify $0$ and $1$. Rather than sampling the signal at all times $t=0, \ldots, n-1$ we sample the signal at a subset of times $t_1, \ldots t_m$ with each $t_j \in \left\{0, \ldots, n-1 \right\}$. Our main contribution is summarized by the following theorem.
\begin{theorem}
  \label{thm:main_general} Suppose we observe the signal
\begin{align}\label{eqn:signalmodel}
x_j^\star = \sum_{k=1}^s c_k e^{i2\pi f_k j}, j = 0, \ldots, n-1
\end{align}
with unknown frequencies $\{f_1, \ldots, f_s\} \subset [0, 1]$ on an index set  $T \subset \{0, \ldots, n-1\}$ of size $m$ selected
uniformly at random.  Additionally, assume $\sign(c_k) := c_k/|c_k|$ are drawn
i.i.d.\ from the uniform distribution on the complex unit circle and 
\begin{align*}\label{eq:delta-definition}
  \Delta_f & = \min_{k \neq j} \left| f_k - f_j \right|\end{align*}
{where the distance $\left| f_k - f_j \right|$ is understood as the wrap-around
distance on the unit circle.}  If $\Delta_f \geq \frac{1}{\lfloor (n-1)/4  \rfloor}$, then there exists a numerical constant $C$ such that
  \begin{align*}
    m & \geq  C \max \left\{ \log^2 \frac{n}{\delta}, s \log \frac{s}{\delta}
    \log \frac{n}{\delta} \right\},
  \end{align*}
is sufficient to guarantee that we can recover $x^{\star}$ and localize the frequencies via a semidefinite program with probability at least $1 - \delta$.
\end{theorem}

The frequencies may be identified directly using the dual solution of the atomic norm minimization problem we propose in this paper. Alternatively, once the missing entries are recovered exactly, the frequencies can be identified by Prony's method \cite{deProny:tg}, a matrix pencil approach \cite{hua1990matrix}, or other linear prediction methods \cite{Stoica:1993cr}. After identifying the frequencies, the coefficients $\{c_k\}_{k=1}^s$ can be obtained by solving a linear system.  

\begin{remark}
{\bf{(Resolution)}} An interesting artifact of using convex optimization methods is the necessity of a particular resolution condition on the spectrum of the underlying signal. For the signal to be recoverable via our methods using $O(s\log s \log n)$ random time samples from the set $\left\{0, 1, \ldots, n-1 \right\}$, the spikes in the spectrum need to be separated by roughly $\frac{4}{n}$. In contrast, if one chose to acquire $O(s\log s \log n)$ \emph{consecutive} samples from this set (equispaced sampling), the required minimum separation would be  $\frac{4}{s \log s \log n}$; this sampling regime was studied by Cand\`es and Fernandez-Granda \cite{Candes:2012uf}. Therefore, in some sense, the resolution is determined by the region over which we take the samples, either in full or in a uniform random manner. We comment that numerical simulations of Section \ref{sec:experiments} suggest that the critical separation is actually $\frac{1}{n}$. We leave tightening our bounds by the extra constant of 4 to future work.
\end{remark}

\begin{remark}
{\bf (Random Signs)}
The randomness of the signs of the coefficients essentially assumes that the sinusoids have random phases.  Such a model is practical in many spectrum sensing applications as argued in \cite[Chapter 4.1]{Stoica:2005wf}.  Our proof will reveal that the phases can obey any symmetric distribution on the unit circle, not simply the uniform distribution.
\end{remark}

\begin{remark}
  {\bf{(Band-limited Signal Models)}} Note that any mixture of sinusoids with frequencies bandlimited to $[-W, W]$, after appropriate normalization, can be assumed to have frequencies in $[0, 1]$. Consequently, a bandlimited signal of such a form leads to samples of the form \eqref{eqn:signalmodel}.    More precisely, suppose the frequencies $\{ w_k \}$ lie in $[-W, W]$, and $x^\star \left( t \right)$ is a continuous signal of the
  form:
  \begin{align*}
    x^\star \left( t \right) & = \sum_{k = 1}^s c_k e^{i 2 \pi
    w_k t}\,.
  \end{align*}
By taking regularly spaced Nyquist samples at $t \in \left\{ {0}/{2W}, {1}/{2W}, \dots, (n-1)/{2W} \right\}$, we observe
  \begin{align*}
    x^\star_j & := x^\star \left( {j}/{2W} \right)  = \sum_{k = 1}^s c_k
    e^{i 2 \pi \frac{w_k}{2W} j}\nn\\
    & =  \sum_{k = 1}^s c_k e^{i 2 \pi f_k j} \tmop{with}
    f_k = \frac{w_k}{2W} \in \left[ - \frac{1}{2}, \frac{1}{2} \right],
  \end{align*}
  which is exactly the same as our model (\ref{eqn:signalmodel}) after a trivial translation of the frequency domain.  We emphasize that one does not need to actually acquire all of these Nyquist samples and then discard some of them to obtain the set $T$. In practice, one could pre-determine the set $T$ and sample only at locations specified by $T$. This procedure allows one to achieve compression at the sensing stage, the central innovation behind the theory of compressed sensing.
  \end{remark}

\begin{remark}{\bf (Basis Mismatch)}
Finally, we note that our result completely obviates the \emph{basis mismatch} conundrum~\cite{Chi:2011ke, Herman:ek, Duarte:2011wy} of discretization methods, where the frequencies might well fall off the grid. Since our continuous dictionary is globally coherent, Theorem \ref{thm:main_general} shows that the global coherence of the frame is not an obstacle to recovery. What matters more is the \emph{local} coherence between the atoms composing the true signal, as characterized by the separation between the frequencies. 
\end{remark}

This paper is organized as follows.   First, we specify our reconstruction algorithm as the solution to an atomic norm minimization problem in Section~\ref{sec:atomic-norm}. We show that this convex optimization problem can be exactly reformulated as a semidefinite program and that our methodology is thus computationally tractable.  We outline connections to prior art and the foundations that we build upon in Section~\ref{sec:related}.  We then proceed to develop the proofs in Section~\ref{sec:dualconstruction}. Our proof  requires the construction of an explicit certificate that satisfies certain interpolation conditions. The production of this  certificate requires us to consider certain random polynomial kernels, and derive concentration inequalities for these kernels that may be of independent interest to the reader.  In Section \ref{sec:experiments}, we validate our theory by extensive numerical experiments, confirming that random under-sampling as a means of compression coupled with atomic norm minimization as a means of recovery are a viable, superior alternative to discretization techniques.

\section{The Atomic Norm and Semidefinite Characterizations}\label{sec:atomic-norm}

Our signal model is a positive combination of complex sinusoids with arbitrary phases.  As motivated in~\cite{Chandrasekaran:2010wp}, a natural regularizer that encourages a sparse combination of such sinusoids is the~\emph{atomic norm} induced by these signals.  Precisely, define atoms $a(f,\phi) \in \C^{|J|}$, $f \in [0, 1]$ and $\phi \in [0,2\pi)$ as
\begin{equation*}\label{eqn:atom}
  [a\left(f, \phi\right)]_j  =
    e^{i(2{\pi}f j + \phi)},\, j \in J
\end{equation*}
and rewrite the signal model \eqref{eqn:signalmodel} in matrix-vector form
\begin{align}
x^\star = \sum_{k=1}^s |c_k| a(f_k,\phi_k) \label{eqn:signalmodel:mv}
\end{align}
where  $J$ is an index set with values being either $\{0, \ldots, n-1\}$ or $\{-2M, \ldots, 2M\}$ for some positive integer $n$ and $M$, and $\phi_k$ is the phase of the complex number $c_k$. In the rest of the paper, we use $\Omega = \{f_1, \ldots, f_s\} \subset [0, 1]$ to denote the unknown set of frequencies. In the representation \eqref{eqn:signalmodel:mv}, we could also choose to absorb the phase $\phi_k$ into the coefficient $|c_k|$ as we did in \eqref{eqn:signalmodel}. We will use both representations in the following and explicitly specify that the coefficient $c_k$ is positive when the phase term $\phi_k$ is in the atom $a(f_k, \phi_k)$.

The set of atoms $\mathcal{A} = \{a(f,\phi)~:~ f \in [0, 1], \phi \in [0, 2\pi) \}$ are building blocks of the signal $x^\star$, the same way that canonical basis vectors are building blocks for sparse signals, and unit-norm rank one matrices are building blocks for low-rank matrices. In sparsity recovery and matrix completion, the unit balls of the sparsity-enforcing norms, e.g., the $\ell_1$ norm and the nuclear norm, are exactly the convex hulls of their corresponding building blocks. In a similar spirit, we define an atomic norm $\|\cdot\|_\A$ by identifying its unit ball with the convex hull of $\A$
\begin{align*} 
  \left\| x \right\|_{\mathcal{A}} & = \inf \left\{ t > 0: x \in t \tmop{conv}
  \left( \mathcal{A} \right) \right\} 
   = \inf_{\substack{c_k \geq 0, \; \phi_k \in [0, 2\pi) \\ f_k \in [0, 1]}} \Big\{ \sum_k c_k: x = \sum_k c_k a(f_k, \phi_k)  
  \Big\} .
\end{align*}
Roughly speaking, the atomic norm $\|\cdot\|_\A$ can enforce sparsity in $\A$ because low-dimensional faces of $\operatorname*{conv}(\A)$ correspond to signals involving only a few atoms. The idea of using atomic norms to enforce sparsity for a general set of atoms was first proposed and analyzed in \cite{Chandrasekaran:2010wp}.

When the phases $\phi$ are all $0$, the set $\mathcal{A}_0 = \{a(f,0):  f \in [0, 1] \}$ is called the~\emph{moment curve} which traces out a one-dimensional variety in $\R^{2|J|}$.  It is well known that the convex hull of this curve is characterizable in terms of Linear Matrix Inequalities, and membership in the convex hull can thus be computed in polynomial time (see~\cite{Sturmfels11} for a proof of this result and a discussion of many other algebraic varieties whose convex hulls are characterized by semidefinite programming).  When the phases are allowed to range in $[0,2\pi)$, a similar semidefinite characterization holds.

\begin{proposition}\label{prop:sdp-char}
	For $x\in \C^{|J|}$ with $J = \{0, \ldots, n-1\}$ or $\{-2M, \ldots, 2M\}$,
	\begin{equation}
	\label{eq:atomic_sdp}
	\|x\|_{\mathcal{A}} = \inf\left\{ \tfrac{1}{2|J|} \trace(\toep(u)) +\tfrac{1}{2} t~:~  \begin{bmatrix}
\operatorname*{Toep}(u) & x\\
x^* & t
\end{bmatrix} \succeq 0 \right\}.
	\end{equation}
\end{proposition}

In the proposition, we used the superscript $^*$ to denote conjugate transpose and $\operatorname*{Toep}(u)$ to denote the Toeplitz matrix whose first column is equal to $u$.  The proof of this proposition relies on the following classical Vandermonde decomposition lemma for positive semidefinite Toeplitz matrices
\begin{lemma}[Caratheodory-Toeplitz, \cite{caratheodory1911variabilitatsbereich,caratheodory1911zusammenhang,toeplitz1911theorie}]\label{lm:vand} Any positive semidefinite Toeplitz matrix $P$ can be represented as
  follows
  \[
    P  =  V D V^*,
  \]
  where
  \[
  \begin{aligned}
    V & = \left[ {a} \left( f_1,0 \right) \cdots
    {a} \left( f_r,0 \right) \right]\,,\\
    D & = \tmop{diag} \left( \left[ d_1 \cdots d_r \right] \right)\,,
  \end{aligned}
  \]
 $d_k$ are real positive numbers, and $r = \rank(P)$.
\end{lemma}

The Vandermonde decomposition can be computed efficiently via root finding or by solving a generalized eigenvalue problem~\cite{hua1990matrix}.

\begin{proof}[Proof of Proposition~\ref{prop:sdp-char}]
We prove the case $J = \{0, \ldots, n-1\}$. The other case can be proved in a similar manner. Denote the value of the right hand side of \eq{atomic_sdp} by $\mathrm{SDP}(x)$.  Suppose $x = \sum_k c_k a(f_k,\phi_k)$ with $c_k>0$.  Defining $u = \sum_k c_k a(f_k,0)$ and $t = \sum_k c_k$, we note that 
\[
\toep(u) = \sum_{k}c_k a(f_k,0) a(f_k,0)^* = \sum_{k}c_k a(f_k,\phi_k) a(f_k,\phi_k)^*.
\] 
Therefore, 
\begin{align}\label{eqn:verifyfeasibility}
		\begin{bmatrix}
\operatorname*{Toep}(u) & x\\
x^* & t
\end{bmatrix}
= \sum_{k} c_k \begin{bmatrix} a(f_k,\phi_k) \\ 1 \end{bmatrix}\begin{bmatrix} a(f_k,\phi_k) \\ 1 \end{bmatrix}^* 		 \succeq 0
	\end{align}

Now, $\frac{1}{n}\trace(\toep(u)) = t = \sum_k c_k$ so that $SDP(x) \leq \sum_k c_k$. Since this holds for any decomposition of $x$, we conclude that $\|x\|_{\mathcal{A} }\geq \mathrm{SDP}(x)$.

Conversely,  suppose for some $u$ and $x$,
\begin{equation}\label{eq:toep-psd}
\begin{bmatrix}
\operatorname*{Toep}(u) & x\\
x^* & t
\end{bmatrix} \succeq 0\,.
\end{equation}
In particular, $\operatorname*{Toep}(u)\succeq 0$.  Form a Vandermonde decomposition
\[
	\operatorname*{Toep}(u)= V D V^*
\]
as promised by Lemma~\ref{lm:vand}.  Since $V D V^* = \sum_k d_k a(f_k,0) a(f_k,0)^*$ and $\|a(f_k,0)\|_2=\sqrt{n}$, we have $\frac{1}{n}\trace(\operatorname*{Toep}(u)) = \trace(D)$.

Using this Vandermonde decomposition and the matrix inequality~\eq{toep-psd}, it follows that $x$ is in the range of $V$, and hence
\[
	x = \sum_k w_k a(f_k,0) = Vw
\]
for some complex coefficient vector $w = [\cdots, w_k, \cdots]^T$.  Finally, by the Schur Complement Lemma, we have
\[
	V D V^* \succeq t^{-1} V w w^* V^*
\]
Let $q$ be any vector such that $V^*q = \sign(w)$.  Such a vector exists because $V$ is full rank.  Then
\[
	\trace(D)= q^* V D V^*q \succeq t^{-1}q^* V w w^* V^*q = t^{-1} \left(\sum_k |w_k|\right)^2.
\]
implying that $\trace(D) t \geq \left(\sum_k |w_k|\right)^2$.   By the arithmetic geometric mean inequality,
\[
	\tfrac{1}{2n} \trace(\operatorname*{Toep}(u)) + \tfrac{1}{2} t = 	\tfrac{1}{2} \trace(D) + \tfrac{1}{2} t \geq \sqrt{\trace(D) t} \geq \sum_k |w_k| \geq \|x\|_\A
\]
implying that $\mathrm{SDP}(x)\geq \|x\|_{\mathcal{A}}$ since the previous chain of inequalities hold for any choice of $u,t$ that are feasible.
\end{proof}

There are several other approaches to proving the semidefinite programming characterization of the atomic norm.  As we will see below, the dual norm of the atomic norm is related to the maximum modulus of trigonometric polynomials (see equation \eqref{eqn:dualatomicnorm}).  Thus, proofs based on Bochner's Theorem~\cite{Megretski03}, the bounded real lemma~\cite{Dumitrescu:2007vw, Bhaskar:2012tq}, or spectral factorization~\cite{roh2006discrete} would also provide a tight characterization. It is interesting that the SDP for the continuous case is not very different from an SDP derived for basis pursuit denoising on a grid in \cite{spice_likes, stoica_new, stoica_offgrid_heuristic}. {The work \cite{stoica_offgrid_heuristic} also provides a heuristic to deal with off-grid frequencies, though there is no theory to certify exact recovery.}

\subsection{Atomic Norm Minimization for Continuous Compressed Sensing}

Recall that we observe only a subset of entries $T \subset J$. As prescribed in~\cite{Chandrasekaran:2010wp}, a natural algorithm for estimating the missing samples of a sparse sum of complex exponentials is the atomic norm minimization problem
\begin{equation}  \label{eqn:minimization}
\begin{array}{ll}
\minimize_{x} & \left\| x \right\|_{\mathcal{A}} \\
  \text{subject to} & \ x_j = x^{\star}_j , j \in T
 \end{array}
\end{equation}
or, equivalently, the semidefinite program
\begin{align}
\minimize_{u,\ x,\ t} &\ \frac{1}{2|J|} \trace(\toep(u)) + \frac{1}{2} t \nn\\
\text{subject\ to\ } & \begin{bmatrix}
\operatorname*{Toep}(u) & x\\
x^* & t
\end{bmatrix} \succeq 0 \label{eqn:tracemin}\\
&\ x_j = x^\star_j, j \in T.\nn
\end{align}
The main result of this paper is that this semidefinite program almost always recovers the missing samples and identifies the frequencies provided the number of measurements is large enough and the frequencies are reasonably well-separated. We formalize this statement for the case $J = \{-2M, \ldots, 2M\}$ in the following theorem.

\begin{theorem}
  \label{thm:main} Suppose we observe the time samples of
\[
x_j^\star = \sum_{k=1}^s c_k e^{i2\pi f_k j}
\]
on the index set  $T \subset J = \{-2M, \ldots, 2M\}$ of size $m$ selected uniformly at random.  Additionally, assume $\sign(c_k)$ are drawn i.i.d. from a symmetric distribution on the complex unit circle. If $\Delta_f \geq \frac{1}{M}$, then there exists a numerical constant $C$ such that
\[
    m \geq  C \max \left\{ \log^2 \frac{M}{\delta}, s \log \frac{s}{\delta}
    \log \frac{M}{\delta} \right\},
\]
is sufficient to guarantee that with probability at least $1 - \delta$, $x^{\star}$ is the unique optimizer to (\ref{eqn:minimization}).\end{theorem}

We prove this theorem in Section~\ref{sec:dualconstruction}.  Note that Theorem~\ref{thm:main_general} is a corollary of Theorem~\ref{thm:main} via a simple reformulation.  We provide a proof of the equivalence in Appendix~\ref{apx:thm:main_general}.

\subsection{Duality and Frequency Localization}

To every norm, there is an associated dual norm, and the dual of the atomic norm for complex sinusoids has useful structure for both analysis and implementations.  In this section, we analyze the structure of the dual problem and show that the dual optimal solution can be used as a method to identify the frequencies $\{f_k\}$ that comprise the optimal $x^\star$.

Define the inner product as $\left<q, x\right> = x^*q$, and the real inner product as $\left<q, x\right>_\R = \operatorname*{Re}(\left<q, x\right>)$. Then the dual norm of $\|\cdot\|_\A$ is defined as
\begin{align}\label{eqn:dualatomicnorm}
  \left\| q \right\|_{\mathcal{A}}^{\ast}  = \sup_{\left\| x
  \right\|_{\mathcal{A}} \leq 1} \left\langle q, x \right\rangle_\R =
  \sup_{\phi \in [0, 2\pi), f \in [0, 1]} \langle q, e^{i\phi}a(f,0)
  \rangle_\R =  \sup_{f \in [0, 1]} | \langle q, a(f,0)
  \rangle |   
 \end{align}
 that is, the dual atomic norm is equal to the maximum modulus of the polynomial $Q(z) = \sum_{j\in J} q_j z^{-j}$ 
on the unit circle.   The dual problem of (\ref{eqn:minimization}) is thus
\begin{align}
  \operatorname*{maximize}_{q} &\ \ \left\langle q_T, x_T^{\star} \right\rangle_\R \nonumber \\
  \tmop{subject}\ \tmop{to} &\ \ \left\| q \right\|_{\mathcal{A}}^{\ast} \leq
  1\label{eqn:dual}
 \\
  &\ \ q_{T^c} = 0 \nonumber
\end{align}
which follows from a standard Lagrangian analysis~\cite{Chandrasekaran:2010wp}.   Note that the dual atomic norm problem is an optimization over polynomials with bounded modulus on the unit disk.

Let $\left( x, q \right)$ be primal-dual feasible to \eqref{eqn:minimization} and \eqref{eqn:dual}. We have that
  \begin{align*}
    \left\langle q, x \right\rangle_\R & = \left\langle q_T, x_T \right\rangle_\R, \text{\ since\ } q_{T^c} = 0.\\
    & = \left\langle q_T, x_T^\star \right\rangle_\R,  \text{\ since\ } x_T =
    x^{\star}_T\\
    & = \left\langle q, x^{\star} \right\rangle_\R  .
  \end{align*}
 Since the primal is only equality constrained, Slater's condition naturally
  holds, implying strong duality \cite[Section 5.2.3]{Boyd:2004uz}. By weak duality, we always have
  have
  \begin{align*}
    \left\langle q, x \right\rangle_\R  = \left<q, x^\star\right>_\R & \leq \left\| x \right\|_{\mathcal{A}}
  \end{align*}
  for any $x$ primal feasible and any $q$ dual feasible. Strong duality implies equality holds  \emph{if and only if} $q$ is dual optimal and $x$ is primal optimal. A straightforward consequence of strong duality is
a certificate of the support of the solution to \eqref{eqn:minimization}.

\begin{proposition}
  \label{pro:optimality}Suppose the atomic set $\A$ is composed of atoms defined by $[a(f,0)]_j = e^{i2\pi f j}, j \in J$ with $J$ being either $\{-2M, \cdots, 2M\}$ or $\{0, \cdots, n-1\}$. Then $\hat{x} = x^{\star}$ is the unique optimizer to
  (\ref{eqn:minimization}) if there exists a dual polynomial
  \begin{align}
    Q \left( f \right) & = \left<q, a(f,0)\right> =  \sum_{j \in J} q_j
    e^{- i 2 \pi jf}   \label{eqn:dualpoly}
  \end{align}
  satisfying
  \begin{align}
    Q \left( f_k \right) & = \tmop{sign} \left( c_k \right), \forall
    f_k \in \Omega  \label{eqn:condition:Q1}\\
    \left| Q \left( f \right) \right| & < 1, \forall f \notin \Omega 
    \label{eqn:condition:Q2}\\
    q_j & = 0, \forall j \notin T.  \label{eqn:condition:q1}
  \end{align}
\end{proposition}

The polynomial $Q \left( f \right)$ works as a dual certificate to certify
that $x^\star$ is the primal optimizer. The conditions on $Q
\left( f \right)$ are imposed on the values of the dual polynomial (condition
(\ref{eqn:condition:Q1}) and (\ref{eqn:condition:Q2})) and on the coefficient
vector $q$ (condition (\ref{eqn:condition:q1})).  To prove Theorem~\ref{thm:main}, we will construct a dual certificate satisfying the conditions of Proposition~\ref{pro:optimality} in Section \ref{sec:dualconstruction}.

\begin{proof}[Proof of Proposition \ref{pro:optimality}]
Any vector $q$ that satisfies the conditions of Proposition \ref{pro:optimality}
   is  dual feasible.  We also have that
  \begin{align*}
    \left\langle q, x^{\star} \right\rangle_\R & = \Big\langle q, \sum_{k = 1}^s c_k a
    \left( f_k, 0 \right) \Big\rangle_\R\\
    & = \sum_{k = 1}^s \tmop{Re} \left( c_k^{\ast} \left\langle q, a
    \left( f_k, 0\right) \right\rangle \right)\\
    & = \sum_{k = 1}^s \tmop{Re} \left( c_k^{\ast} \tmop{sign} \left( c_k
    \right) \right)\\
    & = \sum_{k = 1}^s \left| c_k \right|\\
    & \geq \left\| x^{\star} \right\|_{\mathcal{A}},
  \end{align*}
  where the last inequality is due to the definition of atomic norm. 
  On the other hand, H\"older's inequality states $\left<q, x^\star\right>_\R \leq \|q\|_\A^* \|x^\star\|_\A \leq \|x^\star\|_\A$, implying $\left<q, x^\star\right>_\R = \|x^\star\|_\A$.  
  Since $(x^\star, q)$ is primal-dual feasible,  we conclude that $x^{\star}$ is a primal optimal solution and $q$ is a dual optimal solution because of strong duality.
  
  For uniqueness, suppose $\hat{x} = \sum_k \hat{c}_k a ( \hat{f}_k, 0
  )$ with $\left\| \hat{x} \right\|_{\mathcal{A}} = \sum_k \left|
  \hat{c}_k \right|$ is another optimal solution. We then have for the dual
  certificate $q$:
  \begin{align*}
    \left\langle q, \hat{x} \right\rangle_\R & = \Big\langle q, \sum_k
    \hat{c}_k a ( \hat{f}_k, 0 ) \Big\rangle_\R\\
    & = \sum_{f_k \in \Omega} \tmop{Re} \left( \hat{c}_k^{\ast}
    \left\langle q, a \left( f_k, 0 \right) \right\rangle \right) + \sum_{\hat{f}_l \notin
    \Omega} \tmop{Re} \left( \hat{c}_l^{\ast} \left\langle q, a \left( \hat{f}_l, 0
    \right) \right\rangle \right)\\
    & < \sum_{f_k \in \Omega} \left| \hat{c}_k \right| + \sum_{\hat{f}_l \notin
    \Omega} \left| \hat{c}_l \right|\\
    & = \left\| \hat{x} \right\|_{\mathcal{A}}
  \end{align*}
  due to condition (\ref{eqn:condition:Q2}) \ if $\hat{x}$ is not solely
  supported on $\Omega$, contradicting strong duality. So all optimal solutions are supported on $\Omega$. \
  Since for both $J = \{-2M, \cdots, 2M\}$ and $\{0, \cdots, n-1\}$, the set of atoms with frequencies in $\Omega$ are linearly
  independent, the optimal solution is unique.
\end{proof}

An interesting and useful consequence of this proposition is that the dual solution $\hat{q}$ provides a way to determine the composing frequencies of $x^\star$.  One could evaluate the dual trigonometric polynomial
$\hat{Q}(f) := \left<\hat{q}, a(f, 0)\right>$ and localize the frequencies by identifying the locations where the polynomial $\hat{Q}(f)$ achieves modulus $1$. The evaluation of $\hat{Q}(f)$ can be performed efficiently using Fast Fourier Transform. Once the frequencies are estimated, the coefficients can be obtained by solving a linear system of equations. We illustrate frequency localization from dual polynomial in Figure \ref{fig:freqlocalization}.
\begin{figure}
\begin{center}
 \includegraphics[width=.5\textwidth, trim = 15mm 10mm 15mm 15mm, clip = true]
 {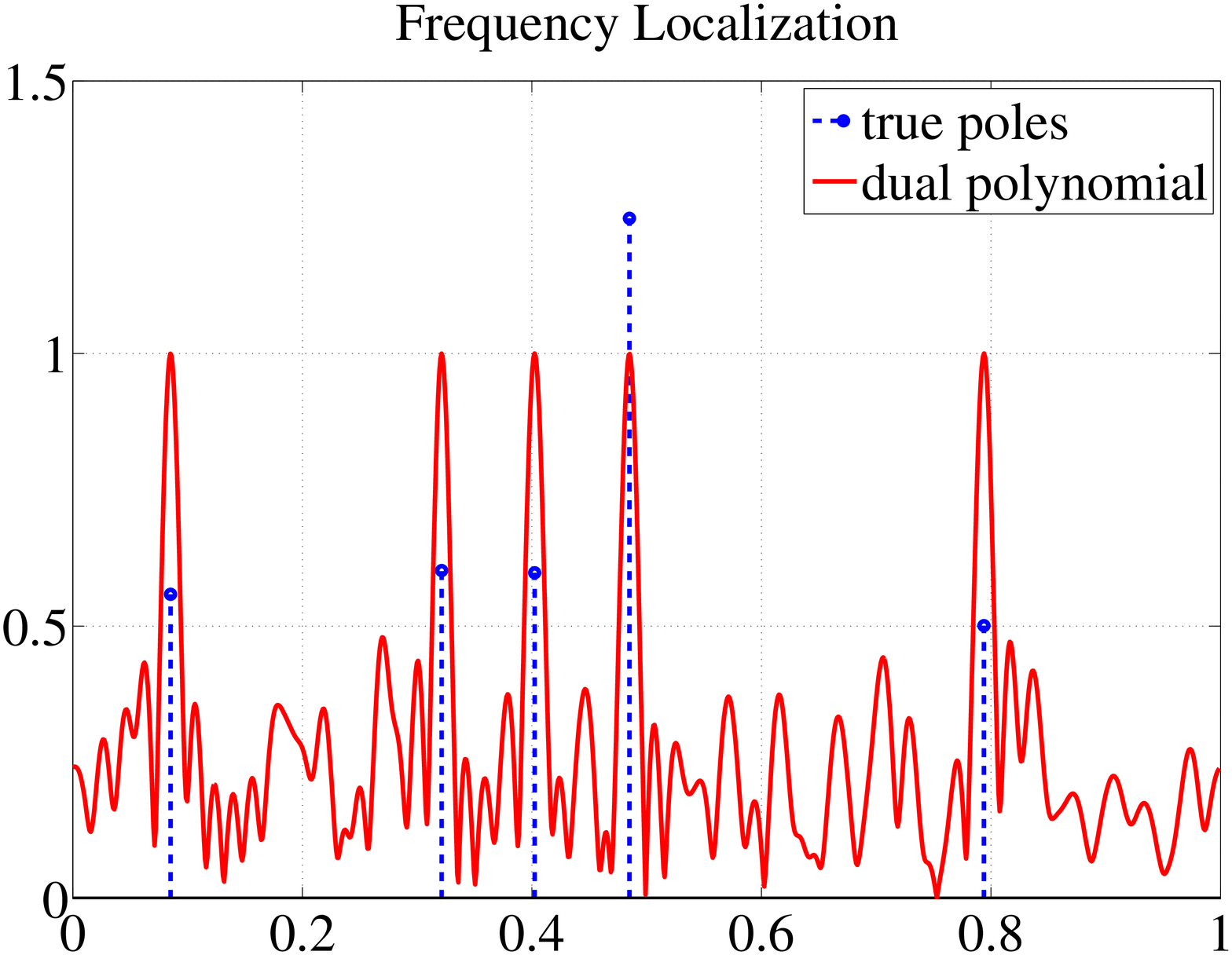} \caption{ \small {\bf Frequency localization from dual polynomial.} Original frequencies in the signal (dashed blue) with their heights representing the coefficient magnitudes and the dual polynomial modulus (solid red) obtained from a dual optimum. The recovered frequencies are obtained by identifying points where the dual polynomial has modulus one.} \label{fig:freqlocalization} 
 \end{center} 
\end{figure}

We would like to caution the reader that the dual optimal solutions are not unique in general. However, we can show that any dual optimal solution \emph{must contain the $s$ frequencies in $\Omega$ whenever the optimal primal solution is $x^\star$.} To see this, let $\hat{\Omega} = \{f: |\hat{Q}(f)| = 1\}$ be the set of recovered frequencies, and assume $\Omega /\hat{\Omega} \neq \emptyset$, then we have 
\begin{align*}
\left<\hat{q}, x^\star\right>_\R &= \langle \hat{q}, \sum_k
    {c}_k a ( {f}_k, 0 ) \rangle_\R\\
    & = \sum_{f_k \in \hat{\Omega}\bigcap\Omega} \tmop{Re} \left( {c}_k^{\ast}
   \hat{Q}(f_k) \right) + \sum_{f_l \in
    \Omega/\hat{\Omega}} \tmop{Re} \left({c}_l^{\ast} \hat{Q}(f_l)  \right)\\
    & < \sum_{f_k \in \hat{\Omega}\bigcap \Omega } \left| {c}_k \right| + \sum_{f_l \in
    \Omega/\hat{\Omega}} \left|{c}_l \right|\\
    & = \left\| x^\star \right\|_{\mathcal{A}},
\end{align*}
where get the strict inequality because $|\hat{Q}(f_l)| < 1$ for $f_l \in \Omega/\hat{\Omega}$, contradicting strong duality. Therefore, we have $\Omega \subset \hat{\Omega}$.  

In general, the set $\hat{\Omega}$ might contain spurious frequencies.  However, if we leverage the fact that the atomic norm is representable in terms of linear matrix inequalities, we can show that most semidefinite programming solvers will find good dual certificates.  To make this precise, let us study the dual semidefinite program of the atomic norm problem~\eqref{eqn:tracemin}. This dual problem also has a semidefinite programming formulation: 
\begin{align}
  \operatorname*{maximize}_{q, H} &\ \ \left\langle q_T, x_T^{\star} \right\rangle_\R\nonumber \\
  \tmop{subject}\ \tmop{to} &
 \ \ \left[\begin{array}{ll}
  H & -q\\
  -q^* & 1
  \end{array} \right] \succeq 0 \label{eqn:dual1}\\
  &\ \ \sum_{k=1}^{|J|-j} H_{k, k+j} = \left\{\begin{array}{ll}
  1, & j = 0,\\
  0, & j = 1, 2, \ldots, |J|-1.
  \end{array}\right.\label{eqn:dual2}\\
  & \ \ \ H \text{\ is Hermitian}\label{eqn:dual3}\\
  &\ \ \ q_{T^c} = 0 \nonumber
\end{align}
Here the three linear matrix inequalities \eqref{eqn:dual1}, \eqref{eqn:dual2}, and \eqref{eqn:dual3} are equivalent to the dual norm constraint $\|q\|_\A^* \leq 1$.  We emphasize that most solvers can directly return a dual optimal solution for free when solving the primal problem. So it is not necessary to solve the dual semidefinite program to obtain a dual optimum.

The following proposition addresses when $\hat{\Omega} = \Omega$. The proof is given in Appendix \ref{apx:pro:freqlocalization}.

\begin{proposition}[Exact Frequency Localization]
\label{pro:freqlocalization}
For signal $x = \sum_{k=1}^s c_k a(f_k, 0)$, denote by $\mathcal{D} = \{(q, H)\}$ the set of optimal solutions of the dual semidefinite program. If there exists one $(q, H) \in \mathcal{D}$ such that the dual polynomial $Q(f) = \left<{q}, a(f,0)\right>$ satisfies
\begin{align}
Q(f_k) &= \sign(c_k), k = 1, \ldots, s\label{eq:interpolation}\\
|Q(f)| &< 1, f \neq f_k, \forall k\label{eq:magnitude}
\end{align}
then the following statements hold:
\begin{enumerate}
\item all $(q, H)$ in the relative interior of $\mathcal{D}$ form strictly complementary pairs with the (unique) primal optimal solution, i.e., 
\begin{align}
\rank\left(\begin{bmatrix} H & -q\\
-q^* & 1\end{bmatrix} \right) + s = |J|+1;
\end{align}
\item all $(q, H)$ in the relative interior of $\mathcal{D}$ satisfy \eqref{eq:interpolation} and \eqref{eq:magnitude};
\item the dual central path converges to a point in the relative interior point of $\mathcal{D}$.
\end{enumerate}
\end{proposition}

Statement 2) of the proposition implies that we could use any optimal solution in the relative interior of the dual optimal set to localize frequencies, while statement 3) says any primal-dual path following algorithm for solving semidefinite programs (e.g., SDPT3) will produce such an interior point in the dual optimal set.

Under the conditions of Proposition \ref{pro:freqlocalization}, the relative interior of $\mathcal{D}$ excludes dual optimal solutions that contain spurious frequencies. A particular pathological case is when all coefficients ${c_k}$ in \eqref{eqn:signalmodel} are positive and $q = e_1$ (assume $J = \{0, \ldots, n-1\}$ and every index is observed), which is apparently a dual optimal solution. The dual polynomial corresponding to $e_1$ is constant $1$ and contains every frequency in $[0, 1]$. It is easy to show that the only $H$ such that $(H, e_1)$ satisfies \eqref{eqn:dual1}, \eqref{eqn:dual2}, and \eqref{eqn:dual3} is $H = e_1 e_1^*$. Hence, $(H, e_1)$ is not maximal complementary unless $x^\star$ contains at least $s = n$ frequencies, in which case the only $n-1$ degree trigonometric polynomial satisfying \eqref{eq:interpolation} is constant $1$.

However, this pathological case will not happen if we have a few well-separated frequencies and the phases are random\footnote{We believe the randomness requirement of phases are merely technical.}. Indeed, as we will see in Section \ref{sec:dualconstruction}, the way we prove Theorem \ref{thm:main} and Theorem \ref{thm:main_general} is to explicitly construct a dual polynomial satisfying \eqref{eq:interpolation} and \eqref{eq:magnitude}. Therefore, combining these theorems  and Proposition \ref{pro:freqlocalization}, we obtain the following corollary:
\begin{corollary}
Under the conditions of Theorem \ref{thm:main_general} (or Theorem \ref{thm:main}), with high probability, we could identify the frequencies in $x^\star$ from a dual optimal solution in the relative interior of the set of all dual optimal solutions.
\end{corollary}

\subsection{The Power of Rank Minimization}

The semidefinite programming characterization of the atomic norm also allows us to draw connections to the study of rank minimization~\cite{Candes:2008tg, Recht:2010ht, Gross:2009id, Recht:2011up}.  A direct way to exploit sparsity in the frequency domain is via minimization of the following ``$\ell_0$-norm" type quantity
\[
\|x\|_{\A, 0} = \min_{\substack{c_k \geq 0, \; \phi_k \in [0, 2\pi) \\ f_k \in [0, 1]}}\Big\{s: x = \sum_{k=1}^s c_k a(f_k, \phi_k)\Big\} \,.
\]
This penalty function chooses the \emph{sparsest} representation of a vector in terms of complex exponentials.  This combinatorial quantity is closely related to the rank of positive definite Toeplitz matrices as delineated by the following Proposition:

\begin{proposition}\label{pro:rank}
The quantity $\|x\|_{\A,0}$ is equal to the optimal value of the following rank minimization problem:
\begin{equation}\label{eqn:rankminimization}
\begin{array}{ll} \minimize_{u, t}& \  \rank(\toep(u))\\
\operatorname*{subject\ to\ }&  \begin{bmatrix}
\operatorname*{Toep}(u) & x\\
x^* & t
\end{bmatrix} \succeq 0.
\end{array}
\end{equation}
\end{proposition}

\begin{proof}
The case for $x = 0$ is trivial. For $x \neq 0$, denote by $r^\star$ the optimal value of \eqref{eqn:rankminimization}. We first show $r^\star \leq \|x\|_{\A, 0}$. Suppose $\|x\|_{\A, 0} = s < n$.  Assume the decomposition $x = \sum_{k=1}^s c_k a(f_k,\phi_k)$ with $c_k > 0$ achieves $\|x\|_{\A, 0}$, and set $u = \sum_k c_k a(f_k, 0)$ so that $\operatorname*{Toep}(u) = \sum_k c_k a(f_k, \phi_k) a(f_k, \phi_k)^* \succeq 0, t = \sum_k c_k > 0$.  Then, as we saw in \eqref{eqn:verifyfeasibility},
\[
\begin{bmatrix}\operatorname*{Toep}(u) & x\\
x^* & t
\end{bmatrix} 	=  \sum_{k=1}^s c_k
\begin{bmatrix} a(f_k,\phi_k) \\ 1
\end{bmatrix}
\begin{bmatrix} a(f_k,\phi_k) \\ 1
\end{bmatrix}^* \succeq 0\,.
\]
This implies that $r^\star = \rank(\operatorname*{Toep}(u)) \leq s$.
 
We next show $\|x\|_{\A, 0} \leq r^\star$. The $r^\star = n$ case is trivial as we could always expand $x$ on a Fourier basis, implying $\|x\|_{\A, 0} \leq n$. We focus on $r^\star < n$. Suppose $u$ is an optimal solution of \eqref{eqn:rankminimization}. Then if
\[
	\operatorname*{Toep}(u) = VDV^*
\]
is a Vandermonde decomposition, positive semidefiniteness implies that $x$ is in the range of $V$ which means that $x$ can be expressed as a combination of at most $r^\star$ atoms, completing the proof.
\end{proof}

Hence, for this particular set of atoms, atomic norm minimization is a trace relaxation of a rank minimization problem.   The trace relaxation has been proven to be a powerful relaxation for recovering low rank matrices subject to random linear equations~\cite{Recht:2010ht}, values at a specified set of entries~\cite{Candes:2008tg}, Euclidean distance constraints~\cite{javanmard2011localization}, and partial quantum expectation values~\cite{gross2010quantum}.  However, our sampling model is far more constrained and none of the existing theory applies to our problem.  Indeed, typical results on trace-norm minimization demand that the number of measurements exceeds the rank of the matrix times the number of rows in the matrix.  In our case, this would amount to $O(sn)$ measurements for an $s$ sparse signal.  We prove in the sequel that only $O(s \, \mathrm{polylog}(n))$ samples are required, dramatically reducing the dependence on $n$.

We close this section by noting that a~\emph{positive} combination of complex sinusoids with zero phases observed at the first $2s$ samples can be recovered via the trace relaxation with no limitation on the resolution~\cite{positive_fourier_donoho,positive_fuchs}. Why does the story change when we bring unknown phases into the picture?  A partial answer is provided by Figure~\ref{fig:moment-curve}.  Figure~\ref{fig:moment-curve} (a) and (b) display the set of atoms with no phase (i.e., $\{a(f,0)\}$) and phase either $0$ or $\pi$ respectively.  That is,  Figure~\ref{fig:moment-curve} (a) plots the set 
\begin{equation*}\label{eq:small-set}
\mathcal{A}_1 = \left\{ \begin{bmatrix} \cos(2\pi f) & \cos(4\pi f) & \cos(6\pi f) \end{bmatrix}~:~ f \in [0,1] \right\}\,,
\end{equation*}
while (b) displays the set
\begin{equation*}\label{eq:small-set-w-phase}
\mathcal{A}_2=\left\{ \begin{bmatrix} \cos(2\pi f+ \phi) & \cos(4\pi f+ \phi) & \cos(6\pi f+\phi) \end{bmatrix}~:~ f \in [0,1],\, \phi \in \{0,\pi\} \right\}\,.
\end{equation*}
Note that $\mathcal{A}_2$ is simply $\mathcal{A}_1 \cup -\mathcal{A}_1$.  Their convex hulls are displayed in Figure~\ref{fig:moment-curve} (c) and (d) respectively.  The convex hull of $\mathcal{A}_1$ is \emph{neighborly} in the sense that every edge between every pair of atoms is an exposed face and every atom is an extreme point.  On the other hand, the only secants between atoms in $\mathcal{A}_2$ that are faces of the convex hull of $\mathcal{A}_2$ are those between atoms with far apart phase angles and frequencies.  Problems only worsen if we let the phase range in $[0,2\pi)$. Thus, our intuition from positive moment curves does not extend to the compressed sensing problem of sinusoids with complex phases.  Nonetheless, we are able to demonstrate that under mild resolution assumptions, we can still recover sparse superpositions from very small sampling sets.

\begin{figure}
\begin{tabular}{cc}
 \includegraphics[width=.45\textwidth, trim = 15mm 75mm 15mm 75mm, clip = true]{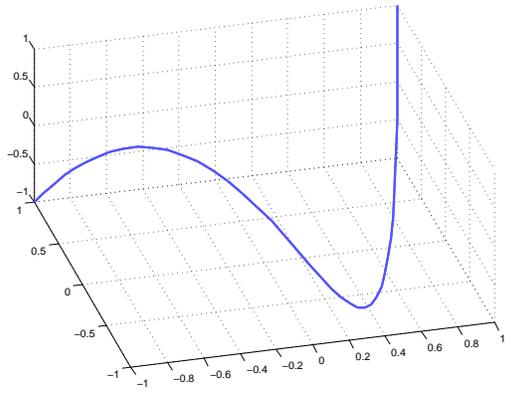} & 
 \includegraphics[width=.45\textwidth, trim = 15mm 75mm 15mm 75mm, clip = true]{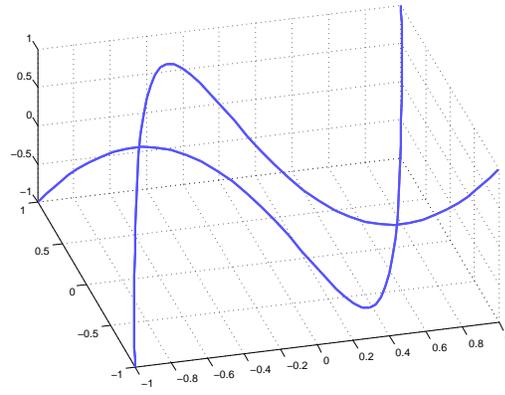}\\
 (a) & (b)\\
 \includegraphics[width=.45\textwidth, trim = 15mm 75mm 15mm 75mm, clip = true]{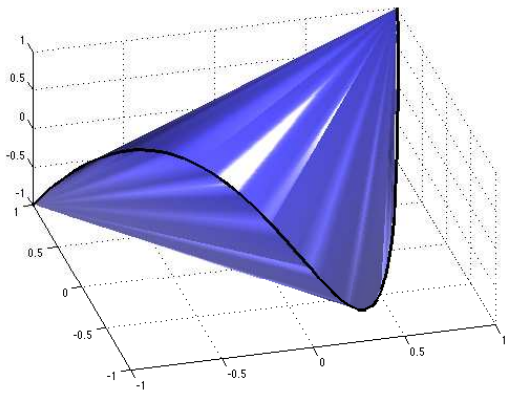} &
 \includegraphics[width=.45\textwidth, trim = 15mm 75mm 15mm 75mm, clip = true]{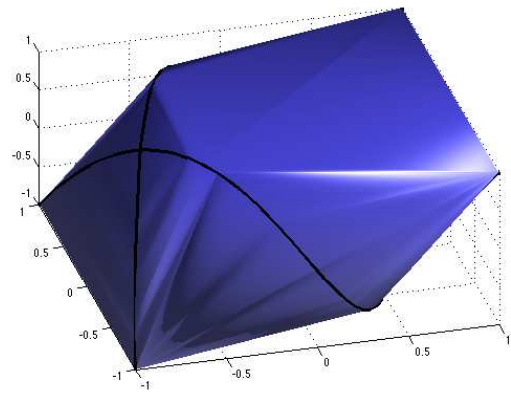}\\
 (c) & (d)
\end{tabular}
\caption{ \small {\bf Moments and their convex hulls.} (a) The real moment curve for the first three moments.  (b) The moment curve for the same frequencies, but adding in phase.  (c) The convex hull of (a).  (d) The convex hull of (b). Whereas all of the secants of (a) are extreme in their convex hull (c), many segments between atoms of (b) lie inside the convex hull (d).  }  \label{fig:moment-curve} 
\end{figure}

\section{Prior Art and Inspirations}\label{sec:related}

Frequency estimation is extensively studied and techniques for estimating sinusoidal frequencies from time samples date back to the work of Prony \cite{deProny:tg}. Many linear prediction algorithms based on Prony's method were proposed to estimate the frequencies from \emph{regularly spaced} time samples. A survey of these methods can be found in \cite{Blu:2008cd} and an extensive list of references is given in \cite{Stoica:1993cr}. With equispaced samples, these root-finding based procedures deal with the problem directly on the continuous frequency domain, and can recover frequencies provided the number of samples is at least twice of the number of frequencies, regardless of how closely these frequencies are located \cite{deProny:tg, Blu:2008cd, Stoica:1993cr, Stoica:2005wf}. 

In recent work \cite{Candes:2012uf}, Cand\`es and Fernandez-Granda studied this problem from the point of view of convex relaxations and proposed a total-variation norm minimization formulation that provably recovers the spectrum exactly. However, the convex relaxation requires the frequencies to be well separated by the inverse of the number of samples.  The proof techniques of this prior work form the foundation of analysis in the sequel, but many major modifications are required to extend their results to the compressed sensing regime.

In \cite{Bhaskar:2012tq}, the authors proposed using atomic norm to denoise a line spectral signal corrupted with Gaussian noise, and reformulated the resulting atomic norm minimization problem as a semidefinite program using the bounded real lemma \cite{Dumitrescu:2007vw}. Denoising is important to frequency estimation since the frequencies in a line spectral signal corrupted with moderate noise can be identified by linear prediction algorithms. Since the atomic norm framework in \cite{Bhaskar:2012tq} is essentially the same as the total-variation norm framework of \cite{Candes:2012uf}, the same semidefinite program can also be applied to total-variation norm minimization. 

What is common to all aforementioned approaches, including linear prediction methods, is the reliance on observing uniform or equispaced time samples. In sharp contrast, we show that nonuniform sampling is not only a viable option, and that the original spectrum can be recovered exactly in the continuous domain, but in fact is a means of \emph{compressive} or compressed sampling.  Indeed non-uniform sampling allows us to effectively sample the signal at a sub-Nyquist rate.  We point out that we have only handled the case of undersampling uniform samples as opposed to arbitrarily nonuniform samples in this paper. However, this is still of practical importance. For array signal processing applications, this corresponds to a reduction in the number of sensors required for exact recovery, since each sensor obtains one spatial sample  of the field. An extensive justification of the necessity of using randomly located sensor arrays can be found in \cite{Carin:eqa}. To the best of our knowledge, little is known about exact line-spectrum recovery with non-uniform sampling using \emph{parametric} methods, except sporadic work using $\ell_2$-norm minimization to recover the missing samples \cite{Dowski:1988ec}, or based on nonlinear least square data fitting\cite{Stoica:ir}. \emph{Nonparametric} methods such as Periodogram and Correlogram for nonuniform sampling have gained popularity in recent years \cite{Wang:2005vq, Shaghaghi:2012wa}, but their resolutions are usually low. 

An interesting feature related to using convex optimization based methods for estimation such as \cite{Candes:2012uf} is a particular resolvability condition: the separation between frequencies is required to be greater than $\frac{4}{n}$ where $n$ is the number of measurements. Linear prediction methods do not have a resolvability limitation, but it is known that in practice the numerical stability of root finding limits how close the frequencies can be. Theorem \ref{thm:main} can be viewed as an extension of the theory to nonuniform samples. Note that our approach gives an exact semidefinite characterization and is hence computationally tractable.  We believe our results have potential impact on two related areas: extending compressed sensing to continuous dictionaries, and extending line spectral estimation to nonuniform sampling, thus providing new insight in sub-Nyquist sampling and super-resolution.

\section{Proof of Theorem~\ref{thm:main}} \label{sec:dualconstruction}
The key to show that the optimization \eqref{eqn:minimization} succeeds is to construct a dual variable satisfying the conditions \eqref{eqn:condition:Q1}, \eqref{eqn:condition:Q1}, and \eqref{eqn:condition:q1} in Proposition \ref{pro:optimality} to certify the optimality of $x^\star$. The rest of the paper's proofs focus on the symmetric case $J = \{-2M, \ldots, 2M\}$.

As shown in Proposition \ref{pro:optimality}, the dual certificate can be interpreted as a polynomial with bounded modulus on the unit circle.  The polynomial is constrained to have most of its coefficients equal to zero.  In the case that all of the entries are observed, the polynomial constructed by Cand\`es and Fernandez-Granda~\cite{Candes:2012uf} suffices to guarantee optimality. Indeed they write the certificate polynomial via a kernel expansion and show that one can explicitly find appropriate kernel coefficients that certify optimality. We review this construction in Section~\ref{sec:detour}.  The requirements of the certificate polynomial in our case are far more stringent and require a non-trivial modification of their construction using a~\emph{random} kernel. This random kernel has nonzero coefficients only in the indices corresponding to observed locations (the randomness enters because the samples are observed at random).   The expected value of our random kernel is a multiple of the kernel developed in~\cite{Candes:2012uf}.  

Using a matrix Bernstein inequality, we show that we can find suitable coefficients to satisfy most of the optimality conditions.  We then write our solution in terms of the deterministic kernel plus a random perturbation.  The remainder of the proof is dedicated to showing that this random perturbation is small everywhere.  First, we show that the perturbation is small on a fine grid of the circle in Section~\ref{sec:perturbations}.  To do so, we emulate the proof of Cand\`{e}s and Romberg for reconstruction from incoherent bases\cite{Candes:2007es}.  Finally, in Section~\ref{sec:ext_cont}, we complete the proof by estimating the Lipschitz constant of the random polynomial, and, in turn, proving that the perturbations are small everywhere.  Our proof is based on Bernstein's polynomial inequality which was used to estimate the noise performance of atomic norm denoising by Bhaskar \emph{et al}~\cite{Bhaskar:2012tq}.

\subsection{A Detour: When All Entries are Observed}\label{sec:detour}
Before we consider the random observation model, we explain how to construct a dual polynomial when all entries in $J = \{-2M, \ldots, 2M\}$ are
observed, i.e., $T = J$. The kernel-based construction method, which was first proposed in \cite{Candes:2012uf}, inspires our random kernel based construction in Section
\ref{sec:randompoly}. The results presented in this subsection are also necessary for our later
proofs.

When all entries are observed, the optimization problem
(\ref{eqn:minimization}) has a trivial solution, but we can still apply duality to certify the optimality of a particular decomposition. Indeed, a dual
polynomial satisfying the conditions given in Proposition
\ref{pro:optimality} with $T^c = \emptyset$ means that $\left\| x^{\star}
\right\|_{\mathcal{A}} = \sum_k \left| c_k \right|$, namely, the decomposition
$x^{\star} = \sum_k c_k a \left( f_k,0 \right)$ achieves the atomic norm. To
construct such a dual polynomial, Cand\`es and  Fernandez-Granda suggested considering a polynomial
$\bar{Q}$ of the following form~\cite{Candes:2012uf} :
\begin{align}
  \bar{Q} \left( f \right) & = \sum_{k=1}^s \alpha_k 
  \bar{K}_M  \left( f - f_k \right) + \sum_{k=1}^s \beta_k 
  \bar{K}_M'  \left( f - f_k \right) .  \label{eqn:dualpolydeterministic}
\end{align}
Here $\bar{K}_M \left( f \right)$ is the squared Fej\'er kernel
\begin{align}
  \bar{K}_M (f) & = \left[ \frac{\sin (\pi M f)}{M \sin (\pi f)}
  \right]^4  = \frac{1}{M}  \sum_{j = - 2 M}^{2 M} g_M \left(j \right)
  e^{- i 2 \pi f j}\label{eqn:FejerKernelExpansion}
\end{align}
with $g_M \left(j\right) = \frac{1}{M}  \sum_{k = \max \left( j -
M, - M \right)}^{\min \left( j + M, M \right)} \left( 1 - \left| \frac{k}{M}
\right| \right)  \left( 1 - \left| \frac{j}{M} - \frac{k}{M} \right| \right)$
the discrete convolution of two triangular functions. The squared Fej\'er kernel is a good candidate kernel because it attains the value of $1$ at its peak, and rapidly decays to zero. Provided a separation condition is satisfied by the original signal, a suitable set of coefficients $\alpha$ and $\beta$ can always be found.

We use $\bar{K}_M',
\bar{K}_M'', \bar{K}_M'''$ to denote the first three derivatives of
$\bar{K}_M$.  We list some useful facts about the kernel $\bar{K}_M(f)$:
\begin{align*}
\nn  \bar{K}_M \left( 0 \right) & = 1\\
\nn  \bar{K}_M' \left( 0 \right) = \bar{K}_M''' \left( 0 \right) & = 0\\
 \nn \bar{K}_M'' \left( 0 \right) & = - \frac{4 \pi^2  \left( M^2 - 1
  \right)}{3}\,.
  \end{align*}
For the weighting function $g_M \left( \cdot \right)$, we have
\begin{align}\label{eqn:gbd}
  \left\| g_M \right\|_{\infty} & = \sup_j \left| g_M \left( j \right)
  \right| \leq 1.
\end{align}

We require that the dual polynomial \eqref{eqn:dualpolydeterministic} satisfies
\begin{align}\label{eqn:Q1:expected}
  \bar{Q} \left( f_j \right) & = \sum_{k = 1}^s \alpha_k \bar{K}_M
  \left( f_j - f_k \right) + \sum_{k = 1}^s \beta_k \bar{K}_M' \left( f_j -
  f_k \right) = \tmop{sign} \left( c_j \right),
\end{align}
\begin{align}\label{eqn:Q2:expected}
  \bar{Q}' \left( f_j \right) & = \sum_{k = 1}^s \alpha_k \bar{K}_M'
  \left( f_j - f_k \right) + \sum_{k = 1}^s \beta_k \bar{K}_M'' \left(
  f_j - f_k \right) = 0,
\end{align}
for all $f_j \in \Omega$.  The constraint (\ref{eqn:Q1:expected}) guarantees that $Q \left( f \right)$
satisfies the interpolation condition (\ref{eqn:condition:Q1}), and the
constraint (\ref{eqn:Q2:expected}) is used to ensure that $|Q \left( f \right)|$
achieves its maximum at frequencies in $\Omega$.  Note that the condition \eqref{eqn:condition:q1} is absent in this section's setting since the set $T^c$ is empty.

We rewrite these linear constraints in the matrix vector form:
\begin{align*}
  \left[\begin{array}{cc}
    \bar{D}_0 & \frac{1}{\sqrt{\left| \bar{K}_M'' \left( 0 \right) \right|}}
    \bar{D}_1\\
    - \frac{1}{\sqrt{\left| \bar{K}_M'' \left( 0 \right) \right|}} \bar{D}_1 & -
    \frac{1}{\left| \bar{K}_M'' \left( 0 \right) \right|} \bar{D}_2
  \end{array}\right] \left[\begin{array}{c}
    \alpha\\
    \sqrt{\left| \bar{K}_M'' \left( 0 \right) \right|} \beta
  \end{array}\right] & = \left[\begin{array}{c}
    u\\
    0
  \end{array}\right]
\end{align*}
where $\left[ \bar{D}_0 \right]_{j k} = \bar{K}_M \left( f_j - f_k \right)$,
$\left[ \bar{D}_1 \right]_{j k} = \bar{K}_M' \left( f_j - f_k \right)$,
$\left[ \bar{D}_2 \right]_{j k} = \bar{K}_M'' \left( f_j - f_k \right)$ and $u
\in \mathbbm{C}^{s}$ is the vector with $u_j = \tmop{sign}
\left( c_j \right)$. We have rescaled the system of linear equations such that the system
matrix is symmetric, positive semidefinite, and very close to identity.  Positive definiteness follows because the system matrix is a positive combination of outer products.  To get an idea of
why the system matrix is near the identity, observe that $\bar{D}_0$ is symmetric with diagonals one,
$\bar{D}_1$ is antisymmetric, and $\bar{D}_2$ is symmetric with negative
diagonals $\bar{K}_M'' \left( 0 \right)$. We define
\begin{align}
  \bar{D} & = \left[\begin{array}{cc}
    \bar{D}_0 & \frac{1}{\sqrt{\left| \bar{K}_M'' \left( 0 \right) \right|}}
    \bar{D}_1\\
    - \frac{1}{\sqrt{\left| \bar{K}_M'' \left( 0 \right) \right|}} \bar{D}_1 & -
    \frac{1}{\left| \bar{K}_M'' \left( 0 \right) \right|} \bar{D}_2
  \end{array}\right] = \left[\begin{array}{cc}
    \bar{D}_0 & \frac{1}{\sqrt{\left| \bar{K}_M'' \left( 0 \right) \right|}}
    \bar{D}_1\\
    \frac{1}{\sqrt{\left| \bar{K}_M'' \left( 0 \right) \right|}}
    \bar{D}_1^{\ast} & - \frac{1}{\left| \bar{K}_M'' \left[ 0 \right) \right|}
    \bar{D}_2
  \end{array}\right]\label{eqn:defbarD}
\end{align}
and summarize properties of the system matrix $\bar{D}$ and its submatrices in the
following proposition, whose proof is given in Appendix \ref{apx:sysmtx:deterministic}.

\begin{proposition}\label{pro:sysmtx:deterministic}
  Suppose $\Delta_f \geq \Delta_{\min} = \frac{1}{M}$. Then $\bar{D}$ is invertible and
  \begin{align}
  \|I-\bar{D}\| & \leq 0.3623,\label{eqn:I_Dbarbd}\\
  \|\bar{D}\| & \leq 1.3623,\label{eqn:Dbarbd}\\
  \|\bar{D}^{-1}\| & \leq 1.568\label{eqn:Dbaribd}.
 \end{align}
 Here $\|\cdot\|$ denotes the matrix operator norm.
\end{proposition}

For notational simplicity, partition the inverse of $\bar{D}$ as
\begin{align*}
\bar{D}^{- 1} = \left[\begin{array}{cc}
\bar{L} & \bar{R}
  \end{array}\right]
\end{align*}
  where $\bar{L}$ and $\bar{R}$ are both $2 s \times s$.  
Then, solving for $\alpha$ and $\sqrt{\left| \bar{K}_M''
\left( 0 \right) \right|} \beta$ yields
\begin{align}\label{eqn:coeffs:deterministic}
  \left[\begin{array}{c}
    \alpha\\
    \sqrt{\left| \bar{K}_M'' \left( 0 \right) \right|} \beta
  \end{array}\right] & = \bar{D}^{- 1} \left[ \begin{array}{c}
    u\\
    0
  \end{array}\right]
   = \bar{L} u.
\end{align}
Then the $\ell$th derivative of the dual polynomial (after normalization) is
\begin{align}\label{eqn:Qfasvlbar}
  \frac{1}{\sqrt{|\bar{K}_M''(0)|}^\ell}\bar{Q}^{(\ell)} \left( f \right) & = \sum_{k = 1}^s \alpha_k \frac{1}{\sqrt{|\bar{K}_M''(0)|}^{\ell}}\bar{K}_M^{(\ell)}
  \left( f - f_k \right) \nn \\
& \qquad\qquad  + \sum_{k = 1}^s \sqrt{\left| \bar{K}_M'' \left(
  0 \right) \right|} \beta_k \frac{1}{\sqrt{\left| \bar{K}_M'' \left( 0 \right)
  \right|}^{\ell+1}} \bar{K}_M^{(\ell+1)} \left( f - f_k \right)\nn\\
  & = \bar{v}_\ell \left( f \right)^{\ast} \bar{L} u = \left<\bar{L}u, \bar{v}_\ell(f)\right>.
\end{align}
where we have defined
\begin{align}\label{eqn:vlbar}
  \bar{v}_{\ell} \left( f \right) & = \frac{1}{\sqrt{\left| \bar{K}_M''
  \left( 0 \right) \right|}^\ell} \left[\begin{array}{c}
    \bar{K}_M^{\left( \ell \right)} \left( f - f_1 \right)^{\ast}\\
    \vdots\\
   \bar{K}_M^{\left( \ell \right)} \left( f - f_s \right)^{\ast}\\
      \frac{1}{\sqrt{\left| \bar{K}_M''
  \left( 0 \right) \right|}}\bar{K}_M^{(\ell+1)} \left( f - f_1 \right)^{\ast}\\
  \vdots\\
        \frac{1}{\sqrt{\left| \bar{K}_M''
  \left( 0 \right) \right|}}\bar{K}_M^{(\ell+1)} \left( f - f_s \right)^{\ast}
    \end{array}\right]
\end{align}
with $\bar{K}_M^{\left( \ell \right)}$ the $\ell$th derivative of $\bar{K}_M$.

To certify that the polynomial with the coefficients~\eqref{eqn:coeffs:deterministic} are bounded uniformly on the unit circle, Cand\`es and Fernandez-Granda divide the domain $[0,1]$ into regions near to and far from the frequencies of $x^\star$. Define 
\begin{align*}
  \Omega_{\tmop{near}} & = \bigcup_{k = 1}^s \left[ f_k - f_{b, 1}, f_k + f_{b, 1} \right]\\
\Omega_{\tmop{far}} & = \left[ 0, 1 \right] /\Omega_{\mathrm{near}}
\end{align*}
with $f_{b, 1} = 8.245\times 10^{- 2} \frac{1}{M}$.  On $\Omega_{\mathrm{far}}$, $|Q(f)|$ was analyzed directly, while on $\Omega_{\mathrm{near}}$ $|Q(f)|$ is bounded by showing that its second order derivative is negative.  The following results are derived in the proofs of Lemmas 2.3 and 2.4 in \cite{Candes:2012uf}:

\begin{proposition}
  \label{pro:expected}Assume
  $\Delta_f \geq \Delta_{\min} = \frac{1}{M}$. Then we have
  \begin{align}\label{eqn:defaraway}
    \left| \bar{Q} \left( f \right) \right| & < 0.99992, \tmop{for} f \in
    \Omega_{\tmop{far}}
  \end{align}
  and for $f \in \Omega_{\tmop{near}}$
  \begin{align}
    \bar{Q}_R \left( f \right) & \geq 0.9182 \label{eqn:near1}\\
    \left| \bar{Q}_I \left( f \right) \right| & \leq 3.611 10^{- 2}\label{eqn:near2}\\
    \frac{1}{\left| \bar{K}_M'' \left( 0 \right) \right|} \bar{Q}''_R \left( f
    \right) & \leq - 0.314\label{eqn:near3}\\
    \left| \frac{1}{\left| \bar{K}_M'' \left( 0 \right) \right|} \bar{Q}_I''
    \left( f \right) \right| & \leq 0.5755\label{eqn:near4}\\
    \left| \frac{1}{\sqrt{\left| \bar{K}_M'' \left( 0 \right) \right|}} \bar{Q}'
    \left( f \right) \right| & \leq 0.4346\label{eqn:near5}.
  \end{align}
  and as a consequence,
  \begin{align*}\label{eqn:neartotal}
    \frac{1}{\left| \bar{K}_M'' \left( 0 \right) \right|} \left( \bar{Q}_R
    \left( f \right) \bar{Q}_R'' \left( f \right) + \left| \bar{Q}' \left( f
    \right) \right|^2 + \left| \bar{Q}_I \left( f \right) \right| \left|
    \bar{Q}_I'' \left( f \right) \right| \right) & \leq - 7.865 10^{- 2} .
  \end{align*}
\end{proposition}

\subsection{Bernoulli Observation Model}
The uniform sampling model is difficult to analyze directly. However, the same argument used in \cite{Candes:2006eq} shows that the probability of recovery failure under the uniform model is at most twice of that under a Bernoulli model. Here by ``recovery failure", we refer to that \eqref{eqn:minimization} would not recover the original signal $x^\star$. Therefore, without loss of generality, we focus on the following Bernoulli observation model in our proof.

We observe entries in $J$ independently with probability $p$. Let $\delta_j = 1$ or $0$ indicate whether we observe the $j$th entry. Then $\{\delta_j\}_{j \in J}$ are i.i.d. Bernoulli random variables such that
\begin{align*}
\mathbbm{P} \left(
\delta_j = 1 \right) = p.
\end{align*}
On average in this model, we will observe $p|J|$ entries. For $J = \{-2M, \ldots, 2M\}$, we use 
\[
p = \frac{m}{M} < 1.
\]

\subsection{Random Polynomial Kernels}\label{sec:randompoly}

We now turn to designing a dual certificate for the Bernoulli observation model.  As for the case that all entries
are observed, the challenge is to construct a dual polynomial satisfying
\begin{align*}
  Q \left( f_k \right) & = \tmop{sign} \left( c_k \right), \forall
  f_k \in \Omega\\
  \left| Q \left( f \right) \right| & <  1, \forall f \notin \Omega,
\end{align*}
as well as an additional constraint
\begin{equation}\label{eqn:qtcconstraint}
  q_j = 0, \forall j \in T^c . 
\end{equation}
The main difference in our random setting is that the demands on the polynomial $Q(f)$ are much stricter as manifested by \eqref{eqn:qtcconstraint}, namely, we require that \emph{most} of the coefficients of the polynomial are equal to zero. Our approach mimics the construction in the deterministic case to write
\begin{align}
  Q \left( f \right) & = \sum_{k=1}^s \alpha_k K_M \left( f - f_k
  \right) + \sum_{k=1}^s\beta_k K_M' \left( f - f_k \right),
  \label{eqn:formofpoly}
\end{align}
but using a~\emph{random kernel} $K_M(\cdot)$, which has nonzero coefficients only on the random subset $T$ and satisfies $\mathbb{E}K_M =p \bar{K}_M$.  We will then prove that $K_M$ concentrates tightly around $p\bar{K}_M$.

Our random kernel is simply the expansion~\eqref{eqn:FejerKernelExpansion}, but with each term multiplied by a Bernoulli random variable corresponding to the observation of a component:
\begin{align*}
  K_M \left( f \right) & := \frac{1}{M}  \sum_{j \in T} g_M \left(
  j \right) e^{- i 2 \pi f j}\\
  & = \frac{1}{M}  \sum_{j = - 2 M}^{2 M} \delta_j g_M \left( j
  \right) e^{- i 2 \pi f j}.
\end{align*}
As before
\begin{align*}
g_M \left( j \right) = \frac{1}{M} \sum_{k = \max
\left( j - M, - M \right)}^{\min \left( j + M, M \right)} \left( 1 - \left|
\frac{k}{M} \right| \right) \left( 1 - \left| \frac{j}{M} - \frac{k}{M}
\right| \right)
\end{align*}
is the convolution of two discrete triangular functions. The $\ell$th derivative of $K_M(f)$ is
\begin{align*}
  K_M^{\left( \ell \right)} \left( f \right) 
  & = \frac{1}{M} \sum_{j = - 2 M}^{2 M} \left( - i 2 \pi j \right)^{\ell}
  g_M \left( j \right) \delta_j e^{- i 2 \pi f j}.
\end{align*}

Both $K_M \left( f - f_k \right)$ and $K_M' \left( f - f_k
\right)$ are random trigonometric polynomials of degree at most $2 M$. More importantly, they contain
monomial $e^{- i 2 \pi f j}$ only if $\delta_j = 1$, or equivalently, $j\in T$. Hence $Q \left( f
\right)$ in \eqref{eqn:formofpoly} is of the form (\ref{eqn:dualpoly}) and satisfies $q_j = 0, j \in
T^c$. It is easy to calculate the expected values of $K_M(f)$ and its $\ell$th derivatives:
\begin{align}
  \mathbbm{E}K_M^{\left( \ell \right)} \left( f \right) & = \frac{1}{M}
  \sum_{j = - 2 M}^{2 M} \left( - i 2 \pi j \right)^{\ell} g_M \left(
  j \right) \mathbbm{E} \{\delta_j\} e^{- i 2 \pi f j}\nn\\
  & = p \frac{1}{M} \sum_{j = - 2 M}^{2 M} \left( - i 2 \pi j
  \right)^{\ell} g_M \left( j \right) e^{- i 2 \pi f j}\nn\\
  & = p \bar{K}_M^{\left( \ell \right)} \left( f \right).\label{eqn:expectationofkernel}
\end{align}

In Figure \ref{fig:random_kernel}, we plot $p^{-1}|K_M(f)|$ and $p^{-1}|K_M'(f)|$ laid over $|\bar{K}_M(f)|$ and $|\bar{K}'_M(f)|$, respectively. We see that far away from the peak, the random coefficients induce bounded oscillations to the kernel.  Near $0$, however, the random kernel remains sharply peaked.
\begin{figure}[ht]
\begin{subfigure}[b]{0.5\textwidth}
\centering
 \includegraphics[width=\textwidth, trim = 10mm 0mm 0mm 10mm, clip = true]{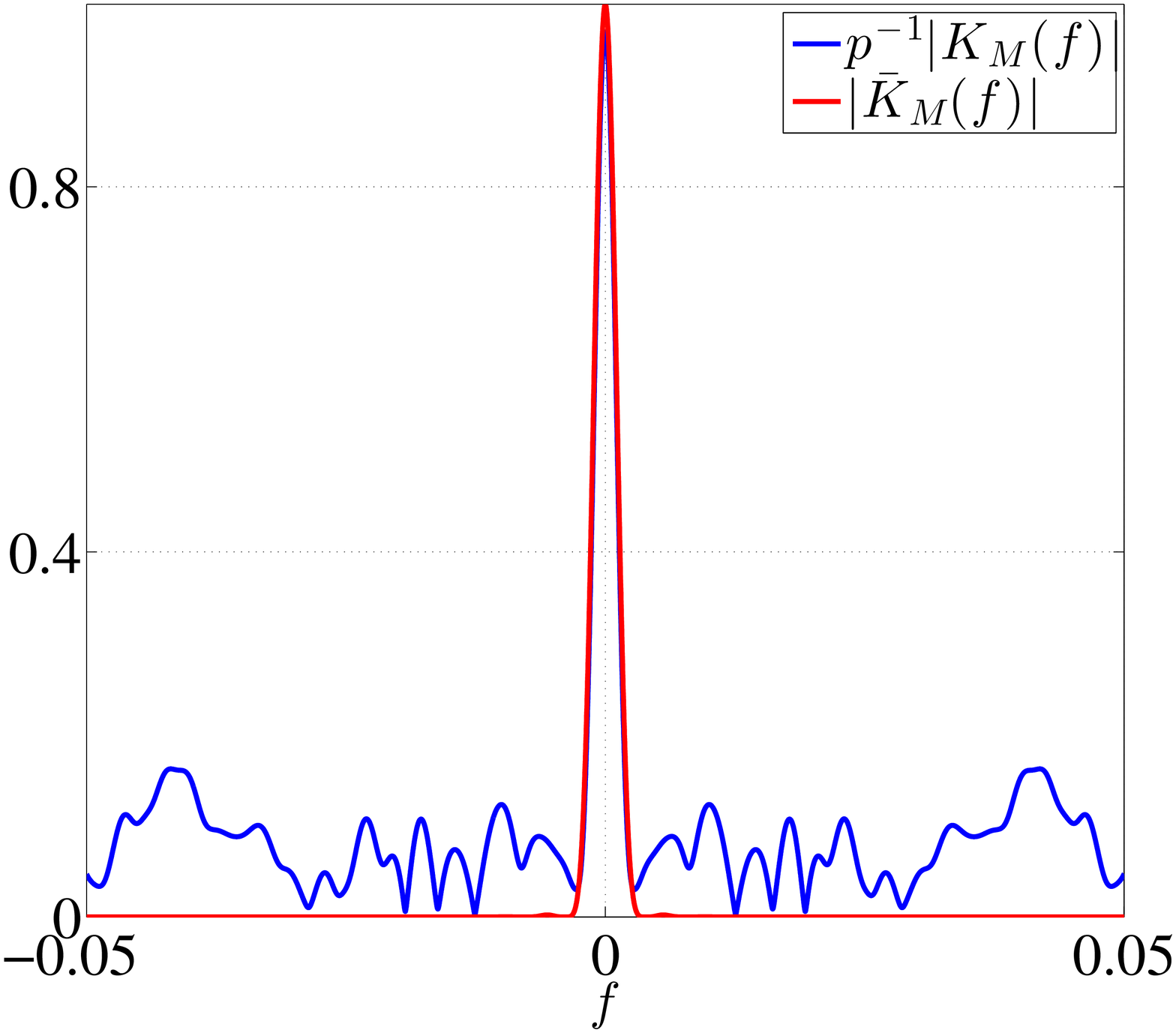}
 \caption{\small $p^{-1}|K_M(f)|$ and $|\bar{K}_M(f)|$}
 \end{subfigure}
 \begin{subfigure}[b]{0.5\textwidth}
 \includegraphics[width=\textwidth, trim = 10mm 0mm 0mm 10mm, clip = true]{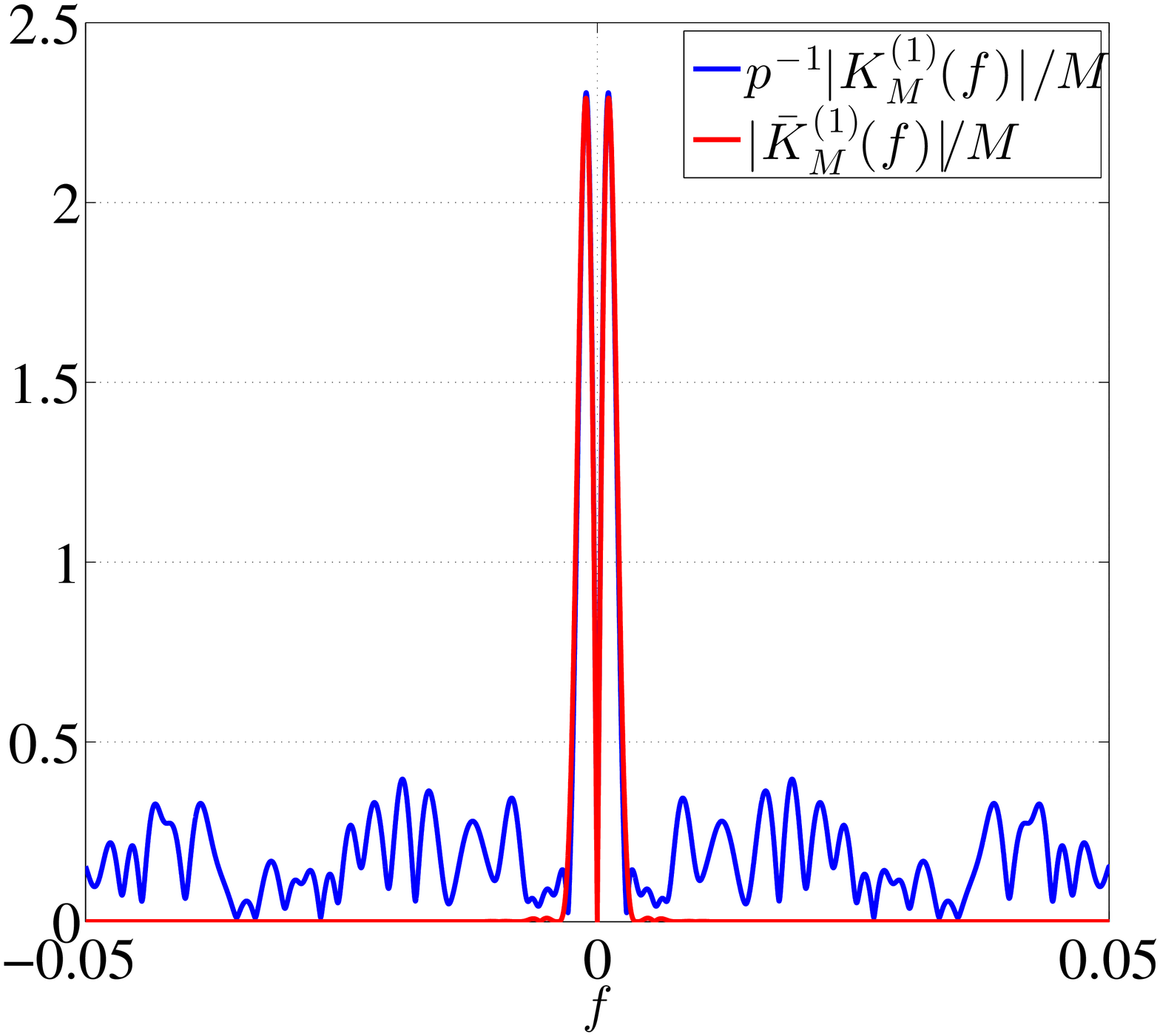}
 \caption{ \small $p^{-1}|K'_M(f)|/M$ and $|\bar{K}'_M(f)|/M$}
\end{subfigure}
 \caption{ \small Plots of the random kernel}
  \label{fig:random_kernel}
\end{figure}

In order to satisfy the conditions \eqref{eqn:condition:Q1} and \eqref{eqn:condition:Q2}, we require that
the polynomial $Q \left( f \right)$ in \eqref{eqn:formofpoly} satisfies
\begin{align}
  Q \left( f_j \right) & = \sum_{k = 1}^s \alpha_k K_M \left( f_j - f_k
  \right) + \sum_{k = 1}^s \beta_k K_M' \left( f_j - f_k \right) =
  \tmop{sign} \left( c_j \right),  \label{eqn:Q1:specific}\\
  Q' \left( f_j \right) & = \sum_{k = 1}^s \alpha_k K_M' \left( f_j - f_k
  \right) + \sum_{k = 1}^s \beta_k K_M'' \left( f_j - f_k \right) = 0
   \label{eqn:Q2:specific}
\end{align}
for all  $f_j  \in \Omega$. As for $\bar{Q}(f)$, the constraint (\ref{eqn:Q1:specific}) guarantees that $Q \left( f \right)$
satisfies the interpolation condition (\ref{eqn:condition:Q1}), and the
constraint (\ref{eqn:Q2:specific}) helps ensure that $|Q \left( f \right)|$
achieves its maximum at frequencies in $\Omega$.

We now have $2 s$ linear constraints (\ref{eqn:Q1:specific}),
(\ref{eqn:Q2:specific}) on $2 s$ unknown variables $\alpha, \beta$.   The remainder of the proof consists of three steps:

\begin{enumerate}
\item Show that the linear system (\ref{eqn:Q1:specific}),
  (\ref{eqn:Q2:specific}) is invertible with high probability using matrix Bernstein inequality \cite{Tropp:2011vm};  
\item Show $\left| Q^{(\ell)}(f) - \bar{Q}^{(\ell)}(f) \right|$, the random perturbations introduced by the random observation process, are small on a set of discrete points with high probability, implying the random dual polynomial satisfies the constraints in Proposition \ref{pro:optimality} on the grid; This step is proved using a modification of the idea in \cite{Candes:2007es}.
\item Extend the result to $[0, 1]$ using Bernstein's polynomial inequality \cite{Schaeffer:1941wm} and eventually show $|Q(f)|<1$ for $ f \notin \Omega$.
\end{enumerate}

\subsection{Invertibility}\label{sec:invertibility}

In this section we show the linear system (\ref{eqn:Q1:specific}) and
(\ref{eqn:Q2:specific}) is invertible. Rewrite the linear system of equations
(\ref{eqn:Q1:specific}) and (\ref{eqn:Q2:specific}) into the following
matrix-vector form:
\begin{align}
  \left[\begin{array}{cc}
    D_0 & \frac{1}{\sqrt{\left| \bar{K}_M'' \left( 0 \right) \right|}} D_1\\
    - \frac{1}{\sqrt{\left| \bar{K}_M'' \left( 0 \right) \right|}} D_1 & -
    \frac{1}{\left| \bar{K}_M'' \left( 0 \right) \right|} D_2
  \end{array}\right] \left[\begin{array}{c}
    \alpha\\
    \sqrt{\left| \bar{K}_M'' \left( 0 \right) \right|} \beta
  \end{array}\right] & = \left[\begin{array}{c}
   u\\
    0
  \end{array}\right],  \label{eqn:matrix:vector}
\end{align}
where $\left[ D_{\ell} \right]_{jk} = K_M^{\left( \ell \right)} \left( f_j -
f_k \right)$, and $u = \operatorname*{sign}(c)$. Note that we still rescale the derivatives using the deterministic quantity $\bar{K}_M'' \left( 0 \right)$ rather than the random variable $K_M'' \left( 0 \right)$.

The expectation computation \eqref{eqn:expectationofkernel} implies that $\mathbbm{E} \left[ D_{\ell} \right]_{j k}
=\mathbbm{E}K_M^{\left( \ell \right)} \left( f_j - f_k \right) = p \left[
\bar{D}_{\ell} \right]_{j k},$ where $\left[ \bar{D}_{\ell} \right]_{jk} =
\bar{K}_M^{\left( \ell \right)} \left( f_j - f_k \right)$. 
Define
\begin{align*}
  D & =  \left[\begin{array}{cc}
    D_0 & \frac{1}{\sqrt{\left| \bar{K}_M'' \left( 0 \right) \right|}} D_1\nn\\
    - \frac{1}{\sqrt{\left| \bar{K}_M'' \left( 0 \right) \right|}} D_1 & -
    \frac{1}{\left| \bar{K}_M'' \left( 0 \right) \right|} D_2
  \end{array}\right]\nn\\
  & =  \left[\begin{array}{cc}
    D_0 & \frac{1}{\sqrt{\left| \bar{K}_M'' \left( 0 \right) \right|}} D_1\\
    \frac{1}{\sqrt{\left| \bar{K}_M'' \left( 0 \right) \right|}} D_1^{\ast} & -
    \frac{1}{\left| \bar{K}_M'' \left( 0 \right) \right|} D_2
  \end{array}\right]\nn\\
  & =  \frac{1}{M} \sum_{j = - 2 M}^{2 M} g_M \left( j \right)
  \delta_j e \left( j \right) e \left( j \right)^{\ast}
\end{align*}
where
\begin{align}\label{eqn:ej}
  e \left( j \right) & =  \left[\begin{array}{c}
    e^{- i 2 \pi f_1 j}\\
    \vdots\\
    e^{- i 2 \pi f_s j}\\
    \frac{i 2 \pi j}{\sqrt{\left| \bar{K}_M'' \left( 0 \right) \right|}} e^{- i
    2 \pi f_1 j}\\
    \vdots\\
    \frac{i 2 \pi j}{\sqrt{\left| \bar{K}_M'' \left( 0 \right) \right|}} e^{- i
    2 \pi f_s j}
  \end{array}\right] .
\end{align}
Then we have
\begin{align*}
\E D &=  \frac{1}{M} \sum_{j = - 2 M}^{2 M} g_M \left( j \right)
  \E \{\delta_j\} e \left( j \right) e \left( j \right)^{\ast}\nn\\
  & =  p \frac{1}{M} \sum_{j = - 2 M}^{2 M} g_M \left( j \right)
e \left( j \right) e \left( j \right)^{\ast}\nn\\
 & = p \bar{D},
\end{align*}
with $\bar{D}$ defined in \eqref{eqn:defbarD}. As a consequence, we have
\begin{align*}
  D -\mathbbm{E}D & = D - p \bar{D}\\
  & = \sum_{j = - 2 M}^{2 M} \frac{1}{M} g_M \left( j \right)
  \left( \delta_j - p \right) e \left( j \right) e \left( j \right)^{\ast}\\
  & = \sum_{j = - 2 M}^{2 M} X_j .
\end{align*}
with $X_j = \frac{1}{M} g_M \left( j \right) \left( \delta_j - p
\right) e \left( j \right) e \left( j \right)^{\ast}$ a zero mean random
self-adjoint matrix. We will apply the noncommutative Bernstein
inequality to show that $D$ concentrates about its mean $p \bar{D}$ with high
probability.

\begin{lemma}[{Noncommutative Bernstein Inequality, \cite[Theorem 1.4]{Tropp:2011vm}}]\label{lm:matrixbernstein} Let $\left\{ X_j \right\}$ be a finite
  sequence of independent, random self-adjoint matrices of dimension d.
  Suppose that
  \begin{align*}
    \mathbbm{E}X_j & = 0\\
    \left\| X_j \right\| & \leq R, \tmop{almost} \tmop{surely}\\
    \sigma^2 & = \Big\| \sum_j \mathbbm{E} \left( X_j^2 \right) \Big\| .
  \end{align*}
  Then for all $t \geq 0$,
  \begin{align*}
    \mathbbm{P} \Big\{ \Big\| \sum_j X_j \Big\| \geq t \Big\} & \leq d
    \exp \left( \frac{- t^2 / 2}{\sigma^2 + R t / 3} \right) .
  \end{align*}
\end{lemma}

For $\tau > 0$, define the event
\begin{align}\label{eqn:dfE1}
\mathcal{E}_{1, \tau} = \left\{ \left\| p^{- 1} D - \bar{D}
  \right\| \leq \tau \right\}.
\end{align}
The following lemma, proved in Appendix \ref{apx:lm:invertibility}, shows that $\mathcal{E}_{1,\tau}$ has a high probability if $m$ is large enough.
\begin{lemma}
  \label{lm:invertibility}If $\tau \in ( 0, 0.6377)$, then we have $\mathbbm{P} \left( \mathcal{E}_{1,\tau}
  \right) \geq 1 - \delta$ provided
  \begin{align*}
    m & \geq \frac{50}{\tau^2} s \log \frac{2
    s}{\delta}.
  \end{align*}
\end{lemma}

Note that an immediate consequence of Lemma~\ref{lm:invertibility} is that $D$ is invertible on $\mathcal{E}_{1,\tau}$.  Additionally, Lemma \ref{lm:invertibility} allows us to control the norms of the submatrices of $D^{-1}$.  For that purpose, we partition $D^{-1}$ as 
\begin{align*}
D^{-1} = \begin{bmatrix}
L & R
\end{bmatrix}
\end{align*}
with $L$ and $R$ both $2s \times s$ and obtain:

\begin{corollary}\label{cor:bdonmatrix}
  On the event $\mathcal{E}_{1, \tau}$ with $\tau \in \left( 0, \frac{1}{4} \right]$, we have
  \begin{align*}
    \left\| L - p^{- 1} \bar{L} 
    \right\| & \leq 2 \left\| \bar{D}^{- 1} \right\|^2 p^{- 1} \tau\\
    \left\| L \right\| & \leq 2 \left\| \bar{D}^{- 1}
    \right\| p^{- 1}.
  \end{align*}
\end{corollary}

The proof of this corollary uses elementary matrix analysis and can be found in Appendix~\ref{sec:bdonmatrix}.  Since on the event $\mathcal{E}_{1, \tau}$ with $\tau < 1/4$ the matrix $D = \left[\begin{array}{cc}
  D_0 & D_1\\
  D_1 & D_2
\end{array}\right]$ is invertible, we solve for $\alpha$ and $\sqrt{\left| \bar{K}_M'' \left( 0 \right) \right|} \beta$ from
(\ref{eqn:matrix:vector}):
\begin{align}\label{eqn:alphabetasoln}
  \left[\begin{array}{c}
    \alpha\\
    \sqrt{\left| \bar{K}_M'' \left( 0 \right) \right|} \beta
  \end{array}\right] & = D^{- 1} \left[\begin{array}{c}
    u\\
    0
  \end{array}\right]\nn\\
  & = L u.
\end{align}
In the next section, we will plug \eqref{eqn:alphabetasoln} back into \eqref{eqn:formofpoly}, and analyze the effect of random perturbations on the polynomial $Q(f)$.

\subsection{Random Perturbations}\label{sec:perturbations}

In this section, we show that the dual
polynomial $Q \left( f \right)$ concentrates around $\bar{Q}(f)$ on a discrete set
$\Omega_{\mathrm{grid}}$.

We introduce a random analog of $\bar{v}_\ell$, defined by \eqref{eqn:vlbar}, as
\begin{align}\label{eqn:vl}
  {v}_{\ell} \left( f \right) & = \frac{1}{\sqrt{\left| \bar{K}_M''
  \left( 0 \right) \right|}^\ell} \left[\begin{array}{c}
    {K}_M^{\left( \ell \right)} \left( f - f_1 \right)^{\ast}\\
    \vdots\\
    {K}_M^{\left( \ell \right)} \left( f - f_s \right)^{\ast}\\
      \frac{1}{\sqrt{\left| \bar{K}_M''
  \left( 0 \right) \right|}}{K}_M^{(\ell+1)} \left( f - f_1 \right)^{\ast}\\
  \vdots\\
        \frac{1}{\sqrt{\left| \bar{K}_M''
  \left( 0 \right) \right|}}{K}_M^{(\ell+1)} \left( f - f_s \right)^{\ast}
    \end{array}\right]\nn\\
     & =  \frac{1}{M}  \sum_{j = - 2 M}^{2 M} \Bigg( \frac{i 2 \pi
  j}{\sqrt{\left| \bar{K}_M'' \left( 0 \right) \right|}} \Bigg)^{\ell} g_M
  \left( j \right) \delta_j e^{i2\pi f j} e(j).
\end{align}
with ${K}_M^{\left( \ell \right)}$ the $\ell$th derivative of ${K}_M$, and $e(j)$ defined in \eqref{eqn:ej}.  The expectation of $v_\ell$ is equal to $p$ times its deterministic counterpart defined by~\eqref{eqn:vlbar}:
\begin{align*}
  \mathbbm{E}v_{\ell} \left( f \right) & =  p \bar{v}_{\ell} \left( f \right), \forall f \in \left[
  0, 1 \right]
\end{align*}
Then, in a similar fashion to \eqref{eqn:Qfasvlbar}, we rewrite
\begin{align*}
  \frac{1}{\sqrt{|\bar{K}_M''(0)|}^\ell} Q^{(\ell)} \left( f \right) & = \sum_{k = 1}^s \alpha_k \frac{1}{\sqrt{|\bar{K}_M''(0)|}^\ell}K_M^{(\ell)}
  \left( f - f_k \right)\nn\\
  & \qquad\qquad+ \sum_{k=1}^s 
  \sqrt{\left| \bar{K}_M'' \left(
  0 \right) \right|} \beta_k \frac{1}{\sqrt{|\bar{K}_M''(0)|}^{\ell+1}}{K}_M^{(\ell+1)} \left( f - f_k \right)\nn\\
  & = {v}_\ell \left( f \right)^{\ast} Lu = \left<{L}u, {v}_\ell(f)\right> = \left<u, L^* v_\ell(f)\right>.
\end{align*}
We decompose $L^*v_\ell(f)$ into three parts:
\begin{align*}
L^*v_\ell(f) & = [(L-p^{-1}\bar{L}) + p^{-1}\bar{L}]^*[(v_\ell(f) - p\bar{v}_\ell(f)) + p \bar{v}_\ell(f)]\nn\\
& = \bar{L}^*\bar{v}_\ell(f) + L^*(v_\ell(f) - p\bar{v}_\ell(f)) + (L-p^{-1}\bar{L})^*p \bar{v}_\ell(f),
\end{align*}
which induces a decomposition on $\frac{1}{\sqrt{|\bar{K}_M''(0)|}^\ell} Q^{(\ell)}(f)$
\begin{align}\label{eqn:Qfdecomp}
  \frac{1}{\sqrt{|\bar{K}_M''(0)|}^\ell} Q^{(\ell)} \left( f \right) & =  \left<u, L^*v_\ell(f)\right>\nn\\
  & = \left<u, \bar{L}^*\bar{v}_\ell(f)\right> + \left<u, L^*(v_\ell(f) - p\bar{v}_\ell(f))\right> + \left<u, (L-p^{-1}\bar{L})^*p \bar{v}_\ell(f)\right>\nn\\
  & = \frac{1}{\sqrt{|\bar{K}_M''(0)|^\ell}}\bar{Q}^{(\ell)} \left( f \right) + I_1^\ell \left( f \right) + I_2^\ell \left( f
  \right) .
\end{align}
Here $\frac{1}{\sqrt{|\bar{K}_M''(0)|^\ell}}\bar{Q}^{(\ell)}(f) = \left<u, \bar{L}^*\bar{v}_\ell(f)\right> = \left<\bar{L}u, \bar{v}_\ell(f)\right>$ as in \eqref{eqn:Qfasvlbar} and we have defined
  \begin{align*}
  I_1^\ell(f) = \left<u, L^*(v_\ell(f) - p\bar{v}_\ell(f))\right>
  \end{align*}
  and 
  \begin{align*}
I_2^\ell(f) =  \left<u, (L-p^{-1}\bar{L})^*p \bar{v}_\ell(f)\right>.
  \end{align*}

The goal of the remainder of this section is to show, in Lemma \ref{lm:I1} and \ref{lm:I2}, that $I_1^\ell \left( f
\right)$ and $I_2^\ell \left( f \right)$ are small on a set of grid points $\Omega_{\mathrm{grid}}$ with
high probability. 


The proof of Lemma \ref{lm:I1}, which shows $I_1^\ell(f)$ is small on $\Omega_{\mathrm{grid}}$, essentially follows that of Cand\`es and Romberg \cite{Candes:2007es}. We include the proof details here for completeness, but very little changes in the argument. Since  $I_1^\ell(f) = \left<u, L^*(v_\ell(f) - p\bar{v}_\ell(f))\right>$ is a weighted sum of independent random variables following a symmetric distribution on the complex unit circle, for fixed $f \in [0, 1]$, we apply Hoeffding's inequality to control its value. This in turn requires an estimate of $\|L^*(v_\ell(f) - p \bar{v}_\ell(f))\|_2$. In Lemma \ref{lm:vw}, we first use
concentration of measure (Lemma \ref{lm:concentration}) to
establish that $\left\| v_{\ell} \left( f \right) - p \bar{v}_{\ell} \left( f
\right) \right\|_2$ is small with high probability. In Lemma \ref{lm:vsd}, we then combine Lemma \ref{lm:vw} and Lemma \ref{lm:invertibility} to show $\left\| L^*(v_\ell(f) - p \bar{v}_\ell(f))\right\|_2$ is small. The extension from a fixed $f$ to a finite set $\Omega_{\mathrm{grid}}$ relies on the union bound.

We start with bounding $\|v_\ell(f) - p \bar{v}_\ell(f)\|_2$ in the following lemma. The proof given in Appendix \ref{apx:lm:vw} is based on an inequality of Talagrand.
\begin{lemma}
  \label{lm:vw}Fix $f \in \left[ 0, 1 \right]$. Let
  \begin{align*}
    \bar{\sigma}_{\ell}^2 & := 2^{4 \ell + 1} \frac{m}{M^2}
    \max \left\{ 1, 2^4 \frac{s}{\sqrt{m}} \right\}
  \end{align*}
  and fix a positive number
  \begin{align*}
  a \leq
  \begin{cases}
  \sqrt{2}m^{1/4} & \text{if $2^4 \frac{s}{\sqrt{m}} \geq 1$,}
\\
\frac{\sqrt{2}}{4} \sqrt{\frac{m}{s}} &\text{otherwise.}
  \end{cases}
  \end{align*}
Then we have
  \begin{align*}
    \mathbbm{E} \left\| v_{\ell} \left( f \right) -p\bar{v}_\ell \left(
    f \right) \right\|_2 & \leq 2^{2 \ell + 3}
    \frac{\sqrt{m s}}{M}
  \end{align*}
  \begin{align*}
    \mathbbm{P} \left(\left\| v_{\ell} \left( f \right)
    -p\bar{v}_\ell \left( f \right) \right\|_2 > 2^{2 \ell + 3}
     \frac{\sqrt{m s}}{M} + a \bar{\sigma}_{\ell}, \ell = 0, 1, 2, 3
    \right) & \leq  64 e^{- \gamma a^2}
  \end{align*}
  for some $\gamma > 0$.
\end{lemma}

The following lemma combines Lemma \ref{lm:vw} and Corollary \ref{cor:bdonmatrix} to show $\|L^*(v_\ell(f) - p\bar{v}_\ell(f)\|_2$ is small with high probability. 
\begin{lemma}
  \label{lm:vsd}Let $\tau \in (0, 1/4]$. Consider a finite set $\Omega_{\mathrm{grid}} = \left\{ f_d \right\}$. With
  the same notation as last lemma, we have
  \begin{align*}
    &  & \mathbbm{P} \left[ \sup_{f_d \in \Omega_{\mathrm{grid}}} \left \| L^*(v_{\ell} \left( f_d \right)
    -p\bar{v}_\ell \left( f_d \right))  \right\|_2 \geq 4 \left( 2^{2\ell + 3}\sqrt{\frac{s}{m}} + \frac{M}{m} a
    \bar{\sigma}_\ell \right), \ell = 0, 1, 2, 3 \right]\\
    &  & \leq 64 \left| \Omega_{\mathrm{grid}} \right| e^{- \gamma a^2} +\mathbbm{P} \left(
    \mathcal{E}_{1, \tau}^c \right) .
  \end{align*}
\end{lemma}

\begin{proof}[Proof of Lemma \ref{lm:vsd}]
  Conditioned on the event
  \begin{align*}
  \bigcap_{\ell, f_d \in \Omega_{\mathrm{grid}}} \left\{ \left\|
  v_{\ell} \left( f_d \right) -p\bar{v}_\ell \left( f_d \right) \right\|_2
  \leq 2^{2 \ell + 3} \frac{\sqrt{m s}}{M} + a
  {\bar{\sigma}_{\ell}} \right\} \bigcap \mathcal{E}_{1, \tau}
  \end{align*}
   we have
  \begin{align*}
    \left\| L^*\left( v_\ell \left( f_d \right) -p\bar{v}_\ell \left( f_d \right)
    \right)\right\|_2 & \leq \left\| L\right\| \left( 2^{2\ell+3} \frac{\sqrt{m s}}{M} + a \bar{\sigma}_\ell \right)\nn\\
    & \leq 2 \left\| \bar{D}^{- 1} \right\| p^{- 1} \left( 2^{2\ell+3} \frac{\sqrt{m s}}{M} + a \bar{\sigma}_\ell \right)\nn\\
    & \leq 4 \left( 2^{2\ell + 3}\sqrt{\frac{s}{m}} +
    \frac{M}{m} a \bar{\sigma}_\ell \right),
  \end{align*}
  where we have used Proposition \ref{pro:sysmtx:deterministic} and Corollary \ref{cor:bdonmatrix}, and plugged in $p = m/M$. The claim of the lemma then follows from the union bound.
\end{proof}

Lemma \ref{lm:vsd} together with Hoeffding's inequality allow us to control the
size of $\sup_{f_d \in \Omega_{\mathrm{grid}}} I_1^\ell \left( f_d \right)$:

\begin{lemma}
  \label{lm:I1}There exists a numerical constant $C$ such that if
  \begin{align*}
m & \geq C  \max \left\{ \frac{1}{\varepsilon^2} \max\left( s\log\frac{|\Omega_{\mathrm{grid}}|}{\delta}, \log^2 \frac{\left| \Omega_{\mathrm{grid}} \right|}{\delta} \right), s \log \frac{s}{\delta} \right\},
  \end{align*}
  then we have
  \begin{align*}
    \mathbbm{P} \Big\{ \sup_{f_d \in \Omega_{\mathrm{grid}}} \left| I_1^\ell \left( f_d \right)
    \right| \leq \varepsilon, \ell = 0, 1, 2, 3 \Big\} & \geq 1 - 12 \delta
  \end{align*}
\end{lemma}

Next lemma controls $I_2^\ell(f)$.  Its proof is similar to the proof of Lemma~\ref{lm:I1}.
\begin{lemma}
  \label{lm:I2}There exists a numerical constant $C$ such that if
  \begin{align*}
    m & \geq C \frac{1}{\varepsilon^2} s \log \frac{s}{\delta} \log
    \frac{\left| \Omega_{\mathrm{grid}} \right|}{\delta},
  \end{align*}
then we have
  \begin{align*}
    \mathbbm{P} \Big( \sup_{f_d \in \Omega_{\mathrm{grid}}} \left| I_2^\ell \left( f_d \right)
    \right| < \varepsilon, \ell = 0, 1, 2, 3 \Big) & \leq 1 - 8 \delta
  \end{align*}
\end{lemma}

Both Lemmas~\ref{lm:I1} and~\ref{lm:I2} are proven in the Appendix.

Denote 
\begin{align}
\mathcal{E}_2 = \Bigg\{ \sup_{f_d \in \Omega_{\mathrm{grid}}} \Bigg|\frac{1}{\sqrt{|\bar{K}_M''(0)|^\ell}} Q^{(\ell)}(f_d) - \frac{1}{\sqrt{|\bar{K}_M''(0)|^\ell}}\bar{Q}^{(\ell)}(f_d)\Bigg| \leq \frac{\varepsilon}{3}, \ell = 0, 1, 2, 3\Bigg\}.\label{eqn:E2}
\end{align}
Combining the decomposition \eqref{eqn:Qfdecomp}, Lemma \ref{lm:I1}, and Lemma \ref{lm:I2}  with suitable redefinition of $\varepsilon$ and $\delta$ immediately yields the following proposition
\begin{proposition}
Suppose $\Omega_{\mathrm{grid}} \subset [0, 1]$ is a finite set of points, $\varepsilon > 0$ controls the maximal deviation of the random polynomial and its derivatives from their expectations as in \eqref{eqn:E2}, and $\delta \in (0, 1)$ is small constant to control probability. Then there exists constant $C$ such that \begin{align}\label{eqn:lbdm}
m & \geq C \frac{1}{\varepsilon^2} \max \left\{\log^2 \frac{\left| \Omega_{\mathrm{grid}} \right|}{\delta}, s \log \frac{s}{\delta}\log\frac{|\Omega_{\mathrm{grid}}|}{\delta} \right\},
\end{align}
is sufficient to guarantee 
\begin{align*}
\Pr(\mathcal{E}_2) \geq 1 - \delta.
\end{align*}
\end{proposition}

\subsection{Extension to Continuous Domain}\label{sec:ext_cont}
We have proved that $\frac{1}{\sqrt{|\bar{K}_M''(0)|^\ell}} Q^{(\ell)}(f)$ and $\frac{1}{\sqrt{|\bar{K}_M''(0)|^\ell}}\bar{Q}^{(\ell)}(f)$ are not far on a set of grid points. This section aims extending this statement to everywhere in $[0, 1]$, and show $|Q(f)| < 1$ for $f \notin \Omega$ eventually. The key is the following Bernstein's polynomial inequality:
\begin{lemma}[Bernstein's polynomial inequality, \cite{Schaeffer:1941wm}] Let $p_N$ be any polynomial of degree
  $N$ with complex coefficients. Then
  \begin{align*}
    \sup_{\left| z \right| \leq 1} \left| p' \left( z \right) \right| & \leq
    N \sup_{\left| z \right| \leq 1} \left| p \left( z \right) \right| .
  \end{align*}
\end{lemma}

Our first proposition verifies that our random dual polynomial is close to the deterministic dual polynomial on all of $[0,1]$

\begin{proposition}
  \label{pro:extcont}Suppose $\Delta_f \geq \Delta_{\min} = \frac{1}{M}$ and
  \begin{align}
m & \geq
    C \max \left\{ \frac{1}{\varepsilon^2} \log^2 \frac{M}{\delta
    \varepsilon}, \frac{1}{\varepsilon^2} s \log \frac{s}{\delta} \log
    \frac{M}{\delta \varepsilon} \right\}.\label{eqn:everywherecondition}
  \end{align}
  Then with probability $1-\delta$, we have
  \begin{align}\label{eqn:erreverywhere}
\left|\frac{1}{\sqrt{|\bar{K}_M''(0)|^\ell}} Q^{(\ell)}(f) - \frac{1}{\sqrt{|\bar{K}_M''(0)|^\ell}}\bar{Q}^{(\ell)}(f)\right| \leq \varepsilon, \forall f \in [0, 1], \ell = 0, 1, 2, 3.
  \end{align}
\end{proposition}

\begin{proof}
  It suffices to prove \eqref{eqn:erreverywhere} on $\mathcal{E}_{1, 1/4}$ and $\mathcal{E}_2$ and then modify the lower bound \eqref{eqn:lbdm}. We first give a very rough estimate of $\sup_f \frac{1}{\sqrt{|\bar{K}_M''(0)|}^\ell}\left| Q^{(\ell)} \left( f \right) \right|$
  on the set $\mathcal{E}_{1, 1 / 4}$:
  \begin{align*}
    \frac{1}{\sqrt{|\bar{K}_M''(0)|}^\ell}\left| Q^{(\ell)} \left( f \right) \right| & = \left| \left<u, L^*v_\ell(f)\right> \right|\nn\\
    & \leq \|u\|_2 \|L\| \|v_\ell(f)\|_2\nn\\
    & \leq C p^{- 1} s\nn\\
    & \leq C M^2
  \end{align*}
  where we have used $\left\| u \right\|_2 \leq \sqrt{s}$ and $\left\| v_{\ell}
  \left( f \right) \right\|_2 \leq C \sqrt{s}$. To see the
  latter, we note
  \begin{align*}
    \left\| v_{\ell} \left( f \right) \right\|_2 & \leq  \sum_{j = - 2 M}^{2
    M} \Bigg\| \frac{1}{M} \Bigg( \frac{i 2 \pi j}{\sqrt{\left| \bar{K}_M''
    \left( 0 \right) \right|}} \Bigg)^{\ell} g_M \left( j\right)
    e(j) \Bigg\|_2\nn\\
    & \leq \left( 4 M + 1 \right) \frac{1}{M} 4^{\ell+1} s^{1 / 2},
  \end{align*}
  where we have used
  \begin{align*}
\|g_M\|_\infty &\leq 1, \\
\Big|\frac{2\pi j}{\sqrt{\bar{K}_M''(0)}}\Big| &\leq 4 \tmop{when} M \geq 2,\\
    \|e(j)\|_2^2 &\leq s \left( 1 +
    \max_{\left| j \right| \leq 2 M} \frac{\left( 2 \pi j \right)^2}{\left|
    K_M'' \left( 0 \right) \right|} \right) \leq 14 s \text{\ when\ } M \geq 4.
\end{align*}
  Viewing $\frac{1}{\sqrt{|\bar{K}''_M(0)|}}Q^{(\ell)} \left( \cdot \right)$ as a trigonometric polynomial in $z = e^{-i2\pi f}$ of degree $2M$, according to
  Bernstein's polynomial inequality, we get
  \begin{align*}
    \left| \frac{1}{\sqrt{|\bar{K}''_M(0)|}}Q^{(\ell)} \left( f_a \right) - \frac{1}{\sqrt{|\bar{K}''_M(0)|}}Q^{(\ell)} \left( f_b \right) \right| & \leq \left|
    e^{- i 2 \pi f_a} - e^{- i 2 \pi f_b} \right| \sup_z \left| \frac{d \frac{1}{\sqrt{|\bar{K}''_M(0)|}}Q^{(\ell)} \left( z \right)}{d
    z} \right|\nn\\
    & \leq 4 \pi \left| f_a - f_b \right| 2M \sup_f \left| \frac{1}{\sqrt{|\bar{K}''_M(0)|}}Q^{(\ell)} \left( f \right) \right|\nn\\
    & \leq C M^3 \left| f_a - f_b \right| .
  \end{align*}
  We select $\Omega_{\mathrm{grid}} \subset [0, 1]$ such that for any $f \in [0, 1]$, there exists a point $f_d \in \Omega_{\mathrm{grid}}$ satisfying
  $\left| f - f_d \right| \leq \frac{\varepsilon}{3 C M^3}$. The size of
  $\Omega_{\mathrm{grid}}$ is less than $3 C M^3 / \varepsilon$.

With this choice of $\Omega_{\mathrm{grid}}$, on the set $\mathcal{E}_{1, 1/4} \bigcap \mathcal{E}_2$ we have
  \begin{align*}
& \left|\frac{1}{\sqrt{|\bar{K}_M''(0)|^\ell}} Q^{(\ell)}(f) - \frac{1}{\sqrt{|\bar{K}_M''(0)|^\ell}}\bar{Q}^{(\ell)}(f)\right|\nn\\
&\leq \left|\frac{1}{\sqrt{|\bar{K}_M''(0)|^\ell}} Q^{(\ell)}(f) - \frac{1}{\sqrt{|\bar{K}_M''(0)|^\ell}} Q^{(\ell)}(f_d)\right| + \left|\frac{1}{\sqrt{|\bar{K}_M''(0)|^\ell}} Q^{(\ell)}(f_d) - \frac{1}{\sqrt{|\bar{K}_M''(0)|^\ell}} \bar{Q}^{(\ell)}(f_d)\right|\nn\\
& + \left|\frac{1}{\sqrt{|\bar{K}_M''(0)|^\ell}} \bar{Q}^{(\ell)}(f_d) - \frac{1}{\sqrt{|\bar{K}_M''(0)|^\ell}}\bar{Q}^{(\ell)}(f)\right|\nn\\
&\leq CM^3 |f-f_d| + \frac{\varepsilon}{3} + CM^3 |f-f_d|\nn\\
&\leq \varepsilon, \forall f \in [0, 1].
  \end{align*}
Finally, we modify the condition \eqref{eqn:lbdm} according to our choice of $\Omega_{\mathrm{grid}}$:
  \begin{align*}
    m & \geq C \max \left\{ \frac{1}{\varepsilon^2} \log^2 \frac{M}{\delta
    \varepsilon}, \frac{1}{\varepsilon^2} s \log \frac{s}{\delta} \log
    \frac{M}{\delta \varepsilon} \right\}
  \end{align*}
  \end{proof}
  
An immediate consequence of Proposition \ref{pro:extcont} and the bound \eqref{eqn:defaraway} of Proposition \ref{pro:expected} is the following estimate on $Q(f)$ for $f \in \Omega_{\mathrm{far}} = [0, 1]/\bigcup_k \left[ f_k - f_{b, 1}, f_k + f_{b, 1}\right]$:
\begin{lemma}\label{lm:faraway}
Suppose $\Delta_f \geq \Delta_{\min} = \frac{1}{M}$ and
  \begin{align*}
m & \geq C \max \left\{\log^2 \frac{M}{\delta}, s \log \frac{s}{\delta} \log
    \frac{M}{\delta} \right\}.
  \end{align*}
  Then with probability $1-\delta$, we have
 \begin{align}
|Q(f)| < 1, \forall f \in \Omega_{\mathrm{far}}. 
\end{align}
\end{lemma}
\begin{proof}
It suffices to choose $\varepsilon = 10^{-5}$. The rest follows from \eqref{eqn:erreverywhere}, triangle inequality, and modification of the constant in \eqref{eqn:everywherecondition}.
\end{proof}

A similar statement holds for $f \in
\Omega_{\tmop{near}} = \bigcup_k \left[ f_k - f_{b, 1}, f_k + f_{b, 1}
\right]$.

\begin{lemma}
  \label{lm:near}Suppose $\Delta_f \geq \Delta_{\min} = \frac{1}{M}$ and
  \begin{align*}
    m & \geq C \max \left\{ \log^2 \left( \frac{M}{\delta} \right), s \log
    \frac{s}{\delta} \log \frac{M}{\delta} \right\} .
  \end{align*}
  Then we have $\left| Q \left( f \right) \right| < 1$ for all $f \in
  \Omega_{\tmop{near}} .$
\end{lemma}

\begin{proof}
 Define $Q_R(f) = \operatorname*{Re}(Q(f))$ and $Q_I(f) = \operatorname*{Im}(Q(f))$. Since $\left| Q \left( f_k \right) \right| = 1 $ and $Q' \left( f_k
  \right) = 0$ with the latter implying
\begin{align*}
\frac{d |Q|}{d f}(f) = \frac{Q_R'(f) Q_R(f) + Q_I'(f)Q_I(f)}{|Q(f)|} = 0
\end{align*}  
we only need to show $\frac{d^2 |Q(f)|}{d f^2} < 0$ on $\Omega_{\mathrm{near}}$. Take the second order derivative of $\left| Q \left( f \right) \right|$:
  \begin{align*}
    \frac{d^2 \left| Q \right|}{df^2} \left( f \right) & = - \frac{\left(
    Q_R \left( f \right) Q_R' \left( f \right) + Q_I \left( f \right) Q_I'
    \left( f \right) \right)^2}{\left| Q \left( f \right) \right|^3} +
    \frac{\left| Q' \left( f \right) \right|^2 + Q_R \left( f \right) Q_R''
    \left( f \right) + Q_I \left( f \right) Q_I'' \left( f \right)}{\left| Q
    \left( f \right) \right|} .
  \end{align*}
  So it suffices to show that for $f \in \Omega_{\tmop{near}}$
  \begin{align*}
    Q_R \left( f \right) Q_R'' \left( f \right) + \left| Q' \left( f \right)
    \right|^2 + \left| Q_I \left( f \right) \right| \left| Q_I'' \left( f
    \right) \right| & < 0.
  \end{align*}
  
As a consequence of \eqref{eqn:erreverywhere} in Proposition \ref{pro:extcont}, triangle inequality, and \eqref{eqn:near1}-\eqref{eqn:near5} of Proposition \ref{pro:expected}, we have on the set $\mathcal{E}_2$ for any $f
    \in \Omega_{\tmop{near}}$
  \begin{align*}
    Q_R \left( f \right) & \geq \bar{Q}_R \left(f \right) -\varepsilon \geq 0.9182 - \varepsilon\\
    \left| Q_I \left( f \right) \right| & \leq \left| \bar{Q}_I \left( f
    \right) \right| + \varepsilon \leq 3.611\times 10^{-2} + \varepsilon\\
    \frac{1}{\left| \bar{K}_M'' \left( 0 \right) \right|} Q_R'' \left( f \right)
    & \leq \frac{1}{\left| \bar{K}_M'' \left( 0 \right) \right|} \bar{Q}_R''
    \left( f \right)  + \varepsilon \leq -0.314 + \varepsilon\\
    \Big| \frac{1}{\left| \bar{K}_M'' \left( 0 \right) \right|} Q_I'' \left( f
    \right) \Big| & \leq \Big| \frac{1}{\left| \bar{K}_M'' \left( 0 \right)
    \right|} \bar{Q}_I'' \left( f \right) \Big| + \varepsilon \leq 0.5755 + \varepsilon\\
    \Big| \frac{1}{\sqrt{\left| \bar{K}_M'' \left( 0 \right) \right|}} Q_I'
    \left( f \right) \Big| & \leq \Big| \frac{1}{\sqrt{\left| \bar{K}_M''
    \left( 0 \right) \right|}} \bar{Q}_I' \left( f \right) \Big| +
    \varepsilon \leq 0.4346 + \varepsilon.
  \end{align*}
implying
  \begin{align*}
    &  & \frac{1}{\left| \bar{K}_M'' \left( 0 \right) \right|} \left( Q_R
    \left( f \right) Q_R'' \left( f \right) + \left| Q' \left( f \right)
    \right|^2 + \left| Q_I \left( f \right) \right| \left| Q_I'' \left( f
    \right) \right| \right) \leq - 7.865 10^{- 2} + 2.714 \varepsilon + \varepsilon^2< 0
  \end{align*}
  when $\varepsilon$ assumes a sufficiently small numerical value. With this
  choice of $\varepsilon$, the condition of $m$ becomes
  \begin{align*}
    m & \geq C \max \left\{ \log^2 \frac{M}{\delta}, s \log \frac{s}{\delta}
    \log \frac{M}{\delta} \right\} .
  \end{align*}
  Therefore, $|Q(f)| < 1$ on $\Omega_{\mathrm{near}}$ except for $f \in \Omega$.
 We actually proved a stronger result that with probability at least $1 -\delta$
  \begin{align}
    |Q \left( f \right)| & \leq 1 - 0.07 \left| \bar{K}_M'' \left( 0 \right)
    \right| \left( f - f_k \right)^2 \leq 1, \forall f \in \left[ f_k - f_{b,
    1}, f_k + f_{b, 1} \right].
  \end{align}
\end{proof}

\begin{proof}[Proof of Theorem \ref{thm:main}]
Finally, if $\Delta_{\min} \geq \frac{1}{M}$ and
\begin{align*}
  m & \geq C \max \left\{ \log^2 \frac{M}{\delta}, s \log \frac{s}{\delta}
  \log \frac{M}{\delta} \right\},
\end{align*}
combining Lemma \ref{lm:faraway} and \ref{lm:near}, we have proved the claim
of Theorem \ref{thm:main}.
\end{proof}

\section{Numerical Experiments}\label{sec:experiments}
We conducted a series of numerical experiments to test the performance of \eqref{eqn:minimization} under various parameter settings (see Table \ref{tbl:config}).  We use  $J = \{0, \ldots, n-1\}$ for all numerical experiments. 

We compared the performance of two algorithms: the semidefinite program \eqref{eqn:tracemin} and the basis pursuit obtained through discretization:
\begin{align}\label{eqn:bp}
\minimize_{c} \ \|c\|_1\
\text{subject\ to} \ x_j^\star = (Fc)_j, j \in T\,.
\end{align}
Here $F$ is a DFT matrix of appropriate dimension depending on the grid size.   Note that since the components of $c$ are complex, this is a second-order cone problem.  In the following, we use SDP and BP to label the semidefinite program algorithm and the basis pursuit algorithm, respectively.  We solved the SDP with the SDPT3 solver \cite{SDPT3} and the basis pursuit \eqref{eqn:bp} with CVX \cite{cvx} coupled with SDPT3. All parameters of the SDPT3 solver were set to default values and CVX precision was set to `high'.  For the BP, we used three levels of discretization at $4$, $16$, and $64$ times the signal dimension.

To generate our instances of form \eqref{eqn:signalmodel:mv}, we sampled $s = \rho_s n$ normalized frequencies from $[0, 1]$, either \emph{randomly}, or \emph{equispaced}. Random frequencies are sampled randomly on $[0, 1]$ with an additional constraint on the minimal separation $\Delta_f$. Given $s = \rho_s n$, $s$ equispaced frequencies are generated with the same separation $1/s$ with an additional random shift. This random shift  ensures that in most cases, basis mismatch occurs for discretization method. The signal coefficient magnitudes $|c_1|, \cdots, |c_s|$ are either \emph{unit}, i.e., equal to 1, or \emph{fading}, i.e., equal to $.5 + w^2$ with $w$ a zero mean unit variance Gaussian random variable. The signs $\{e^{i\phi_k}, k = 1, \cdots, s\}$  follow either Bernoulli $\pm 1$ distribution, labeled as \emph{real}, or uniform distribution on the complex unit circle, labeled as \emph{complex}. A length $n$ signal was then formed according to model \eqref{eqn:signalmodel:mv}. As a final step, we uniformly sample $m = \rho_m n$ entries of the resulting signal.

We tested the algorithms on three sets of experiments. In the first experiment, by running the algorithms on a randomly generated instance with $n = 256, s = 6$ and $m = 40$ samples selected uniformly at random, we compare SDP and BP's ability of frequency localization and visually illustrate the effect of discretization. We see from Figure \ref{fig:effectofgriding} that SDP recovery followed by retrieving the frequencies according to Proposition \ref{pro:freqlocalization} gives the most accurate result. We also observe that increasing the level of discretization can increase BP's accuracy in locating the frequencies. 
\begin{figure}
\centering
\includegraphics[width=.9\textwidth, trim = 50mm 75mm 50mm 50mm, clip = true]{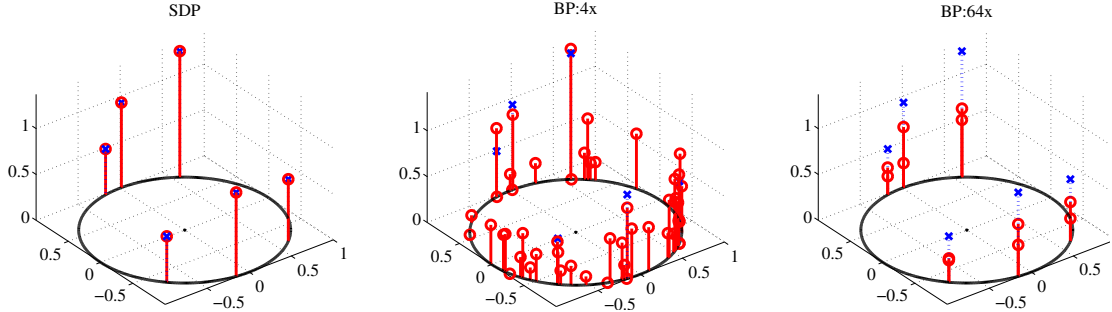}
\caption{ \small Frequency Estimation: Blue represents the true frequencies, while red represents the estimated ones.}
\label{fig:effectofgriding}
\end{figure} 

In the second set of experiments, we compare the performance of SDP and BP with three levels of discretization in terms of solution accuracy and running time. The parameter configurations are summarized in Table \ref{tbl:config}. Each configuration was repeated $10$ times, resulting a total of $1920$ valid experiments excluding those with $\rho_m \geq 1$. 
\begin{table}
\centering
  \caption{ \small Parameter configurations}
  \begin{tabular}{|l|l|}
  \hline
$n$ & 64, 128, 256\\
\hline
$\rho_s$ & 1/16, 1/32, 1/64\\
\hline
${\rho_m}/{\rho_s}$ & 5, 10, 20\\
\hline
$|c_k|$ & unit, fading\\
\hline
frequency & random, equispaced\\
\hline
sign & real, complex\\
\hline
  \end{tabular}
  \label{tbl:config}
\end{table}

We use the performance profile as a convenient way to compare the performance of different algorithms. The performance profile proposed in \cite{Dolan:2002wm} visually presents the performance of a set of algorithms under a variety of experimental conditions. More specifically, let $\mathcal{P}$ be the set of experiments and $\mathcal{M}_a(p)$ specify the performance of algorithm $a$ on experiment $p$ for some metric $\mathcal{M}$ (the smaller the better), e.g., running time and solution accuracy. Then the performance profile $\mathcal{P}_a(\beta)$ is defined as
\begin{equation*}
\mathcal{P}_a(\beta) = \frac{\#\{p \in \mathcal{P}: \mathcal{M}_a(p) \leq \beta \min_a \mathcal{M}_a(p)}{\#(\mathcal{P})}, \beta \geq 1.
\end{equation*}
Roughly speaking, $\mathcal{P}_a(\beta)$ is the fraction of experiments such that the performance of algorithm $a$ is within a factor $\beta$ of that of the best performed one.

We show the performance profiles for numerical accuracy and running times in Figure \ref{fig:profile_accuracy} and \ref{fig:profile_runningtime}, respectively. We see that SDP significantly outperforms BP for all tested discretization levels in terms of numerical accuracy. When the discretization levels are higher, e.g., 64x, the running times of BP exceed that of SDP. 

\begin{figure}
\centering
\begin{subfigure}[b]{0.4\textwidth}
\includegraphics[width=\textwidth, trim = 10mm 0mm 20mm 0mm, clip = true]{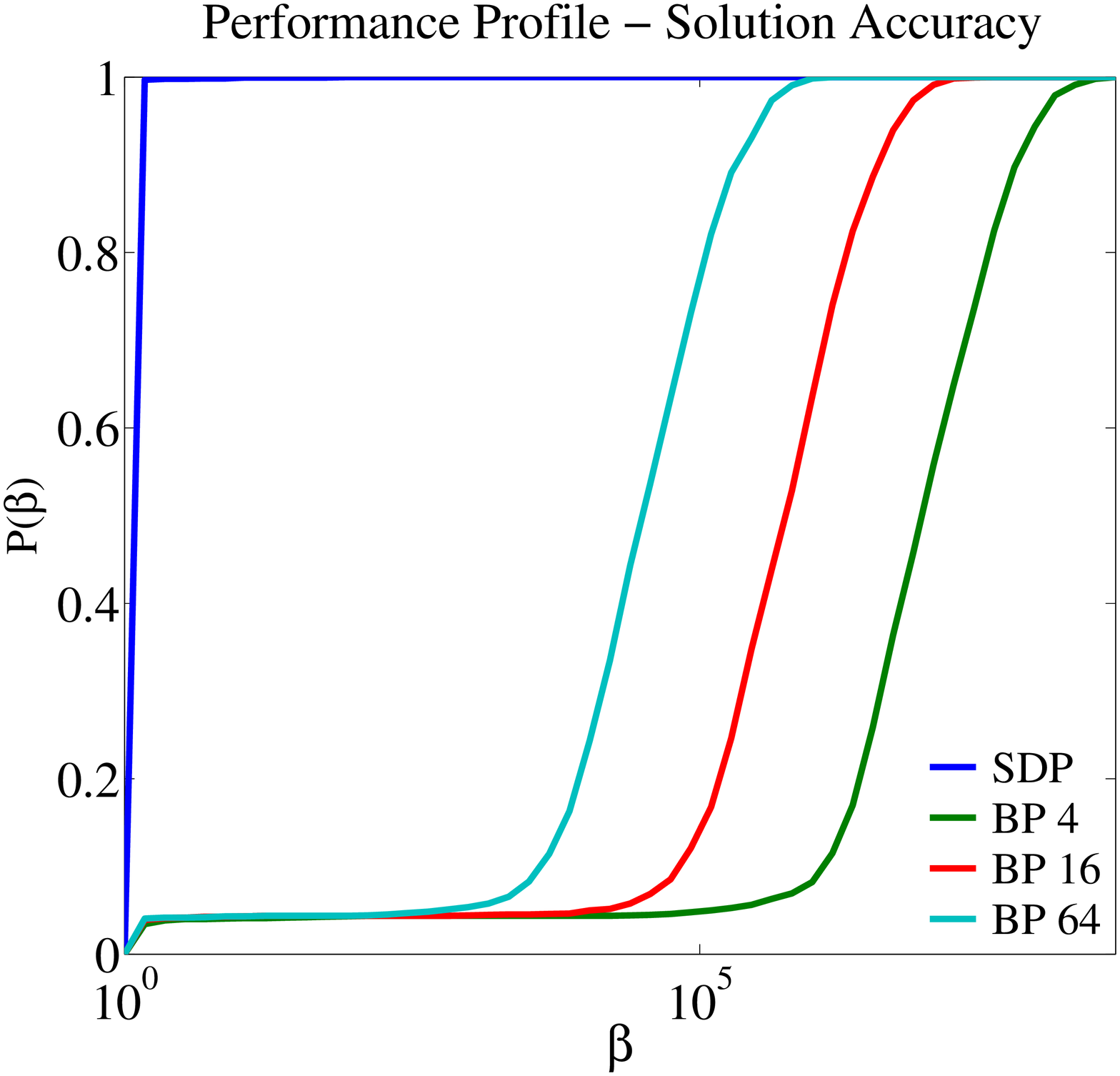}
\caption{ \small Solution accuracy}
\label{fig:profile_accuracy}
\end{subfigure}
  \begin{subfigure}[b]{0.4\textwidth}
\includegraphics[width=\textwidth, trim = 10mm 0mm 20mm 0mm, clip = true]{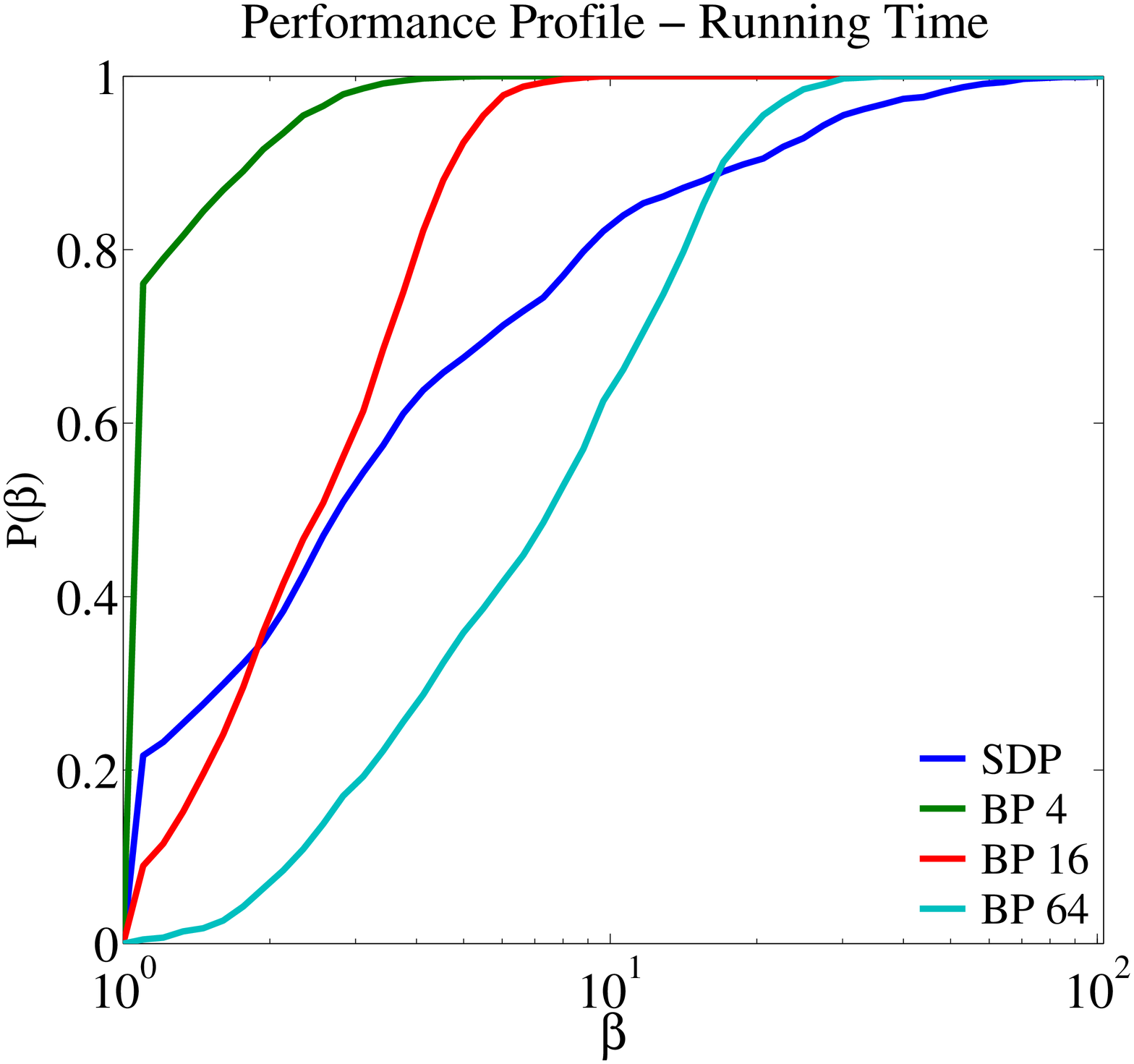}
\caption{ \small Running times}
\label{fig:profile_runningtime}
\end{subfigure}  
\caption{ \small Performance profiles for solution accuracy and running times. Note the $\beta$-axes are in logarithm scale for both plots.}
\end{figure}

To give the reader a better idea of the numerical accuracy and the running times,  in Table \ref{tbl:err} we present their medians and median absolute deviations for the four algorithms. As one would expect, the running time increases as the discretization level increases. We also observe that SDP is very accurate, with a median error at the order of $10^{-9}$. Increasing the level of discretization can increase the accuracy of BP. However, with discretization level $N = 64 n$, we get a median accuracy at the order of $10^{-5}$, but the median running time already exceeds that of SDP.
\begin{table}
\centering
  \caption{ \small Medians and median absolute deviation (MAD) for solution accuracy and running time}
  \begin{tabular}{|c|c|c|c|c|c|}
  \hline
 & & SDP & BP: 4x & BP: 16x & BP: 64x\\
 \hline
 \multirow{2}{*}{Solution Accuracy} & Median & 1.39e-09
 & 1.23e-02 &  7.67e-04
 & 4.65e-05\\
 \cline{2-6}
& MAD & 1.26e-09
 & 9.44e-03 &  6.05e-04
 & 3.64e-05\\
 \hline
 \multirow{2}{*}{Running Time (s)} & Median & 34.03 & 11.72 & 20.39 & 70.46\\
  \cline{2-6}
 & MAD & 27.32 & 4.83 & 12.19 & 55.37\\
\hline
  \end{tabular}
  \label{tbl:err}
\end{table}

In the third set of experiments, we compiled two phase transition plots. To prepare Figure \ref{fig:phasetran}, we pick $n = 128$ and vary $\rho_s = \frac{2}{n}:\frac{2}{n}:\frac{100}{n}$ and $\rho_m = \frac{2}{n}:\frac{2}{n}:\frac{126}{n}$. For each fixed $(\rho_m, \rho_s)$, we randomly generate $s = n\rho_s$ frequencies while maintaining a frequency separation $\Delta_f \geq \frac{1}{n}$. The coefficients are generated with random magnitudes and random phases, and the entries are observed uniform randomly. We then run the SDPT3-SDP algorithm to recover the missing entries. The recovery is considered successful if the relative error $\|\hat{x} - x^\star\|_2/\|x^\star\|_2 \leq 10^{-6}$. This process was repeated $10$ times and the rate of success was recorded. Figure \ref{fig:phasetran} shows the phase transition results. The $x$-axis indicates the fraction of observed entries $\rho_m$, while the $y$-axis is $\rho_s = \frac{s}{n}$. The color represents the rate of success with red corresponding to perfect recovery and blue corresponding to complete failure. 

We also plot the line $\rho_s = \rho_m/2$. Since a signal of $s$ frequencies has $2s$ degrees of freedom, including $s$ frequency locations and $s$ magnitudes, this line serves as the boundary above which any algorithm should have a chance to fail. In particular, Prony's method requires $2s$ consecutive samples in order to recover the frequencies and the magnitudes.

From Figure \ref{fig:phasetran}, we see that there is a transition from perfect recovery to complete failure. However, the transition boundary is not very sharp. In particular, we notice failures below the boundary of the transition where complete success should happen. Examination of the unsuccessful instances show that they correspond to instances with minimal frequency separations marginally exceeding $\frac{1}{n}$. We expect to get cleaner phase transitions if the frequency separation is increased.

To prepare Figure \ref{fig:phasetran2}, we repeated the same process in preparing Figure \ref{fig:phasetran} except that the frequency separation was increased from $\frac{1}{n}$ to $\frac{1.5}{n}$. In addition, to respect the minimal separation, we reduced the range of possible sparsity levels to $\{2, 4, \ldots, 70\}$. We now see a much sharper phase transition. The boundary is actually very close to the $\rho_s = \rho_m/2$ line. When $\rho_m$ is close to $1$, we even observe successful recovery above the line.
\begin{figure}
\begin{subfigure}[b]{0.5\textwidth}
\includegraphics[width=\textwidth, trim = 5mm 0mm 20mm 10mm, clip = true]{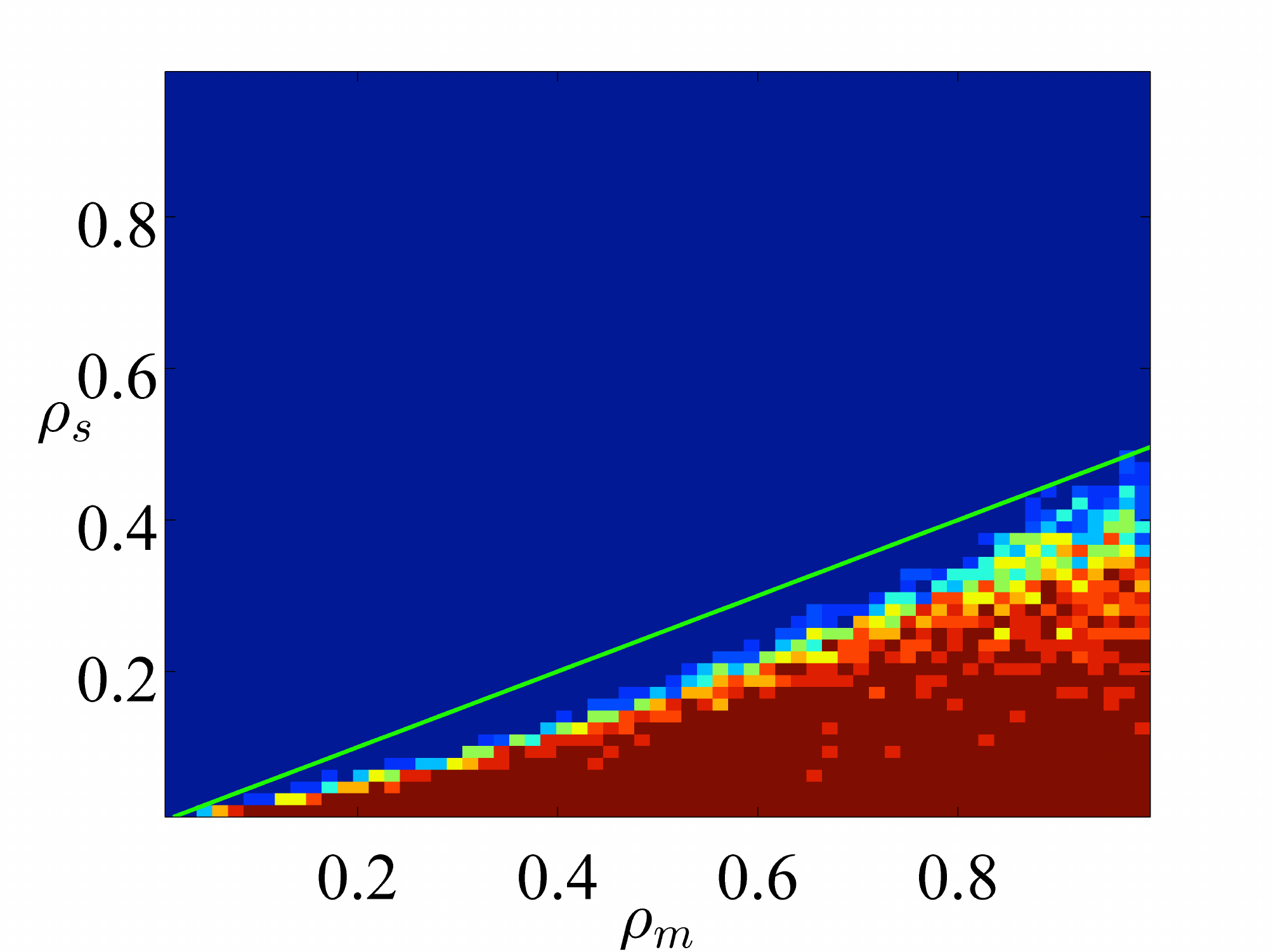}
\caption{\small Phase transition: $\Delta_f \geq \frac{1}{n}$}
\label{fig:phasetran}
\end{subfigure}
  \begin{subfigure}[b]{0.5\textwidth}
\includegraphics[width=\textwidth, trim = 5mm 0mm 20mm 10mm, clip = true]{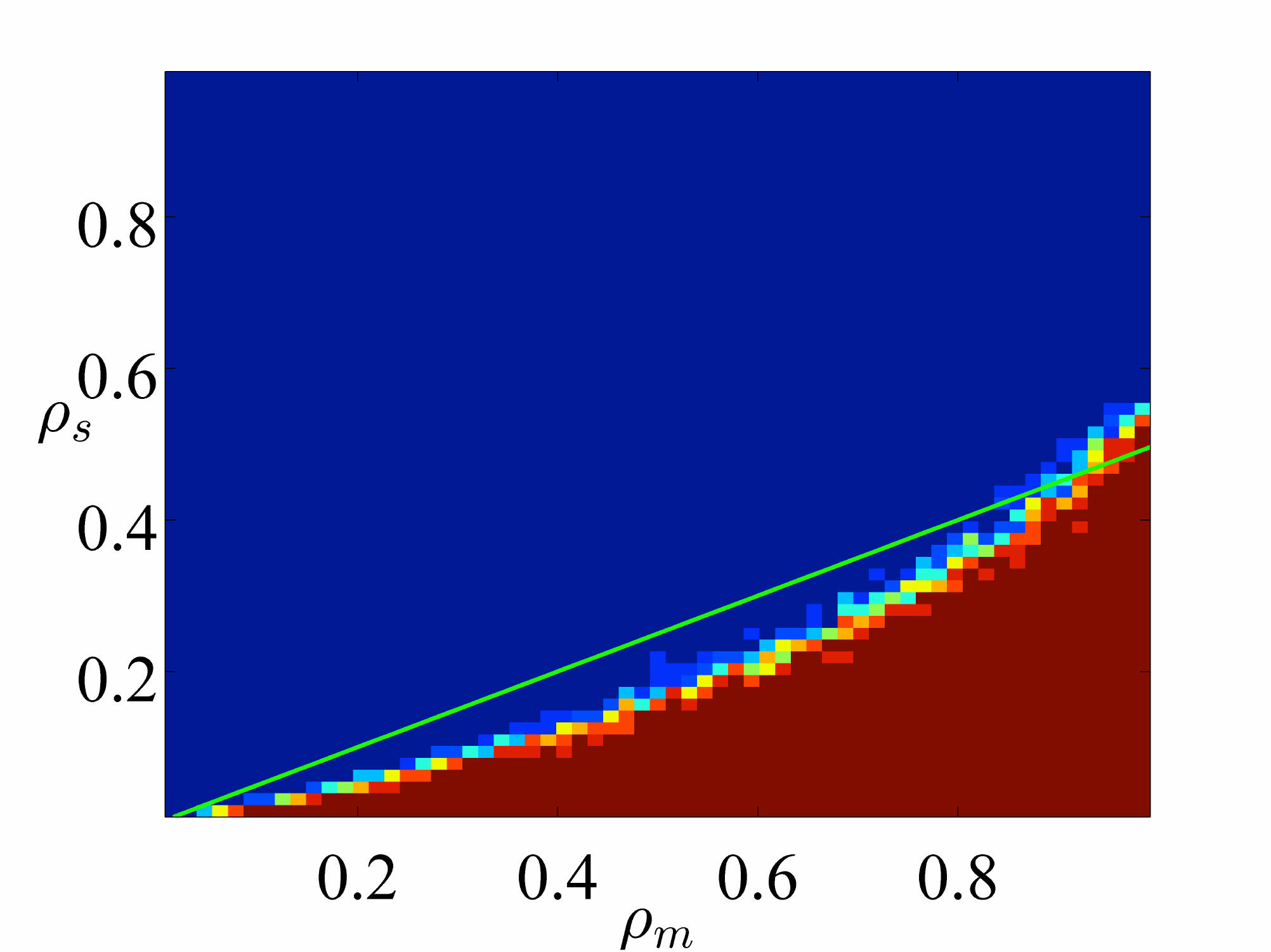}
\caption{\small Phase transition: $\Delta_f \geq \frac{1.5}{n}$}
\label{fig:phasetran2}
\end{subfigure}  
\caption{\small {\bf Phase transition:} The phase transition plots were prepared with $n = 128$, and $\rho_m = 2/n: 2/n: 126/n$. The frequencies were generated randomly with minimal separation $\Delta_f$. Both signs and magnitudes of the coefficients are random. In Figure \ref{fig:phasetran}, the separation $\Delta_f \geq 1/n$ and $\rho_s = 2/n: 2/n: 100/n$, while in Figure \ref{fig:phasetran2}, the separation $\Delta_f \geq 1.5/n$ and $\rho_s = 2/n: 2/n: 70/n$. }
\end{figure}  

\section{Conclusion and Future Work}
\label{sec:conclusion}

By leveraging the framework of atomic norm minimization, we were able to resolve the basis mismatch problem in compressed sensing of line spectra.   For signals with well-separated frequencies, we show the number of samples needed is roughly proportional to the number of frequencies, up to polylogarithmic factors.  This recovery is possible even though our continuous dictionary is not incoherent at all and does not satisfy any sort of restricted isometry conditions. 

There are several interesting future directions to be explored to further expand the scope of this work.  First, it would be useful to understand what happens in the presence of noise.  We cannot expect exact support recovery in this case, as our dictionary is continuous and any noise will make the exact frequencies un-identifiable.  In a similar vein, techniques like those used in~\cite{Candes:2012uf} that still rely on discretization are not applicable for our current setting.  However, since our numerical method is rather stable, we are encouraged that a theoretical stability result is possible.  

We show a simple example which provides some evidence for conjecturing that our proposed technique is stable. The signal $x^\star$ was generated with $n = 40$, $s = 3$, random frequencies, fading amplitudes, and random phases. A total number of $18$ uniform samples indexed by $T$ were taken.  The noisy observations $y$ was generated by adding complex noise $w$ with bounded $\ell_2$ norm $\varepsilon = 2$ to $x^\star_T$.  We denoised and recovered the signal by solving the following optimization:
{\setlength\abovedisplayskip{11pt plus 3pt minus 8pt}
\setlength\belowdisplayskip{11pt plus 3pt minus 8pt}
\begin{align}\label{eqn:noiseopt}
\minimize_x\ \|x\|_\A \ \st\ \|y - x_T\|_2 \leq \varepsilon,
\end{align}}
which can be solved by a semidefinite program. Figure \ref{fig:noise} illustrates the approximate frequency recovery achieved by the optimization problem \eqref{eqn:noiseopt} in presence of noise.
\begin{figure}
\centering
\begin{tabular}{cc}
\includegraphics[width=0.4\textwidth, trim = 0mm 5mm 10mm 8mm, clip = true]{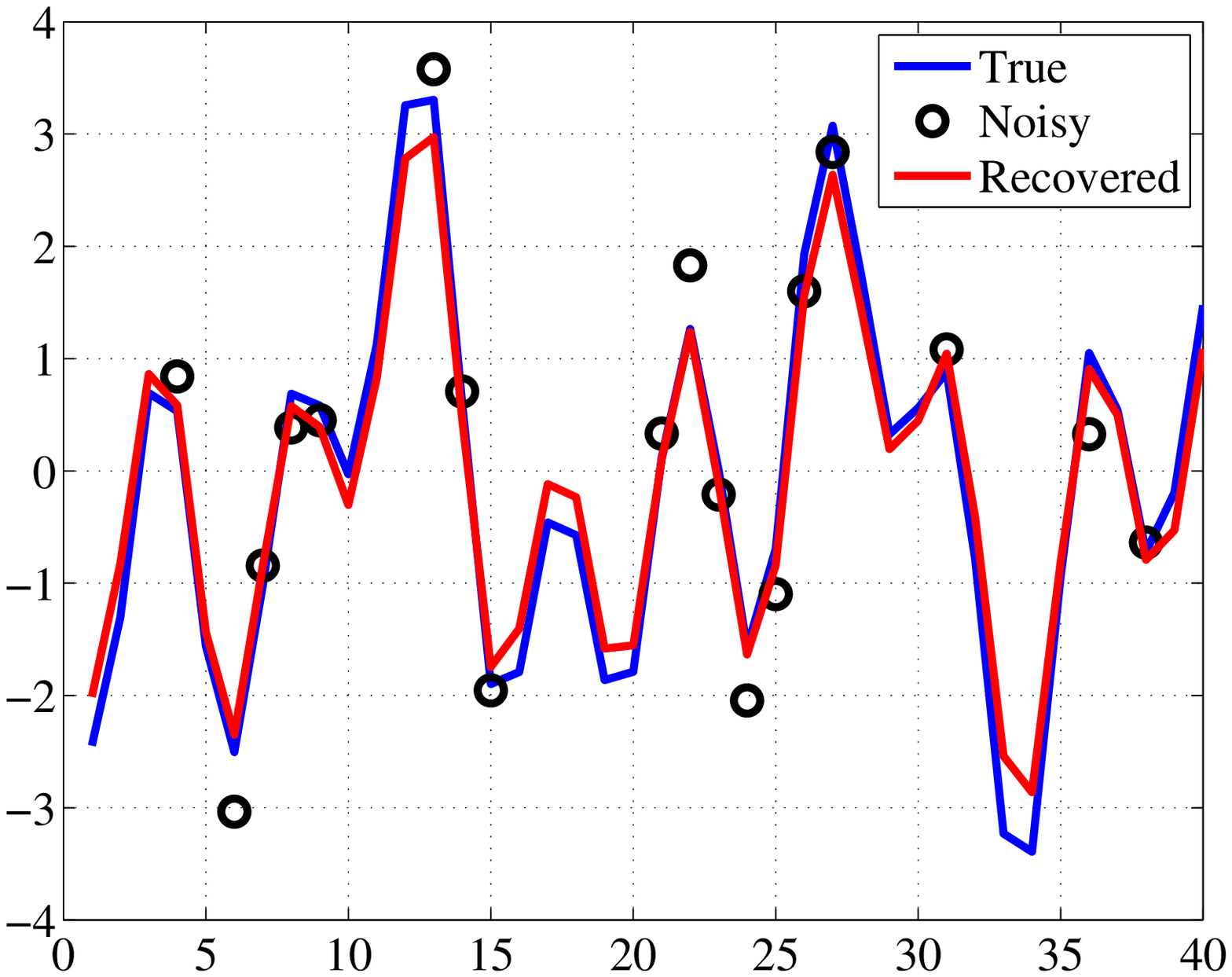}
& \includegraphics[width=0.4\textwidth, trim = 0mm 10mm 10mm 8mm, clip = true]{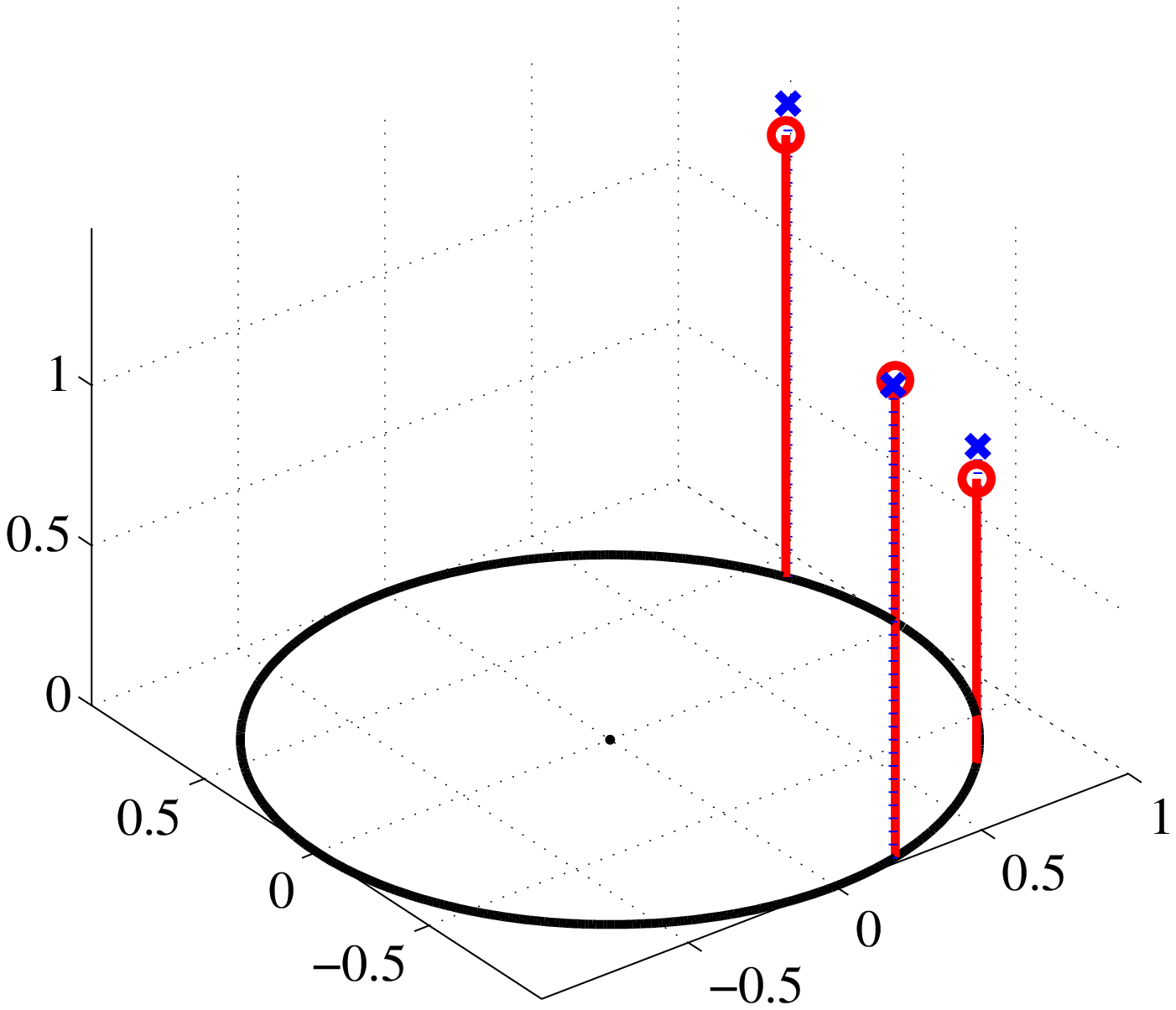}\\
(a) & (b)
\end{tabular}
\caption{\small Noisy frequency recovery: (a) Real part of true, noisy, and recovered signals, (b)True frequencies (blue) and recovered frequencies (red). The (randomly generated) true signal has three frequencies at $\{0.5874, 0.7528, 0.8966\}$ with coefficients $\{-0.1756-0.5306i, 0.1599-0.4344i, -1.0966+0.4216i\}$. The (real parts of) noisy observations at $m = 18$ out of $n = 40$ indices are indicated by black circles. The noise components were generated according to i.i.d. complex Gaussian distribution and then scaled to have an $\ell_2$ norm of $2$. The observed signal has an $\ell_2$ norm $5.9453$. }\label{fig:noise}
\end{figure}  
We leave the verification of stability with numerical simulations and theoretical analysis of this phenomenon to future work.

Second, we saw in our numerical experiments that modest discretization introduces substantial error in signal reconstruction and fine discretization carries significant computational burdens.  In this regard, it would be of great interest to speed up our semidefinite programming solvers so that we can scale our algorithms beyond the synthetic experiments of this paper.  Our rudimentary experimentation with first-order methods developed in~\cite{Bhaskar:2012tq} did not suffice for this problem as they were unable to achieve the precision necessary for fine frequency localization.  So, instead, it would be of interest to explore second order alternatives such as active set methods or the like to speed up our computations.

Finally, we are interested in exploring the class of signals that are semidefinite characterizable in hopes of understanding which signals can be exactly recovered.  Our continuous frequency model is an instance of applying compressed sensing to problems with continuous dictionaries.  It would be of great interest to see how our techniques may be extended to other continuously parametrized dictionaries.  Models involving image manifolds may fall into this category~\cite{Wakin09}.  Fully exploring the space of signals that can be acquired with just a few specially coded samples provides a fruitful and exciting program of future work.

\section*{Acknowledgements}
The authors would like to thank Robert Nowak for many helpful discussions about this work.   BR is generously supported by ONR award N00014-11-1-0723, NSF award CCF-1148243, and a Sloan Research Fellowship.  BB, BR, and PS are generously supported by NSF award CCF-1139953.  GT is generously supported by DARPA  Grant Number N66001-11-1-4090.

\begin{small}
\bibliographystyle{abbrv}
\bibliography{cs}
\end{small}

\appendix

\section{Proof of Theorem \ref{thm:main_general}}\label{apx:thm:main_general}
\begin{proof}
Assume $n = 4M + n_0$ with $M = \lfloor (n-1)/4 \rfloor$ and $n_0 = 1, 2, 3$ or $4$.
Suppose the signal $x^\star$ has decomposition
\begin{align*}
x^\star & = \sum_{k=1}^s c_k \begin{bmatrix}
1\\
e^{i2\pi f_k}\\
\vdots\\
e^{i2\pi (n-1) f_k} 
\end{bmatrix}\\
& = \sum_{k=1}^s \underbrace{c_k e^{i2\pi f_k (2M)}}_{\tilde{c}_k}\begin{bmatrix}
e^{i2\pi f_k (-2M)}\\
e^{i2\pi f_k (-2M + 1)}\\
\vdots\\
e^{i2\pi f_k (2M)} \\
\vdots\\
e^{i2\pi f_k (2M+n_0 -1)}
\end{bmatrix}
\end{align*}

The rest of the proof argues that the dual polynomial constructed for the symmetric case can be modified to certify the optimality of $x^\star$ for the general case. 

If the coefficients $\{c_k, k = 1, \ldots, s\}$ have uniform random complex signs, for fixed $\{f_k\}$, $\{\tilde{c}_k, k = 1, \ldots, s\}$ also have uniform random complex signs. In addition, the Bernoulli observation model $\{\delta_j\}_{j=0}^{n-1}$ on index set $\{0, \cdots, n-1\}$ naturally induces a Bernoulli observation model $\{\tilde{\delta}_j = \delta_{j+2M}\}_{j=-2M}^{2M}$ on $\{-2M, \cdots, 2M\}$ with $\mathbb{P}(\tilde{\delta}_j = 1) = m/n$. Denote $\tilde{T} = \{j: \tilde{\delta}_j = 1\} \subset \{-2M, \cdots, 2M\}$. Therefore, if $\Delta_f \geq \Delta_{\min} = 1/M$ and
  \begin{align}\label{eqn:mapx}
    m & \geq  C \max \left\{ \log^2 \frac{M}{\delta}, s \log \frac{s}{\delta}
    \log \frac{M}{\delta} \right\},
  \end{align}
 according to the proof of Theorem \ref{thm:main}, with probability great than $1-\delta$, we could construct a dual polynomial 
 \begin{align*}
    \tilde{Q} \left( f \right) & =  \sum_{j = - 2 M}^{2 M} \tilde{q}_j
    e^{- i 2 \pi jf} 
      \end{align*}
  satisfying
  \begin{align*}
    \tilde{Q} \left( f_k \right) & = \tmop{sign} \left( \tilde{c}_k \right), \forall
    f_k \in \Omega \\
    \left| \tilde{Q} \left( f \right) \right| & < 1, \forall f \notin \Omega \\
    \tilde{q}_j & = 0, \forall j \notin \tilde{T}.  
  \end{align*}
Now define 
\begin{align*}
q_j = \begin{cases}
\tilde{q}_{j - 2M} & {j = 0, \cdots, 4M}\\
0 & \mbox{otherwise}.
\end{cases}
\end{align*}
and 
\begin{align*}
Q(f) &=\sum_{j=0}^{n-1} q_j e^{-i2\pi f j}\\
& = \sum_{j=0}^{n-1} \tilde{q}_{j-2M} e^{-i2\pi f j}\\
& = e^{-i2\pi f (2M)} \tilde{Q}(f).
\end{align*}
Clearly, the polynomial $Q(f)$ satisfies
  \begin{align*}
    {Q} \left( f_k \right) & = e^{-i2\pi f_k (2M)} \tmop{sign} \left( \tilde{c}_k \right) = \tmop{sign}(c_k), \forall
    f_k \in \Omega  \\
    \left| {Q} \left( f \right) \right|& = |\tilde{Q}(f)| < 1, \forall f \notin \Omega \\
    {q}_j & = 0, \forall j \notin {T},  
  \end{align*} 
  where $T = \{j: \delta_j = 1\} \subset \{0, \ldots, n-1\}$. The theorem then follows from rewriting \eqref{eqn:mapx} in terms of $n$ and Proposition \ref{pro:optimality}.
\end{proof}

\section{Proof of Proposition \ref{pro:freqlocalization}}\label{apx:pro:freqlocalization}
We need the following lemma about maximal rank spectral factorization.
\begin{lemma}
  \label{lm:maxrank}Suppose $R \left( z \right) = \sum_{j = - \left( n - 1
  \right)}^{n - 1} r_j z^j$ with $r_{- j} = r_j^{\ast}$ is a nonnegative
  trigonometric polynomial, which has $s$ zeros on the unit circle. Then there
  exists a positive semidefinite matrix $G$ with rank $n - s$ such that
  \begin{align*}
    R \left( z \right) & =  \psi \left( z^{- 1} \right)^T G \psi \left( z
    \right)
  \end{align*}
  where
  \begin{align*}
    \psi \left( z \right) & =  \begin{bmatrix}
      1\\
      z\\
      \vdots\\
      z^{n - 1}
    \end{bmatrix}.
  \end{align*}
  Furthermore, this is the representation with maximal possible rank.
\end{lemma}

\begin{proof}
  According to the positive real lemma \cite[Theorem 2.5]{Dumitrescu:2007vw}, the nonnegativity of $R \left( z \right)$
  \begin{eqnarray*}
    &  & R \left( e^{i 2 \pi f} \right) \geq 0, \forall f \in \left[ 0, 1
    \right]
  \end{eqnarray*}
  is equivalent to
  \begin{eqnarray}
    &  & \exists \text{\ Hermitian matrix\ } G \in \mathbbm{C}^{n \times n}  , G \succeq 0, r_j =
    \sum_{l = 1}^{n - j} G_{l, l + j}, j = 0, \ldots, n - 1 
    \label{eqn:grammatrix}
  \end{eqnarray}
  In this case, we can write
  \begin{eqnarray*}
    R \left( z \right) & = & \psi \left( z^{- 1} \right)^T G \psi \left( z
    \right).
  \end{eqnarray*}
Such a representation is called a Gram matrix representation. The set
  $\mathcal{C}$ of all Gram matrices satisfying (\ref{eqn:grammatrix}) is a
  compact convex set. By the spectral factorization theorem \cite[Theorem 1.1]{Dumitrescu:2007vw}, all of the rank
  one elements of $\mathcal{C}$ are of the form $G = g g^{\ast}$ such that:
  \begin{eqnarray*}
     g^*\psi(z) & = & d \prod_{k = 1}^s \left( z - e^{i 2
    \pi f_k} \right) \prod_{k = s + 1}^{n - 1} \left( z - z_k \right)
  \end{eqnarray*}
  with $d$ the scalar factor determined by $R \left( z \right) = \psi
  \left( z^{- 1} \right)^T g g^{\ast} \psi \left( z \right)$. Here $\left\{
  e^{i 2 \pi f_k}, k = 1, \ldots, s \right\}$ are the $s \geq 0$ zeros of $R
  \left( z \right)$, as a function defined on the complex plane, on the unit
  circle with multiplicity two. The rest $\left( 2 \left( n - 1 \right) - 2 s
  \right)$ zeros of $R \left( z \right)$ are divided into $n - 1 - s$
  conjugate pairs $\left\{ \zeta_{k}, 1 / \zeta_k^{\ast}, k = s + 1, \ldots, n
  - 1 \right\}$. Choosing $z_k$ as either $\zeta_k$ or $1 / \zeta_k^{\ast}$
  yields different $g$. Therefore, there are $2^{n - 1 - s}$ total such rank-one representations of $R \left( z \right)$.
  
  Since $R \left( e^{i 2 \pi f} \right)$ has $s$ zeros on $\left[ 0, 1
  \right]$, we have for any $G \in \mathcal{C}$, $\psi \left( e^{i 2 \pi f_k}
  \right)^{\ast} G \psi \left( e^{i 2 \pi f_k} \right) = 0$, implying $G \psi
  \left( e^{i 2 \pi f_k} \right) = 0, k = 1, \ldots, s$. Therefore, the null
  space of $G$ has dimension at least $s$, and hence the rank of any Gram
  matrix $G$ is at most $n - s$. 
  
  In the following we construct a maximal rank
  representation of $R \left( z \right)$. For that purpose,
  define $g_l$ for $l = s + 1, \ldots, n$ via
  \begin{eqnarray*}
 g_l^*\psi(z) & = & d_l \left[ \prod_{k = 1}^s \left(
    z - e^{i 2 \pi f_k} \right) \right] \left[ \prod_{k = s + 1, k \neq l}^{n
    - 1} \left( z - \zeta_k \right) \right] \left( z - 1 / \zeta_l^{\ast}
    \right), l = s + 1, \ldots, n - 1\\
g_n^*\psi(z) & = & d_n \left[ \prod_{k = 1}^s \left(
    z - e^{i 2 \pi f_k} \right) \right] \left[ \prod_{k = s + 1}^{n - 1}
    \left( z - \zeta_k \right) \right]
  \end{eqnarray*}
  We first claim that $\left\{ g_l, l = s + 1, \ldots, n \right\}$ are
  linearly independent. Suppose otherwise, there exist constants $\alpha_j$
  with at least one $\alpha_j \neq 0$, such that
  \begin{eqnarray*}
    \sum_{j = s + 1}^n \alpha_j g_j^* & = & 0
  \end{eqnarray*}
  Post-multiplying both sides of the above equation by $\psi \left( \zeta_l
  \right)$ for each $s + 1 \leq l \leq n - 1,$we get
  \begin{eqnarray*}
    \sum_{j = s + 1}^n \alpha_j  g_j^*\psi \left( \zeta_l \right) & = &
    \alpha_l d_l \left[ \prod_{k = 1}^s \left( \zeta_l - e^{i 2 \pi f_k}
    \right) \right] \left[ \prod_{k = s + 1, k \neq l}^{n - 1} \left( \zeta_l
    - \zeta_k \right) \right] \left( \zeta_l - 1 / \zeta_l^{\ast} \right)
    = 0
  \end{eqnarray*}
  implying $\alpha_l = 0$ for $s + 1 \leq l \leq n - 1$. Since $g_n \neq 0$,
  $\alpha_n$ must also be zero. Therefore, the set $\left\{ g_l, l = s + 1,
  \ldots, n \right\}$ is linearly independent. Now form the matrix
  \begin{eqnarray*}
    G & = & \sum_{j = s + 1}^n \lambda_j g_j g_j^{\ast} \succeq 0
  \end{eqnarray*}
  where $\lambda_j > 0$, and $\sum_j \lambda_j = 1$. Then the linear
  independence of $\{g_l\}$ implies that $G$ is of rank $n - s$, and
  \begin{eqnarray*}
    \psi \left( z^{- 1} \right)^T G \psi \left( z \right) & = & \sum_{j = s +
    1}^n \lambda_j \psi \left( z^{- 1} \right)^T g_j g_j^{\ast} \psi \left( z
    \right)\\
    & = & \sum_{j = s + 1}^n \lambda_j R \left( z \right)\\
    & = & R \left( z \right) .
  \end{eqnarray*}
\end{proof}

\begin{proof}[Proof of Proposition \ref{pro:freqlocalization}]
We prove for the case $J = \left\{ 0, \cdots, n - 1 \right\}$. Note that our semidefinite program \eqref{eqn:tracemin} and its dual have interior feasible points, a common assumption underlying many results about semidefinite programs, including the ones cited in the following arguments. For example, to obtain an interior feasible point for \eqref{eqn:tracemin}, one could take $x = x^\star$ on $T$ and zero elsewhere, the Toeplitz matrix equal to $(\|x\|_2^2+1) I $,  and $t = 1$.

  1) For the
  dual optimal solution $q$ given in the proposition, define
  \begin{eqnarray*}
    R \left( z \right) & : = & 1 - \psi \left( z^{- 1} \right)^T q q^{\ast}
    \psi \left( z \right).
  \end{eqnarray*}
  The trigonometric polynomial $R \left( e^{i 2 \pi f} \right)$ is positive
  with $s$ zeros at $f = f_k$ since
  \begin{eqnarray*}
    R \left( e^{i 2 \pi f} \right) & = & 1 - \psi \left( e^{i 2 \pi f}
    \right)^{\ast} q q^{\ast} \psi \left( e^{i 2 \pi f} \right)\\
    & = & 1 - \left| Q \left( f \right) \right|^2  \left\{ \begin{array}{ll}
      = 0 & \tmop{if} f = f_k\\
      > 0 & \tmop{otherwise}
    \end{array} \right.
  \end{eqnarray*}
  Now according to Lemma \ref{lm:maxrank}, there exists $G \succeq 0$ with
  $\tmop{rank} \left( G \right) = n - s$ such that
  \begin{eqnarray*}
    R \left( z \right) & = & 1 - \psi \left( z^{- 1} \right)^T q q^{\ast} \psi
    \left( z \right) = \psi \left( z^{- 1} \right)^T G \psi \left( z \right)
  \end{eqnarray*}
  Define $H = G + q q^{\ast} \succeq q q^{\ast}$ which satisfies $\psi \left(
  z^{- 1} \right)^T H \psi \left( z \right) \equiv 1$. In other words, we have
  a pair $\left( q, H \right)$ such that
  \begin{eqnarray*}
    \bar{H} :=\begin{bmatrix}
          H & - q\\
      - q^{\ast} & 1
    \end{bmatrix} & \succeq & 0\\
    \sum_{j = 1}^{n - l} H_{j, j + l} & = & 0, l = 0, \ldots, n - 1\\
    q_{T^c} & = & 0
  \end{eqnarray*}
  and $\tmop{rank} \left( \bar{H} \right) = n - s + 1$. Therefore, $\left( q,
  H \right)$ is a dual optimal solution. 
  
  The unique primal optimal solution
  has the form (refer to Proposition \ref{pro:optimality} for a proof of the uniqueness):
  \begin{eqnarray*}
    \hat{x} & = & x^{\star} = \sum_{k = 1}^s \left| c_k \right| e^{i \phi_k} a
    \left( f_k, 0 \right)\\
    \tmop{Teop} \left( \hat{u} \right) & = & \sum_{k = 1}^s \left| c_k \right|
    a \left( f_k, 0 \right) a \left( f_k, 0 \right)^{\ast}\\
    \hat{t} & = & \sum_{k = 1}^s \left| c_k \right|
  \end{eqnarray*}
  where $e^{i\phi_k} = \sign{(c_k)}$. 
  Since the following decomposition holds
  \begin{eqnarray*}
    \bar{T} := \begin{bmatrix}
      \tmop{Toep} \left( \hat{u} \right) & \hat{x}\\
      \hat{x}^{\ast} & \hat{t}
    \end{bmatrix} & = & \sum_{k = 1}^s \left| c_k \right|
    \begin{bmatrix}
      a \left( f_k, 0 \right)\\
      e^{i \phi_k}
    \end{bmatrix} \begin{bmatrix}
      a \left( f_k, 0 \right)\\
      e^{i \phi_k}
    \end{bmatrix}^{\ast},
  \end{eqnarray*}
  the rank of $\bar{T}$ is $s$, and hence $\left( \bar{H}, \bar{T} \right)$ is
  a strictly complementary pair (See for \cite{Alizadeh1997Complementary} for more information about
  strict complementarity).
  
 According to \cite[Lemma 3.1]{Goldfarb1998}, all matrices in the relative interior of the optimal solution
  set of a semidefinite program share the same range space. Therefore, for
  any other $( \tilde{q}, \tilde{H} )$ in the relative interior of
  $\mathcal{D}$, the rank of $\begin{bmatrix}
    \tilde{H} & - \tilde{q}\\
    - \tilde{q}^{\ast} & 1
  \end{bmatrix}$ is also $n + 1 - s$, implying that $(\tilde{q},
  \tilde{H})$ also form strictly complementary pairs with the primal optimal
  solution $\left( \hat{x}, \hat{u}, \hat{t} \right)$.
  
  2) \ Since for the $\left( q, H \right)$ in part 1)
  \begin{eqnarray*}
    &  & \begin{bmatrix}
      a \left( f_k, 0 \right)\\
      e^{i \phi_k}
    \end{bmatrix}^{\ast} \begin{bmatrix}
      H & - q\\
      - q^{\ast} & 1
    \end{bmatrix} \begin{bmatrix}
          a \left( f_k, 0 \right)\\
      e^{i \phi_k}
    \end{bmatrix}\\
    & = & a \left( f_k, 0 \right)^{\ast} H a \left( f_k, 0 \right) - a \left(
    f_k, 0 \right)^{\ast} q e^{i \phi_k} - e^{- i \phi_k} q^{\ast} a \left(
    f_k, 0 \right) + 1\\
    & = & 0, k = 1, \ldots, s,
  \end{eqnarray*}
the null space of $\begin{bmatrix}
    H & - q\\
    - q^{\ast} & 1
  \end{bmatrix}$ is
  \begin{eqnarray}
    \tmop{span} \left\{ \begin{bmatrix}
      a \left( f_k, 0 \right)\\
      e^{i \phi_k}
    \end{bmatrix} : k = 1, \ldots, s \right\} &  &  \label{eqn:span}
  \end{eqnarray}
  Again using \cite[Lemma 3.1]{Goldfarb1998}, we conclude that the null space of any
  $\begin{bmatrix}
    \tilde{H} & - \tilde{q}\\
    - \tilde{q}^{\ast} & 1
  \end{bmatrix}$ with $( \tilde{q}, \tilde{H} )$ in the
  relative interior of $\mathcal{D}$ is given by (\ref{eqn:span}). Therefore,
  we have for each $k$
  \begin{eqnarray}
    &  & \begin{bmatrix}
      a \left( f_k, 0 \right)\\
      e^{i \phi_k}
    \end{bmatrix}^{\ast} \begin{bmatrix}
      \tilde{H} & - \tilde{q}\\
      - \tilde{q}^{\ast} & 1
    \end{bmatrix} \begin{bmatrix}
      a \left( f_k, 0 \right)\\
      e^{i \phi_k}
    \end{bmatrix} \nonumber\\
    & = & a \left( f_k, 0 \right)^{\ast} \tilde{H} a \left( f_k, 0 \right) -
    a \left( f_k, 0 \right)^{\ast} \tilde{q} e^{- i \phi_k} - e^{i \phi_k}
    \tilde{q}^{\ast} a \left( f_k, 0 \right) + 1 \nonumber\\
    & = & 2 - 2 \operatorname*{Re} \left( e^{- i \phi_k} \left\langle \tilde{q}, a \left(
    f_k, 0 \right) \right\rangle \right) = 0  \label{eqn:eat}
  \end{eqnarray}
  or equivalently
  \begin{align*}
    \left\langle \tilde{q}, a \left( f_k, 0 \right) \right\rangle & =  e^{i
    \phi_k} = \tmop{sign} \left( c_k \right), k = 1, \ldots, s
  \end{align*}
  using $\left| \left\langle \tilde{q}, a \left( f_k, 0 \right) \right\rangle
  \right| \leq 1$.
  
  For any $\tilde{f} \neq f_k, k = 1, \ldots, s$, if $| \langle
  \tilde{q}, a ( \tilde{f}, 0 ) \rangle | = 1$, then a
  calculation similar to (\ref{eqn:eat}) concludes that
  $\begin{bmatrix}
    a ( \tilde{f}, 0 )\\
    e^{i \tilde{\phi}}
  \end{bmatrix}$ with $e^{i \tilde{\phi}} = \tmop{sign} (
  \langle \tilde{q}, a ( \tilde{f}, 0 ) \rangle )$
  is in the null space of $\begin{bmatrix}
    \tilde{H} & - \tilde{q}\\
    - \tilde{q}^{\ast} & 1
  \end{bmatrix}$. Due to the linear independence of $a ( \tilde{f},
  0 )$ with $\left\{ a \left( f_k, 0 \right), k = 1, \ldots, s
  \right\}$, the null space have dimension greater than $s$, contradicting with
  $\tmop{rank} \left( \begin{bmatrix}
    \tilde{H} & - \tilde{q}\\
    - \tilde{q}^{\ast} & 1
  \end{bmatrix} \right) = n + 1 - s$. Therefore, we have
  \begin{eqnarray*}
    \left| \left\langle \tilde{q}, a \left( f, 0 \right) \right\rangle \right|
    & < & 1, \forall f \neq f_k, k = 1, \ldots, s.
  \end{eqnarray*}
  3) Due to the existence of strictly complementary primal-dual optimal
  solutions, the statement is a direct consequence of \cite[Lemma 3.4]{Luo1998} (See also \cite{deKlerk1997, Halicka2002}). 
 \end{proof}

\section{Proof of Proposition \ref{pro:sysmtx:deterministic}}\label{apx:sysmtx:deterministic}
\begin{proof}
  Under the assumption that $\Delta_{\min} \geq \frac{1}{M}$, we cite the results of \cite[Proof of Lemma 2.2]{Candes:2012uf} as follows:
  \begin{align*}
\left\| I - \bar{D}_0
    \right\|_{\infty} & \leq 6.253 \times 10^{- 3}\\
     \Big\| \frac{1}{\sqrt{\left| \bar{K}_M'' \left( 0
    \right) \right|}} \bar{D}_1 \Big\|_{\infty} & \leq 4.212 \times 10^{-
    2}\\
    \Big\| I - \Big( - \frac{1}{\left|
    \bar{K}_M'' \left( 0 \right) \right|} \bar{D}_2 \Big) \Big\|_{\infty} &
    \leq 0.3201,
  \end{align*}
  where $\|\cdot\|_\infty$ is the matrix infinity norm, namely, the maximum absolute row sum. 
  Since $I - \bar{D}$ is symmetric and has zero diagonals, the Ger\v{s}hgorin
  circle theorem \cite{Horn:1990tf} implies that
  \begin{align*}
    \left\| I - \bar{D} \right\| & \leq \left\| I - \bar{D}
    \right\|_{\infty}\\
    &  & \\
    & \leq \max \Big\{ \left\| I - \bar{D}_0 \right\|_{\infty} + \Big\|
    \frac{1}{\sqrt{\left| \bar{K}_M'' \left( 0 \right) \right|}} \bar{D}_1
    \Big\|_{\infty}, \nn\\
    & \ \ \ \ \ \ \ \ \ \ \ \ \Big\| \frac{1}{\sqrt{\left| \bar{K}_M'' \left( 0
    \right) \right|}} \bar{D}_1 \Big\|_{\infty} + \Big\| I - \Big( -
    \frac{1}{\left| \bar{K}_M'' \left( 0 \right) \right|} \bar{D}_2 \Big)
    \Big\|_{\infty} \Big\}\\
    & = 0.3623.
  \end{align*}
  As a consequence, $\bar{D}$ is invertible and
  \begin{align*}
  \|\bar{D}\| &\leq 1 + \|I - \bar{D}\| \leq 1.3623,\nn\\
    \left\| \bar{D}^{- 1} \right\| & \leq \frac{1}{1 - \left\| I - \bar{D}
    \right\|} \leq 1.568.
  \end{align*}
  
\end{proof}

\section{Proof of Lemma \ref{lm:invertibility}}\label{apx:lm:invertibility}
\begin{proof}[Proof of Lemma \ref{lm:invertibility}]
We start with computing the quantities necessary to apply Lemma \ref{lm:matrixbernstein}:
  \begin{align*}
    \mathbbm{E}X_j & = 0\\
    \left\| X_j \right\| & = \left\| \frac{1}{M} g_M \left( j
    \right) \left( \delta_j - p \right) e \left( j \right) e \left( j
    \right)^{\ast} \right\|\nn\\
    & \leq \frac{1}{M} \left\| g_M \right\|_{\infty} s \left( 1 +
    \max_{\left| j \right| \leq 2 M} \frac{\left( 2 \pi j \right)^2}{\left|
    K_M'' \left( 0 \right) \right|} \right)\nn\\
    & \leq R := 14 \frac{s}{M}
    \tmop{for} M \geq 4.
    \end{align*}
    Here we have used
    \begin{align*}
    \|g_M\|_\infty &\leq 1,\\
    \|e(j)\|_2^2 &= s \left( 1 +
    \max_{\left| j \right| \leq 2 M} \frac{\left( 2 \pi j \right)^2}{\left|
    K_M'' \left( 0 \right) \right|} \right) \leq 14 s, \text{\ for \ } M \geq 4.
    \end{align*}
    We continue with $\sigma^2$:
    \begin{align*}
    \sigma^2 & = \left\| \sum_{j=-2M}^{2M} \mathbbm{E} \left( \frac{1}{M^2} g_M^2
    \left( j \right) \left( \delta_j - p \right)^2 \|e(j)\|_2^2  e \left( j \right) e^{\ast} \left( j \right) \right)
    \right\|\nn\\
    & \leq 14 \frac{p \left( 1 - p \right)}{M} s \left\| \frac{1}{M} \sum_{j = -2M}^{2M}
    g_M^2 \left( j \right) e \left( j \right) e^{\ast} \left( j
    \right) \right\|.
  \end{align*}
  To further bound $\sigma^2$, we note
  \begin{align*}
    &\ \ \ \ \ \frac{1}{M} \sum_{j = -2M}^{2M} g_M^2 \left( j \right) e \left( j
    \right) e^{\ast} \left( j \right)\nn\\
    & \preccurlyeq  \left\| g_M \right\|_{\infty} \Big\{ \frac{1}{M} \sum_{j = -2M}^{2M}
    g_M \left( j \right) e \left( j \right) e^{\ast} \left( j
    \right) \Big\}\nn\\
    & =  \left\| g_M \right\|_{\infty} \bar{D},
  \end{align*}
  which leads to
  \begin{align*}
    & \ \ \ \Big\| \frac{1}{M} \sum_j g_M^2 \left( j \right) e \left(
    j \right) e^{\ast} \left( j \Big) \right\|\nn\\
    & = \lambda_{\max} \Big( \frac{1}{M} \sum_j g_M^2 \left( j
    \right) e \left( j \right) e^{\ast} \left( j \right) \Big)\nn\\
    & \leq \lambda_{\max} \left( \left\| g_M \right\|_{\infty} \bar{D}
    \right)\nn\\
    & = \left\| g_M \right\|_{\infty} \left\| \bar{D} \right\|\nn\\
    & \leq 1.3623 \text{\ by \ } \eqref{eqn:Dbarbd} \text{\ and\ } \eqref{eqn:gbd}.
  \end{align*}
Therefore, we have
  \begin{align*}
    \sigma^2 & \leq 20 \frac{p}{M} s .
  \end{align*}

  Invoking the non-commutative Bernstein's inequality and setting $t = p
  \tau$, we have
 \begin{align*}
    \mathbbm{P} \left( \left\| p^{- 1} D - \bar{D} \right\| \geq \tau \right)
    & \leq 2 s \exp \left( \frac{- p^2 \tau^2 / 2}{20  \frac{p}{M} s + 14
    \frac{s}{M} p \tau / 3} \right)\nn\\
    & \leq 2 s \exp \left( - \frac{1}{50 }
    \tau^2 \frac{m}{s} \right)  \left( \tmop{used} \tau \leq 1 \right)\nn\\
    & \leq \delta .
  \end{align*}
  if
  \begin{align*}
    m & \geq \frac{50}{\tau^2} s \log \frac{2
    s}{\delta} .
  \end{align*}
Consequently, when $\tau < 1 - 0.3623 \leq 1 - \left\| I - \bar{D}
  \right\|$ according to \eqref{eqn:I_Dbarbd}, we have $\left\| I - p^{- 1} D \right\| \leq \left\| I - \bar{D}
  \right\| + \left\| p^{- 1} D - \bar{D} \right\| < 1$, confirming the
  invertibility of $p^{- 1} D$. 
\end{proof}

\section{Proof of Corollary \ref{cor:bdonmatrix}}\label{sec:bdonmatrix}
Assuming $B$ is invertible and $\left\| A - B \right\| \left\| B^{- 1} \right\|
  \leq \frac{1}{2}$, we have the following two inequalities:
  \begin{align*}
    \left\| A^{- 1} \right\| & \leq \frac{\left\| B^{- 1} \right\|}{1 -
    \left\| A - B \right\| \left\| B^{- 1} \right\|} \leq 2 \left\| B^{- 1}
    \right\|\\
    \left\| A^{- 1} - B^{- 1} \right\| & \leq \frac{\left\| A - B \right\|
    \left\| B^{- 1} \right\|^2}{1 - \left\| A - B \right\| \left\| B^{- 1}
    \right\|} \leq 2 \left\| B^{- 1} \right\|^2 \left\| A - B \right\|,
  \end{align*}
  which are rearrangements of
  \begin{align*}
    \left\| A^{- 1} - B^{- 1} \right\| & \leq \left\| A^{- 1} \right\|
    \left\| A - B \right\| \left\| B^{- 1} \right\|\\
    \left\| A^{- 1} \right\| & \leq \left\| A^{- 1} - B^{- 1} \right\| +
    \left\| B^{- 1} \right\|\\
    & \leq \left\| A^{- 1} \right\| \left\| A - B \right\| \left\| B^{- 1}
    \right\| + \left\| B^{- 1} \right\| .
  \end{align*}
  Therefore, we establish that when $\tau \leq \frac{1}{4} < \frac{1}{2
  \left\| \bar{D}^{- 1} \right\|}$ on the set $\mathcal{E}_{1, \tau}$:
  \begin{align*}
    \left\| D^{- 1} - p^{- 1} \bar{D}^{- 1} \right\| & \leq 2 \left\| p^{-
    1} \bar{D}^{- 1} \right\|^2 \left\| D - p \bar{D} \right\| = 2 \left\|
    \bar{D}^{- 1} \right\|^2 p^{- 1} \tau\\
    \left\| D^{- 1} \right\| & \leq 2 \left\| p^{- 1} \bar{D}^{- 1} \right\|
    = 2 \left\| \bar{D}^{- 1} \right\| p^{- 1} .
  \end{align*}
  Since the operator norm of a matrix dominates that of all submatrices, this completes the proof.

\section{Proof of Lemma \ref{lm:vw}}\label{apx:lm:vw}
The proof uses Talagrand's concentration of measure inequality:
\begin{lemma}[{\cite[Corollary 7.8]{LeDoux:2001tx}}]\label{lm:talagrand}
  \label{lm:concentration}Let $\left\{ Y_j \right\}$ be a finite sequence of
  independent random variables taking values in a Banach space and let $V$ be
  defined as
  \begin{align*}
    V & = \sup_{h \in \mathcal{H}} \sum_j h \left( Y_j \right)
  \end{align*}
  for a countable family of real valued functions $\mathcal{H}$. Assume that
  $\left| h \right| \leq B$ and $\mathbbm{E}h \left( Y_j \right) = 0$ for all
  $h \in \mathcal{H}$ and every $j$. Then for all $t > 0,$
  \begin{align*}
    \mathbbm{P} \left( \left| V -\mathbbm{E}V \right| > t \right) & \leq 16
    \exp \left( - \frac{t}{K B} \log \left( 1 + \frac{B t}{\sigma^2 +
    B\mathbbm{E} \bar{V}} \right) \right),
  \end{align*}
  where $\sigma^2 = \sup_{h \in \mathcal{H}} \sum_j \mathbbm{E}h^2 (Y_j)$, $\bar{V} = \sup_{h \in \mathcal{H}} \left| \sum_j h \left( Y_j
  \right) \right|$, and $K$ is a numerical constant.
\end{lemma}

\begin{proof}[Proof of Lemma \ref{lm:vw}]
  Based on the definition of $v_\ell(f)$ in \eqref{eqn:vl} and $\bar{v}_\ell(f)$ in \eqref{eqn:vlbar}, we explicitly write $v_{\ell} \left( f \right) - p\bar{v}_\ell \left( f \right) = v_{\ell} (f) - \mathbbm{E} v_\ell(f)$ as 
  \begin{align*}
    v_{\ell} \left( f \right) -p\bar{v}_\ell \left( f \right) & = 
    \sum_{j = - 2 M}^{2 M} \frac{1}{M} \Bigg( \frac{i 2 \pi j}{\sqrt{\left|
    \bar{K}_M'' \left( 0 \right) \right|}} \Bigg)^{\ell} g_M \left( j
    \right) \left( \delta_j - p \right) e^{i 2 \pi f j}e(j)\\    & = \sum_{j = - 2 M}^{2 M} Y_j^{\ell},
  \end{align*}
  where $e(j)$ is defined in \eqref{eqn:ej} and we have defined $Y_j^{\ell}$ as
  \begin{align*}
  Y_j^{\ell} = \frac{1}{M} \Bigg( \frac{i 2 \pi j}{\sqrt{\left|
    \bar{K}_M'' \left( 0 \right) \right|}} \Bigg)^{\ell} g_M \left( j
    \right) \left( \delta_j - p \right) e^{i 2 \pi f j}e(j). \end{align*}
It is clear that $\{Y_j^\ell\}_{j=-2M}^{2M}$ are independent random vectors with zero mean. 

Define
  \begin{align*}
    V^{\ell} & :=  \left\| v_{\ell} \left( f \right) -p\bar{v}_\ell
    \left( f \right) \right\|_2 = \sup_{h : \left\| h \right\|_2 = 1}
    \left\langle v_{\ell} \left( f \right) -p\bar{v}_\ell \left( f
    \right), h \right\rangle_\R = \sup_{h\in \C^{2s} : \left\| h \right\|_2 = 1} \sum_{j = -
    2 M}^{2 M} \langle Y_j^{\ell}, h \rangle_\R
  \end{align*}
  and
\begin{align*}
\tilde{h} ( Y_j^{\ell} ) = \langle Y_j^{\ell}, h
  \rangle_\R = \operatorname*{Re}\Big(\sum_{k = 1}^{2s} h_k^*Y_{j,k}^\ell\Big).
\end{align*}
To compute the quantities necessary to apply Lemma \ref{lm:talagrand}, we will extensively use the following elementary bounds:
\begin{align*}
\|g_M\|_\infty &\leq 1, \\
\Big|\frac{2\pi j}{\sqrt{\bar{K}_M''(0)}}\Big| &\leq 4 \tmop{when} M \geq 2,\\
    \|e(j)\|_2^2 &\leq s \left( 1 +
    \max_{\left| j \right| \leq 2 M} \frac{\left( 2 \pi j \right)^2}{\left|
    K_M'' \left( 0 \right) \right|} \right) \leq 14 s \text{\ when\ } M \geq 4.
\end{align*}

First, we obtain an upper bound on $|\tilde{h}|$:
  \begin{align*}
    | \tilde{h} ( Y_j^{\ell} ) | & = \Bigg| \Bigg<
    \frac{1}{M} \Bigg( \frac{i 2 \pi j}{\sqrt{\left| \bar{K}_M'' \left( 0
    \right) \right|}} \Bigg)^{\ell} g_M \left( j \right)
   e^{i2\pi f j} e(j) \left( \delta_j - p \right), h \Bigg>_\R \Bigg|\nn\\
   & \leq \frac{1}{M} \Bigg| \frac{i 2 \pi j}{\sqrt{\left| \bar{K}_M'' \left( 0
    \right) \right|}} \Bigg|^{\ell} \|g_M\|_\infty \|e(j)\|_2 \nn\\
    & \leq B_{\ell} := 4^{\ell+1} \frac{\sqrt{s}}{M}.
  \end{align*}
The expected value of $\left\| v_{\ell} \left( f \right)
  -p\bar{v}_\ell \left( f \right) \right\|_2^2$ is upper bounded as
  follows:
  \begin{align}\label{eqn:evnorm2}
    \mathbbm{E} \left\| v_{\ell} \left( f \right) -p\bar{v}_\ell \left(
    f \right) \right\|_2^2 & =  \sum_{j = - 2 M}^{2 M} \mathbbm{E}
    \left\langle Y_j^{\ell}, Y_j^{\ell} \right\rangle_\R + \sum_{j \neq k}
    \mathbbm{E} \left\langle Y_j^{\ell}, Y_k^{\ell} \right\rangle_\R\nn\\
    & =  \sum_{j = - 2 M}^{2 M} \mathbbm{E} \left\langle Y_j^{\ell},
    Y_j^{\ell} \right\rangle_\R\nn\\
    & \leq  \sum_{j = - 2 M}^{2 M} \frac{1}{M^2} \Bigg| \frac{2 \pi
    j}{\sqrt{\left| \bar{K}_M'' \left( 0 \right) \right|}} \Bigg|^{2 \ell}
    g_M^2 \left( j \right) p \left( 1 - p \right) \|e(j)\|_2^2\nn\\
    & \leq  4^{2 \ell + 3}  \frac{p s}{M} \text{\ when\ } M \geq 4.
  \end{align}
Observe that $\bar{V}^{\ell} = V^{\ell} = \|v_\ell(f) - p\bar{v}_\ell(f)\|_2$. We apply Jensen's inequality and combine with \eqref{eqn:evnorm2} to get
  \begin{align*}
    \mathbbm{E} \bar{V}^{\ell} & = \mathbbm{E}V^{\ell} \leq
    \sqrt{\mathbbm{E}V^{\ell 2}} \leq \sqrt{4^{2 \ell + 3} \frac{p s}{M}}\\
    & \leq 2^{2 \ell + 3} \frac{\sqrt{m
    s}}{M}.
  \end{align*}
  Next, we upper bound $\sigma^2$:
  \begin{align*}
   & \ \ \ \mathbbm{E} \tilde{h}^2 ( Y_j^{\ell} )  =  \mathbb{E} \langle
    Y_j^{\ell}, h \rangle_\R^2 \nn\\
    &\leq \frac{1}{M^2} 4^{2 \ell}
    \left\| g_M \right\|_{\infty} \mathbb{E}(\delta_j - p)^2
    \left|\left<\sqrt{g_M(j)} e(j), h\right>\right|^2 \end{align*}
  implying
  \begin{align*}
    \sum_j \mathbbm{E}\tilde{h}^2 ( Y_j^\ell ) & \leq \frac{1}{M^2} 4^{2 \ell}
     p \sum_{j = -2M}^{2M} \left|h^*\sqrt{g_M(j)}e(j)\right|^2\nn\\
    & = \frac{1}{M^2} 4^{2 \ell} p \left\| h^*P\right\|_2^2\nn\\
    & \leq 4^{2 \ell} \frac{p \left\| P \right\|^2}{M^2}
  \end{align*}
  where $P$ is a matrix in $\mathbb{C}^{2s\times (4M+1)}$ whose $j$th column is $\sqrt{g_M(j)}e(j)$.
 Note that
  \begin{align*}
    \frac{P P^*}{M} & = \frac{1}{M} \sum_{j=-2M}^{2M} g_M(j) e(j)e(j)^* = \bar{D}.
  \end{align*}
  Therefore, we have
  \begin{align*}
    \sigma_{\ell}^2 = \sum_j \mathbbm{E}\tilde{h}^2 \left( Y_j \right) & \leq 4^{2
    \ell} \frac{p }{M^2} \left\| P
    \right\|^2\nn\\
    & \leq 4^{2 \ell} \frac{1}{M^2} p M \left\| \bar{D} \right\| \nn\\
    & \leq 2^{4 \ell + 1} \frac{m}{M^2} 
    \left( \tmop{used} \left\| \bar{D} \right\| \leq 2 \text{\ from \ } \eqref{eqn:Dbarbd}\right)
  \end{align*}
  In conclusion, Lemma \ref{lm:concentration} shows that
  \begin{align*}
    & \ \ \ \mathbbm{P} \left( \left| \left\| v_{\ell} \left( f \right)
    -p\bar{v}_\ell \left( f \right) \right\|_2 -\mathbbm{E} \left\|
    v_{\ell} \left( f \right) -p\bar{v}_\ell \left( f \right) \right\|_2
    \right| > t \right)\nn\\
    & \leq 16 \exp \left( - \frac{t}{K B_{\ell}} \log \left( 1 +
    \frac{B_{\ell} t}{\sigma_{\ell}^2 + B_{\ell} \mathbbm{E}
    {\bar{V}^{\ell}}} \right) \right)\nn\\
    & \leq 16 \exp \left( - \frac{t}{K B_{\ell}} \log \left( 1 +
    \frac{B_{\ell} t}{2^{4 \ell + 1}
    \frac{m}{M^2} + B_{\ell} 2^{2 \ell + 3}
    \frac{\sqrt{m s}}{M}} \right) \right)
  \end{align*}
  Suppose $ B_{\ell} 2^{2 \ell + 3} \frac{\sqrt{m s}}{M} \geq 2^{4 \ell + 1} \frac{m}{M^2}$. Then, define $\bar{\sigma}_\ell^2 = B_{\ell} 2^{2 \ell + 3} \frac{\sqrt{m s}}{M}$ and fix $t = a \bar{\sigma}_{\ell}$. Then
  it follows that
  \begin{align*}
    \mathbbm{P} \left( \left| \left\| v_{\ell} \left( f \right)
    -p\bar{v}_\ell \left( f \right) \right\|_2 -\mathbbm{E} \left\|
    v_{\ell} \left( f \right) -p\bar{v}_\ell \left( f \right) \right\|_2
    \right| > a \bar{\sigma}_{\ell} \right) & \leq 16 e^{- \gamma a^2},
  \end{align*}
  for some $\gamma > 0$ provided $B_{\ell} t \leq \bar{\sigma}_{\ell}^2$. The
  same is true if $2^{4 \ell + 1} \frac{m}{M^2} \geq B_{\ell} 2^{2 \ell + 3} \frac{\sqrt{m s}}{M}$ and $B_{\ell} t \leq 2^{4
  \ell + 1} \frac{m}{M^2}$ if we define $\bar{\sigma}_{\ell}^2 = 2^{4 \ell + 1} \frac{m}{M^2}$. Therefore, let
  \begin{align*}
    \bar{\sigma}_{\ell}^2 & = \max \left\{ 2^{4 \ell + 1}\frac{m}{M^2}, B_{\ell} 2^{2 \ell + 3} \frac{\sqrt{m s}}{M} \right\}\nn\\
    & = 2^{4 \ell + 1} \frac{m}{M^2}
    \max \left\{ 1, 2^4 \frac{s}{\sqrt{m}} \right\},
  \end{align*}
  and fix $a > 0$ obeying 
  \begin{align*}
  a \leq 
  \begin{cases}
  \sqrt{2} m^{1 / 4} & \mbox{\ if \ } 2^4 s / \sqrt{m} \geq 1\\
  \frac{\sqrt{2}}{4} \sqrt{\frac{m}{s}} & \mbox{otherwise}.
  \end{cases}
  \end{align*}
Then we have
  \begin{align*}
    \mathbbm{P} \left( \left\| v_{\ell} \left( f \right)
    -p\bar{v}_\ell \left( f \right) \right\|_2 > 2^{2 \ell + 3}
    \frac{\sqrt{m s}}{M} + a \bar{\sigma}_{\ell}
    \right) & \leq 16 e^{- \gamma a^2}
  \end{align*}
  for some $\gamma > 0$. Application of the union bound proves the lemma.
\end{proof}

\section{Proof of Lemma \ref{lm:I1}}\label{apx:lm:I1}

The proof of Lemma \ref{lm:I1} is based on Hoeffding's inequality presented below:
\begin{lemma}[Hoeffding's inequality]
  \label{lm:complex:hoeffding} Let the components of $u\in \C^n$ be sampled i.i.d. from a symmetric distribution on the complex unit circle, $w \in \C^n$, and $t$ be a positive real number. Then
  \begin{align*}
    \mathbbm{P} \left( \left| \left<u, w\right> \right| \geq t \right) & \leq 4 e^{-
    \frac{t^2}{4 \left\| w \right\|_2^2}}.
  \end{align*}
\end{lemma}

\begin{proof}[Proof of Lemma \ref{lm:I1}]
  Consider the random inner product $\left<u, L^*(v_\ell(f) - p\bar{v}_\ell(f))\right>$ where $\{u_j\}$ are i.i.d. symmetric random variables with values on the complex unit circle. Conditioned on a particular realization
\begin{align*}
\omega \in \mathcal{E} := \Big\{\omega: \sup_{f_d \in \Omega_{\mathrm{grid}}} \left \| L^*(v_{\ell} \left( f_d \right)
    -p\bar{v}_\ell \left( f_d \right))  \right\|_2 < \lambda_\ell, \ell = 0, 1, 2 , 3 \Big\},
\end{align*}    
Hoeffding's
  inequality and the union bound then imply
\begin{align*}
\mathbbm{P} \Big(\sup_{f_d \in \Omega_{\mathrm{grid}}} \left|\left<u, L^*(v_\ell(f_d) - p\bar{v}_\ell(f_d))\right> \right| > \varepsilon \Big | \omega
    \Big) \leq 4 \left| \Omega_{\mathrm{grid}} \right| e^{- \frac{\varepsilon^2}{4
    \lambda_\ell^2}}.
\end{align*}
Elementary probability calculation shows 
\begin{align*}
&\mathbbm{P} \Big(\sup_{f_d \in \Omega_{\mathrm{grid}}} \left|\left<u, L^*(v_\ell(f_d) - p\bar{v}_\ell(f_d))\right> \right| > \varepsilon
    \Big) \nn\\
& = 4 |\Omega_{\mathrm{grid}}| e^{-\frac{\varepsilon^2}{4\lambda_\ell^2}} + \mathbb{P}\left(\mathcal{E}^c\right).
\end{align*}    
Setting 
\begin{align*}
\lambda_\ell = 4 \left( 2^{2\ell + 3}\sqrt{\frac{s}{m}} + \frac{M}{m} a
    \bar{\sigma}_\ell \right)
\end{align*}
 in $\mathcal{E}$ and applying Lemma \ref{lm:vsd} yield,
  \begin{align*}
    & \mathbbm{P} \Big(\sup_{f_d \in \Omega_{\mathrm{grid}}} \left|\left<u, L^*(v_\ell(f_d) - p\bar{v}_\ell(f_d))\right> \right| > \varepsilon \Big | \omega
    \Big)\nn\\
    & \leq 4 \left| \Omega_{\mathrm{grid}} \right| e^{- \frac{\varepsilon^2}{4
    \lambda_\ell^2}} + 64 \left| \Omega_{\mathrm{grid}} \right| e^{- \gamma a^2} +\mathbbm{P}
    \left( \mathcal{E}_{1, \tau}^c \right)
  \end{align*}
  For the second term to be less than $\delta$, we choose $a$ such that
  \begin{align*}
    a^2 & =  \gamma^{- 1} \log  \frac{64 \left| \Omega_{\mathrm{grid}} \right|}{\delta}
    ,
  \end{align*}
  and assume this value from now on. The first term is less than $\delta$ if
  \begin{align}
    \frac{1}{\lambda_\ell^2} & \geq \frac{4}{\varepsilon^2} \log \frac{4 \left|
    \Omega_{\mathrm{grid}} \right|}{\delta} .  \label{eqn:1lambda}
  \end{align}
  First assume that $2^4 s /
  \sqrt{m} \geq 1$. The condition in Lemma \ref{lm:vw} is $a \leq \sqrt{2}
  m^{1 / 4}$ or equivalently
  \begin{align}\label{eqn:mbd1}
    m & \geq \frac{1}{4} \gamma^{- 2} \log^2 \frac{64 \left| \Omega_{\mathrm{grid}}
    \right|}{\delta}.
  \end{align}
  In this case, we have $a \bar{\sigma}_\ell \leq 2^{2\ell+3} \frac{\sqrt{m s}}{M}$, leading to
  \begin{align*}
    \frac{1}{\lambda_\ell^2} & =  \frac{1}{16 \left( 2^{2\ell + 3} \sqrt{\frac{s}{m}} + \frac{M}{m} a \bar{\sigma}_\ell
    \right)^2} \geq \frac{1}{4^{2\ell + 5} }
    \frac{m}{s} .
  \end{align*}
  Now suppose that $2^4 s / \sqrt{m} \leq 1$. If $32s \geq a^2$, then $a \bar{\sigma}_\ell
  \leq 2^{2\ell + 3} \frac{\sqrt{m s}}{M}$ which again gives
  the above lower bound on $1 / \lambda_\ell^2$. On the other hand if $32 s \leq a^2$, then $\lambda_\ell \leq 2^{2\ell +3 }\sqrt{2}  \frac{a}{\sqrt{m}}$ and
  \begin{align*}
    \frac{1}{\lambda_\ell^2} & \geq \frac{1}{2^{4\ell +7}} 
    \frac{m}{a^2}
  \end{align*}
  Therefore, to verify (\ref{eqn:1lambda}) it suffices to take $m$ obeying \eqref{eqn:mbd1} and
  \begin{align*}
    m \min \left( \frac{1}{4^{2\ell + 5} } \frac{1}{s},
    \frac{1}{2^{4\ell + 7}} \frac{1}{a^2} \right) & \geq
    \frac{4}{\varepsilon^2} \log \frac{4 \left| \Omega_{\mathrm{grid}} \right|}{\delta}
  \end{align*}
  This analysis shows that the first term is less than $\delta$ if
  \begin{align*}
    m & \geq \max\Big\{ \frac{4}{\varepsilon^2} 4^{2\ell + 5} s \log \frac{4|\Omega_{\mathrm{grid}}|}{\delta}, \frac{4}{\varepsilon^2} 2^{4\ell + 7} \gamma^{-1} \log \frac{64 |\Omega_{\mathrm{grid}}|}{\delta} \log \frac{4|\Omega_{\mathrm{grid}}|}{\delta},\nn\\
    &\ \ \ \ \ \ \ \ \ \ \ \ \ \frac{1}{4} \gamma^{- 2}  \log^2 \frac{64 \left| \Omega_{\mathrm{grid}}
    \right|}{\delta}\Big \}.
  \end{align*}
  According to Lemma \ref{lm:invertibility}, the last term is less than
  $\delta$ if
  \begin{align*}
    m & \geq \frac{50}{\tau^2} s \log \frac{2
    s}{\delta} .
  \end{align*}
  Setting $\tau = 1 / 4$, combining all lower bounds on $m$ together, and absorbing all constants into one, we get
  \begin{align*}
    m & \geq C  \max \left\{ \frac{1}{\varepsilon^2} \max\left( s\log\frac{|\Omega_{\mathrm{grid}}|}{\delta}, \log^2 \frac{\left| \Omega_{\mathrm{grid}} \right|}{\delta} \right), s \log \frac{s}{\delta} \right\},
  \end{align*}
  is sufficient to guarantee
  \begin{align*}
    \sup_{f_d \in \Omega_{\mathrm{grid}}} \left| I_1^\ell \left( f_d \right) \right| & \leq
    \varepsilon
  \end{align*}
  with probability at least $1 - 3 \delta$. The union bound then proves the lemma.
\end{proof}

\section{Proof of Lemma \ref{lm:I2}}\label{apx:lm:I2}
\begin{proof}[Proof of Lemma \ref{lm:I2}]
  Recall that
  \begin{align*}
   {I_2^\ell \left( f \right)} & = \left<u, (L - p^{-1} \bar{L})^* p \bar{v}_\ell \left( f \right)\right>
  \end{align*}
  On the set $\mathcal{E}_{1, \tau}$ defined in \eqref{eqn:dfE1}, we established in Corollary \ref{cor:bdonmatrix} that
  \begin{align*}
    \left\| L - p^{-1}\bar{L}\right\| & \leq 2 \left\| \bar{D}^{- 1} \right\|^2 p^{- 1} \tau.
  \end{align*}
  We use the $\ell_1$ norm to bound the $\ell_2$ norm of $p \bar{v}_\ell(f)$:
\begin{align*}
    \left\| p \bar{v}_{\ell} \left( f \right) \right\|_2 & \leq     \left\| p \bar{v}_{\ell} \left( f \right) \right\|_1\nn\\
     & = p \Bigg(
    \sum_{k = 1}^s \frac{1}{\sqrt{\left| \bar{K}_M'' \left( 0 \right)
    \right|}^{\ell}} \left| \bar{K}_M^{\left( \ell \right)} \left( f
    - f_k \right) \right| + \sum_{k = 1}^s \frac{1}{\sqrt{\left| \bar{K}_M'' \left( 0 \right)
    \right|}^{(\ell + 1)}} \left| \bar{K}_M^{\left( \ell + 1\right)} \left( f
    - f_k \right) \right|  \Bigg).
\end{align*}
To get a uniform bound on $\sum_{k = 1}^s \frac{1}{\sqrt{\left|
  \bar{K}_M'' \left( 0 \right) \right|}^{\ell}} \left|
  \bar{K}_M^{\left( \ell \right)} \left( f - f_k \right) \right|$, we need the following bound:
  \begin{align*}\label{eqn:kbd}
    \frac{1}{\sqrt{\left| \bar{K}_M'' \left( 0 \right) \right|}^{\ell}} \left|
    \bar{K}_M^{\left( \ell \right)} \left( f \right) \right|  & \leq \begin{cases}
    C_1 & \forall f \in [-\frac{1}{2}, \frac{1}{2}],\\
    \frac{C_2}{M^4 |f|^4} & \mbox{\ if \ } \frac{1}{4M} \leq |f| \leq \frac{1}{2}.
    \end{cases}
  \end{align*}
  for suitably chosen numerical constant $C_1$ and $C_2$. The bound over the region $[\frac{1}{4M}, \frac{1}{2}]$ is a consequence of the more accurate bound established in \cite[Lemma 2.6]{Candes:2012uf}, while the uniform bound $C_1$ can be obtained by checking the expression of $\bar{K}^{(\ell)}_M(f)$. Consequently, we have
  \begin{align*}
 &\sum_{k = 1}^s \frac{1}{\sqrt{\left|
  \bar{K}_M'' \left( 0 \right) \right|}^{\ell}} \left|
  \bar{K}_M^{\left( \ell \right)} \left( f - f_k \right) \right|\nn\\
 &\leq \sum_{k: |f-f_k| < \frac{2}{M}} C_1 + \sum_{k: \frac{2}{M} \leq |f-f_k| \leq \frac{1}{2}}\frac{C_2}{M^4 |f-f_k|^4}\nn\\
 & \leq 4 C_1 + C_2 \sum_{k = 1}^\infty\frac{1}{M^4 (k\Delta_{\min})^4}\nn\\
 &\leq 4C_1 + C_2\sum_{k=1}^\infty \frac{1}{k^4}\nn\\
 & = C := 4C_1 + \frac{\pi^4}{90} C_2 .
  \end{align*}
  We conclude that on the set $\mathcal{E}_{1, \tau}$
  \begin{align*}
    \|(L - p^{-1} \bar{L})^* p \bar{v}_\ell \left( f \right)\|_2
    & \leq C \tau.
  \end{align*}
  Again, application of Hoeffding's inequality and the union bound gives
  \begin{align*}
    & \mathbbm{P} \left( \sup_{f_d \in \Omega_{\mathrm{grid}}} \left| I_2^\ell \left( f_d
    \right) \right| > \varepsilon \right)\nn\\
    & \leq 4 \left| \Omega_{\mathrm{grid}} \right| \exp \left( - \frac{\varepsilon^2}{4 C
    \tau^2} \right) +\mathbbm{P} \left( \mathcal{E}_{1, \tau}^c \right) .
  \end{align*}
  To make the first term less than $\delta$, it suffices to take
  \begin{align*}
    \tau^2 & = \frac{\varepsilon^2}{4 C \log \frac{4 \left| \Omega_{\mathrm{grid}}
    \right|}{\delta}} .
  \end{align*}
  To have the second term less than $\delta$, we require
  \begin{align*}
    m & \geq \frac{C}{\tau^2} s \log \frac{2 s}{\delta}\nn\\
    & = \frac{C}{\frac{\varepsilon^2}{\log \frac{4 \left| \Omega_{\mathrm{grid}}
    \right|}{\delta}}} s \log \frac{2 s}{\delta}\nn\\
    & = C \frac{1}{\varepsilon^2} s \log \frac{2 s}{\delta} \log \frac{4
    \left| \Omega_{\mathrm{grid}} \right|}{\delta}.
  \end{align*}
  Another application of the union bound with respect to $\ell = 0, 1, 2, 3$ proves the lemma.
\end{proof}

\end{document}